\documentclass[a4paper,11pt]{article}
\pdfoutput=1
\usepackage{jheppub}
\usepackage[T1]{fontenc}
\usepackage[english]{babel}
\usepackage{amsthm,amsfonts}
\hyphenchar\font=\string"7F
\usepackage{tensor}
\usepackage[mathstyleoff]{breqn}
\usepackage{slashed}
\usepackage{stmaryrd}
\usepackage{enumerate}
 \usepackage[normalem]{ulem}
\usepackage{tikz}
\usepackage{tensind}
\usepackage[mathscr]{euscript}

\DeclareMathAlphabet{\mathdsl}{U}{bbm}{m}{sl}
\usepackage{arydshln}
\usepackage{colortbl}
\newcommand\Tstrut{\rule{0pt}{2.9ex}}         

\DeclareMathAlphabet{\mathbbmsl}{U}{bbm}{m}{sl}
\makeatletter
\DeclareFontFamily{OMX}{MnSymbolE}{}
\DeclareSymbolFont{MnLargeSymbols}{OMX}{MnSymbolE}{m}{n}
\SetSymbolFont{MnLargeSymbols}{bold}{OMX}{MnSymbolE}{b}{n}
\DeclareFontShape{OMX}{MnSymbolE}{m}{n}{
    <-6>  MnSymbolE5
   <6-7>  MnSymbolE6
   <7-8>  MnSymbolE7
   <8-9>  MnSymbolE8
   <9-10> MnSymbolE9
  <10-12> MnSymbolE10
  <12->   MnSymbolE12
}{}
\DeclareFontShape{OMX}{MnSymbolE}{b}{n}{
    <-6>  MnSymbolE-Bold5
   <6-7>  MnSymbolE-Bold6
   <7-8>  MnSymbolE-Bold7
   <8-9>  MnSymbolE-Bold8
   <9-10> MnSymbolE-Bold9
  <10-12> MnSymbolE-Bold10
  <12->   MnSymbolE-Bold12
}{}

\let\llangle\@undefined
\let\rrangle\@undefined
\DeclareMathDelimiter{\llangle}{\mathopen}%
                     {MnLargeSymbols}{'164}{MnLargeSymbols}{'164}
\DeclareMathDelimiter{\rrangle}{\mathclose}%
                     {MnLargeSymbols}{'171}{MnLargeSymbols}{'171}
\makeatother

\usepackage{dsfont}
\usepackage{bbm}

\DeclareMathOperator{\Tr}{Tr}
\DeclareMathOperator{\VEC}{vec}
\DeclareMathOperator{\diag}{diag}
\newtheorem{theorem}{Theorem}
\newcommand{\DFTwzw}{DFT${}_\mathrm{WZW}$}
\allowdisplaybreaks
\newcommand{\HIDDEN}[1]{}

\makeatletter
\def\widebreve{\mathpalette\wide@breve}
\def\wide@breve#1#2{\sbox\z@{$#1#2$}%
     \mathop{\vbox{\m@th\ialign{##\crcr
\kern0.08em\brevefill#1{0.8\wd\z@}\crcr\noalign{\nointerlineskip}%
                    $\hss#1#2\hss$\crcr}}}\limits}
\def\brevefill#1#2{$\m@th\sbox\tw@{$#1($}%
  \hss\resizebox{#2}{\wd\tw@}{\rotatebox[origin=c]{90}{\upshape(}}\hss$}
\makeatletter

\title{\boldmath  Doubled aspects of generalised dualities and\\integrable deformations }
\preprint{}

\author[a,b]{Saskia Demulder,}
\author[c,d,e]{Falk Hassler,}
\author[a,b]{and Daniel C. Thompson}

\emailAdd{Saskia.Demulder@vub.be}
\emailAdd{fhassler@unc.edu}
\emailAdd{D.C.Thompson@Swansea.ac.uk}

\affiliation[a]{
Theoretische Natuurkunde, Vrije Universiteit Brussel \& The International Solvay Institutes,\\ B-1050 Brussels, Belgium}
\affiliation[b]{Department of Physics, Swansea University, Swansea, SA2 8PP, U.K.}
\affiliation[c]{Department of Physics and Astronomy, University of North Carolina, \\ Chapel Hill, NC 27599-3255, USA}
\affiliation[d]{Department of Physics and Astronomy, University of Pennsylvania, Philadelphia, PA 19104, USA}
\affiliation[e]{Department of Physics, University of Oviedo, Avda. Federico Garc\'ia Lorca 18, 33007 Oviedo, Spain.}

\abstract{The worldsheet theories that describe Poisson-Lie T-dualisable $\sigma$-models on group manifolds as well as integrable $\eta$, $\lambda$ and $\beta$-deformations provide examples of ${\cal E}$-models. Here we show how such ${\cal E}$-models can be given an elegant target space description within Double Field Theory  by specifying explicitly generalised frame fields forming an algebra under the generalised Lie derivative.  With this framework we can extract simple criteria for the R/R fields and the dilaton that extend the ${\cal E}$-model conditions to type II backgrounds.  In particular this gives conditions for a type II background to be Poisson-Lie T-dualisable. Our approach gives rise to algebraic field equations for Poisson-Lie symmetric spacetimes and provides an effective tool for their study.}

\begin{document}

\maketitle

\section{Introduction}
Abelian T-duality, the invariance of string theory when the radius of an $S^1$ in target space is inverted, has long served as a catalyst for theoretical developments.  Given its prominent role, a long standing challenge has been to establish T-duality in more general contexts for instance when the target space admits a non-Abelian group of isometries \cite{delaOssa:1992vci}.  Two decades ago, in a remarkable sequence of works \cite{Klimcik:1995ux,Klimcik:1995dy} by Klim\v{c}\'ik and \v{S}evera it was shown that in special circumstances one may even relax the imposition of an isometry in target space and still retain a notion of T-duality called Poisson-Lie (PL) duality. 

Some caution should be exercised here; whilst the maps between non-linear sigma-models induced by non-Abelian \cite{Lozano:1995jx} or more generally PL T-dualities \cite{Klimcik:1995ux,Klimcik:1995dy,Sfetsos:1996xj,Sfetsos:1997pi} are canonical transformations of the classical phase space, it is hard in general to establish them as fully fledged quantum equivalences. Indeed in a generic context one should not expect to have control of either $\alpha^\prime$ or   $g_s$ effects. Optimistically one might suggest that these ``dualities'' constitute a map from a CFT to a new CFT$^\prime$ for which modular invariance may necessitate the inclusion of extra twisted sectors. The partition sums of these theories need not match. This viewpoint dates back to \cite{Rocek:1991ps} and was recently shown to be the case in a simple $SU(2)$ non-Abelian T-dualisation \cite{Fraser:2018cxn}.

Nonetheless these generalised ``dualities'' (and henceforth we drop the quotation marks) retain utility as  solution generating techniques within supergravity and continue then to hold interest for their potential application to holography.  Non-Abelian T-duality for instance can be used to construct novel examples of holographic spacetimes  
(see e.g. \cite{Sfetsos:2010uq,Lozano:2011kb,Itsios:2013wd} for early works in this direction and \cite{Lozano:2016kum} for the field theoretic interpretation).   Poisson-Lie T-duality at first sight appears to concern rather complicated looking spacetimes. However this complexity can in some cases be illusory. In fact a class of integrable models, known variously as $\eta$-deformations or Yang-Baxter sigma-models, introduced by Klim\v{c}\'ik \cite{Klimcik:2002zj} some years ago constitute exemplars of PL T-dualisable theories. A significant amount of activity has followed from the introduction of the integrable $\eta$-deformation of the AdS$_5 \times$S$^5$ spacetime \cite{Delduc:2014kha,Delduc:2013qra}. A further development has been integrable $\lambda$-deformations \cite{Sfetsos:2013wia,Hollowood:2014qma} of (potentially gauged) WZW-models which are related to $\eta$-type deformations\cite{Hoare:2015gda,Sfetsos:2015nya,Klimcik:2016rov} via a Poisson-Lie duality transformation combined with an analytic continuation of certain Euler angles and couplings. 

From the worldsheet perspective such generalised dualities can be rendered manifest in a {\em doubled formalism} much like that of Abelian T-duality introduced in \cite{Tseytlin:1990nb,Tseytlin:1990va}.  In these approaches one considers a   sigma-model whose target space has double the number of dimensions. Half of the coordinates describing this doubled space can be eliminated  to recover a standard sigma-model.\footnote{This reduction requires the imposition of a chirality constraint which is a delicate matter quantum mechanically and in  \cite{Tseytlin:1990nb,Tseytlin:1990va}  it is achieved at the expense of manifest  Lorentz invariance, other alternatives based on gauging e.g.  \cite{Hull:2006va,Lee:2013hma,Hatsuda:2018tcx}  may prove more amenable to a quantum treatment.} When this reduction can be done in multiple ways we recover T-dual related descriptions.   

This philosophy was   extended to generalised T-dualities in the original works  \cite{Klimcik:1995ux,Klimcik:1995dy} as well as in \cite{Hull:2009sg}. Other recent interesting works in this direction include \cite{Lust:2018jsx,Marotta:2018swj}.  The doubled space is equipped with the familiar generalised metric and $O(D,D)$ invariant inner-product and is further required to be a group manifold, ${\mathdsl D}$,and so comes equipped with a canonical three-form. A useful presentation of the doubled worldsheet is provided by the first order formalism  now coined ${\cal E}$-models \cite{Klimcik:2015gba},  first introduced in \cite{Klimcik:1995dy} and  developed in \cite{Stern:1998my}.    PL dualisable sigma models, as well as $\eta$- and $\lambda$- and $
\beta$- deformations are all examples of theories that can be extracted from ${\cal E}$-models.

Abelian T-duality has an elegant target space duality symmetric formulation known as Double Field Theory (DFT) \cite{Hull:2009mi}. Since a worldsheet doubled formalism is available for generalised T-dualities one would hope for a similar understanding at the level of the target space. The first clues here come from studying the one-loop $\beta$-functions of the worldsheet theory \cite{Avramis:2009xi,Sfetsos:2009vt}. In \cite{Avramis:2009xi} it was pointed out that the $\beta$-function for the generalised metric corresponds to the scalar equation of motion of a gauged supergravity with the structure constants of the doubled target space providing the embedding tensor. In DFT the way such gauged supergravities arise is by performing a Scherk-Schwarz reduction \cite{Scherk:1978ta,Scherk:1979zr,Grana:2012rr,Geissbuhler:2011mx,Aldazabal:2011nj,Berman:2012uy,Musaev:2013rq}. Thus what one requires is a precise formulation of DFT on the group manifold ${\mathdsl D}$. \DFTwzw{} \cite{Blumenhagen:2014gva,Blumenhagen:2015zma,Bosque:2015jda} provides exactly such an approach  and the study of its relation to Poisson-Lie T-duality was initiated \cite{Hassler:2017yza}.

This manuscript will continue the development of generalised T-dualities and integrable deformations within DFT. Specifically we will show how the type II extension of \DFTwzw{}  provides an immediate set of criteria 
that extends the structure of ${\cal E}$-models to the R/R and dilaton sector.  In the case where this ${\cal E}$-model describes a PL T-dualisable NS sector, this gives rise to criteria  that must be obeyed for a full type II supergravity background to be PL T-dualisable.   We shall describe backgrounds for which these criteria   hold as being PL symmetric. \DFTwzw{} makes this symmetry manifest and thereby significantly simplifies their analysis. For example, instead of having to cope with difficult, coupled PDEs, the field equations become algebraic. 
 
 A pivotal element in our discussion will be a generalised frame field on the spacetime which allows us to connect the fields on the doubled space with the conventional type II supergravity fields.  In this work we will follow a  technique suggested in \cite{duBosque:2017dfc} to construct a set of $O(D,D)$ valued generalised frame fields that furnish the algebra of $\mathdsl D$ via the generalised Lie derivative.  In the cases we are most interested in, and that includes $\eta$-, $\lambda$- and $\beta$-deformations as well as all PL dualisable models, this construction is carried out explicitly making use of the group theoretic quantities on  $\mathdsl D$ and its coset ${\cal M}= \mathdsl D/\widetilde H$  by a maximal isotropic subgroup $\widetilde H$. Our discussion will be predominantly local in nature, however where this construction can be extended globally this provides an understanding of ${\cal E}$-models as examples of generalised parallelizable spaces \cite{Grana:2008yw,Lee:2014mla}.  
 
The $\eta$- and some $\beta$-deformations are governed by modified type II field equations \cite{Arutyunov:2015mqj,Hoare:2014pna,Borsato:2016ose}.   Modified type II requires a Killing vector, $I$, and a one form, $Z$, in addition to the bosonic field content known from standard type II supergravity. Connections to DFT and ExFT of modified supergravity are discussed in \cite{Sakatani:2016fvh,Baguet:2016prz}, here we show that they also arise from \DFTwzw{} if the subgroup $\widetilde H$ is non-unimodular.  In \cite{Araujo:2017enj,Bakhmatov:2018apn,Araujo:2018rbc} an open string interpretation of such $I$ modified supergravities and integrable deformations was given. We will illustrate these ideas with a number of specific examples. They emphasis how exploiting PL symmetry can make challenging calculations in integral deformations much easier and vindicate the combination of DFT techniques and integrable deformations.

The paper is organised as follows: in Section \ref{sec:Emodels} we review $\mathcal E$-models and how the $\lambda$- and $\eta$-deformations fit into this framework. In Section \ref{sec:TSEmodels}, we develop further the  implementation of Poisson-Lie T-duality in Double Field Theory \cite{Hassler:2017yza} to show how the R/R-sector of (modified-) SUGRA can be elegantly   extracted. After a short reminder of DFT$_\mathrm{WZW}$, we prove how for a group admitting a maximally isotropic subgroup $\widetilde H$ together with some additional conditions there exists a generalised frame field solving the section condition of DFT. Then in Section \ref{sec:RRsector} we discuss how the DFT manifest implementation of Poisson-Lie T-duality can be extended to the R/R-sector and dilaton. In the last Section \ref{sec:intdef} we apply the formalism to integrable deformations and provide some explicit examples of how the R/R-sector can be extracted for e.g. $AdS_3\times S^3$ backgrounds.  The goal of section \ref{sec:intdef}  there is not to present novel backgrounds but to demonstrate the efficacy of the approach proposed in the paper.  
We conclude with a brief discussion of some of the outstanding challenges as we see them. The presentation is complemented with a number of technical appendices. 
\vskip 0.5 cm
{\bf Note added:} Whilst this manuscript was in its very final stages of preparation we received a  preprint  \cite{Severa:2018pag}  from  {\tt math.DG} that overlaps with some of the conclusions of this paper, albeit cast in the language of Courant Algebroids rather than DFT$_\mathrm{WZW}$.

\section{${\cal E}$-models:  Poisson-Lie Duality and Integrable Theories} \label{sec:Emodels}
To make this article self-contained, let us begin by reviewing the basic features of ${\cal E}$-models before describing the specialisation to  Poisson-Lie T-duality and integrable deformations.

 Our starting point\footnote{A guide to notation, conventions and some algebra terminology can be found in Appendix \ref{App:Conventions}.}  is a real Lie algebra $\mathfrak{d}$ of even dimension, $\dim \mathfrak{d} =  2 D$, equipped with non-degenerate, ad-invariant, symmetric inner-product $\llangle \cdot, \cdot \rrangle $  that we assume to be of split-signature. Letting $\mathdsl T_A$ be a basis of generators for $\mathfrak d$, we shall write 
\begin{equation}\label{eq:Dcomms}
  [\mathdsl T_A ,\mathdsl T_B] =\mathdsl F_{AB}{}^C \mathdsl T_C \ , \quad  \llangle\mathdsl T_A,\mathdsl T_B \rrangle = \eta_{AB}  \, . 
\end{equation}
We  denote the components of the matrix inverse of $\eta_{AB}$ as $\eta^{AB}$ and we will raise and lower indices with this. 

The ${\cal E}$-model is a dynamical system 
that can be conveniently parametrised by a set of algebra-valued maps $j= j^A \mathdsl T_A :S_\sigma^1 \to \frak{d}$  
obeying the classical current algebra 
\begin{equation}
  \{ j_A(\sigma), j_B(\sigma') \}_{P.B.} = {\mathdsl F}_{AB}{}^C j_C(\sigma) \delta( \sigma - \sigma' ) +
    \eta_{AB}   \delta^\prime( \sigma - \sigma' ) \, ,
\end{equation} 
 with   dynamics determined by the Hamiltonian
\begin{equation}
  \textrm{Ham} = \frac12 \oint d \sigma \llangle  j(\sigma), {\cal E} (  j(\sigma) ) \rrangle \, .
\end{equation} 
Here ${\cal E}(\mathdsl T_A) = {\cal E}_A{}^B\mathdsl  T_B$, the eponymous operator, is an idempotent involution of $\frak{d}$ that is self-adjoint with respect to  $\llangle \cdot, \cdot \rrangle $. We can parametrise ${\cal E}$ in terms of a {\it generalised metric}, ${\cal H}$, as 
\begin{equation}\label{eq:genmet1}
 {\cal E}_A{}^B =  {\cal H}_{AC}\eta^{C B} \, , \quad {\cal H}_{AB}=  {\cal H}_{BA} \, , \quad  {\cal H}_{AC}\eta^{CD} {\cal H}_{DB} = \eta_{AB} \, .
\end{equation}

We  will be  interested in the case, and assume it henceforth,  that   $\mathfrak{d}$ admits a maximally isotropic subalgebra $\tilde{\frak{ h}}\subset  \mathfrak{d}$.   Then the ${\cal E}$-model can be reduced to a conventional non-linear sigma-model whose target space is the coset ${\cal M} = \mathdsl D/\tilde{H}$ where $\mathdsl D,\tilde{H}$ are the groups corresponding to respectively $\mathfrak{d},\,\tilde{\frak{ h}}$. 
 To fix notation we let $\mathdsl T_A = (\widetilde T^a , T_a)$ where $\widetilde T^a$ are generators of $\tilde{\frak{ h}}$ and $T_a$ are the remaining generators whose span we denote  $\frak{k}$.  In this basis the inner-product can be taken  to be 
\begin{equation}\label{eq:etamat}
  \eta_{AB} = \begin{pmatrix} 0 & \delta^a{}_b \\
    \delta_a{}^b & 0
  \end{pmatrix} \,. 
\end{equation} 
  It is  important to stress that in the decomposition $\frak{d} = \tilde{\frak{ h}} \oplus\frak{k}  $ we place no requirement on $\frak{k}$;   that is to say  in general $ \mathdsl D/\tilde{H}$ is  neither a group manifold itself nor a symmetric space however the examples we shall be most interested in here will indeed be of this type. 

The non-linear sigma-model that follows from the ${\cal E}$-model is  described by an action \cite{Klimcik:2015gba,Klimcik:1995dy}\footnote{Here we restore an overall normalisation $\frac{\tilde{k}}{2\pi}$, which   depending on the specific properties of $ \mathdsl D/\tilde{H}$, may require a quantisation in order to define the WZ term unambiguously in a path integral.}
\begin{align}\label{eq:act}
  S_{\mathdsl D/\tilde{H}} &= \tilde k S_{WZW}[m] - \frac{\tilde{k}}{\pi} \int d\sigma d\tau \llangle  {\cal P}(m^{-1} \partial_+ m) ,  m^{-1} \partial_- m \rrangle  \, , \\ 
  S_{WZW}[m] &= \frac{1}{2 \pi}  \int d\sigma d\tau  \llangle m^{-1} \partial_+ m,  m^{-1} \partial_-  m\rrangle   + \frac{1}{24 \pi} \int_{M_3}  \llangle  m^{-1} d m , [ m^{-1} d m ,m^{-1} d m ]  \rrangle   \, . \end{align}
Here we have  parametrised a  group element on $\mathdsl D$ as $\mathdsl g(X^I) = \tilde h(\tilde{x}_{\tilde i} ) m(x^i)$ where $X^I = (\tilde{x}_{\tilde i} , x^i) $ are  local coordinates on   $\mathdsl D$ such that $\tilde{x}_{\tilde i}$, ${\tilde i} =1\dots D$, are local coordinates on $\tilde{H}\subset \mathdsl D$ and $x^i$, $i=1\dots d$,  are local coordinates that parametrise the coset.  The first term in eq.~\eqref{eq:act} denotes the WZW action on $\mathdsl D$, defined with the inner-product $\llangle \cdot, \cdot \rrangle $, evaluated on the coset representative $m$.  The second term, whose coefficient is $-2$ times that of the kinetic term of the WZW model, is defined with a projector obeying \cite{Klimcik:2015gba}
\begin{equation}
  \textrm{Im}\, {\mathcal P} =\mathfrak{h} \, , \quad  \textrm{Ker}\, {\mathcal P} = (1 +  \textrm{ad}_m \cdot   {\mathcal E} \cdot  \textrm{ad}_{m^{-1}} ) \mathfrak{d} \, .
  \end{equation} 
  
  \subsection{Poisson-Lie Models}
Let us now discuss the special case where $\frak{d}$ is a Drinfel'd double i.e. $\frak{d} =   \tilde{\mathfrak{h}}\oplus\mathfrak{h}$ with \underline{both} $\tilde{\mathfrak{h}}$,$\mathfrak{h}$ maximally isotropic subalgebras. This is the setting of Poisson-Lie T-duality. In this case we can identify the coset with the Lie group manifold $\mathdsl D/\widetilde H \cong  H$ and so in the action eq.~\eqref{eq:act} the representative $m(x)$ can be considered  an element of the group $ H$. Since $m^{-1} d m$ is valued in $\mathfrak{h}$, which is an isotropic with respect to $\llangle \cdot, \cdot \rrangle$, the WZW part of the action eq.~\eqref{eq:act} is identically zero and what remains can be cast in the form of a sigma-model: 
\begin{equation}\label{eq:PLact}
  S_{\mathdsl D/\widetilde H} = \frac{1}{\pi {\mathrm{s}} }  \int d^2 \sigma\, e_+^a \left( E_0^{-1} + \Pi \right)^{-1}_{ab}   e_-^b =    \frac{1}{\pi {\mathrm{s}} }   \int d^2 \sigma\,  ( G(x)- B(x) )_{i j } \partial_+ x^i \partial_- x^j \, , 
\end{equation}
in which $m^{-1} \partial_\pm m = e_\pm^a T_a = e^a{}_i \partial_{\pm }x^i T_a$ are the light cone components of the left-invariant one-forms   pulled back to the worldsheet and the normalisation is $\mathrm{s} =  \tilde{k}^{-1} $. Later we shall also require the right-invariant one-forms $\partial_\pm m m^{-1} = v_\pm^a T_a = v^a{}_i \partial_{\pm }x^i T_a$ together with the vector fields $e_a=e^i{}_a \partial_i$ and  $v_a=v^i{}_a \partial_i$ that generate respectively right and left actions. The matrix $\Pi$ is derived from the adjoint action:
\begin{equation}
  \textrm{ad}_m   \mathdsl T_A = m \mathdsl T_A m^{-1} = M_A{}^B \mathdsl  T_B \ , \quad \Pi^{ab} = M^{ac} M^b{}_c   \, .
\end{equation}
The $D^2$ constant parameters in $E_0 = G_0-  B_0$ are related to those of the generalised metric introduced in eq.~\eqref{eq:genmet1} in the standard way
\begin{equation}
 {\cal H}_{AB} =    \begin{pmatrix} G_0^{-1}   & - G_0^{-1}  B_0 \\
   B_0 G_0^{-1}  & G_0-  B_0 G_0^{-1}B_0  \end{pmatrix}  
    \, . 
\end{equation}
In general the target space metric corresponding to the sigma model eq.~\eqref{eq:PLact} is unappetising and lacking isometry however it has a rather special algebraic structure.  Although the  the currents $J_a $  corresponding to left action on $H$ are not   conserved in the usual sense  they do obey a  non-commutative conservation law 
\begin{equation}
\partial_+ J_{a -   } + \partial_- J_{a +   }  =  \widetilde{F}^{bc}{}_a  J_{  b + }  J_{c -}  \, , 
 \end{equation}
in which we emphasise that the structure constants appearing on the right hand side are those of the     $\tilde{\frak{h}}$.   In terms of the target space data, $E_{ij}= G_{ij}- B_{ij}$, this places a  requirement  that  
\begin{equation}\label{eq:PLcondition}
L_{V_a} E_{ij }  = - \widetilde{F}^{bc}{}_a e^k{}_{b} e^l{}_{c} E_{ i k} E_{l j}  \ . 
 \end{equation} 
 This condition on the target space  is referred to as a Poisson-Lie symmetry.\footnote{Taking a further Lie derivative of this relation invokes an integrability condition, namely that  viewed as a map ${\frak{h}} \to   {\frak{h}} \wedge {\frak{h}} $ the structure constants $\widetilde{F}^{bc}{}_a $ are  required to define a one-cocycle  obeying the co-Jacobi identity.   As explained in the Appendix \ref{app:alg}, this property can be understood as the infinitesimal version of $H$ being a Poisson-Lie group, giving justification for the name.}

 At this stage we make an important observation; using the curved space $G_{ij}$ and $B_{ij}$ that define the sigma model eq.~\eqref{eq:PLact} we may define a coordinate dependent $O(D,D)$  generalised metric 
\begin{align} \label{eq:curvedspacegenmet}
\widehat{\cal H}_{\hat{I}\hat{J} }(x) =   \begin{pmatrix} G^{-1}   & - G^{-1}  B \\
   B G^{-1}  & G -  B G^{-1}B \end{pmatrix}\indices{_{\hat{I}\hat{J}}}\, . 
\end{align} 
A tedious but straightforward calculation reveals that 
\begin{equation}
\widehat{\cal H}_{\hat{I}\hat{J}}(x) =  \widehat{E}^A{}_{\hat{I}} (x) {\cal H}_{AB}  \widehat{E}^B{}_{\hat{J}}  (x) \ , 
\end{equation} 
where 
\begin{equation}
 \widehat{E}^A{}_{\hat{I}} (x)   =  \begin{pmatrix}  \,1 \,   &  \, 0  \,    \\
    \,  \Pi  \,  &  \,1  \,\end{pmatrix}\indices{^A_B}   \begin{pmatrix} e^{-T}  &0    \\
     0   & e  \end{pmatrix}\indices{^B_{\hat I}} =  M^A{}_B    \begin{pmatrix}v^{-T}  &0    \\
     0   & v  \end{pmatrix}\indices{^B_{\hat I}}  \, . 
\end{equation}
  The hats on the indices and frame fields  are introduced to emphasize  dependence only on the coordinates $x^i$ and not on the ``dual'' $\tilde{x}_{\tilde i}$, i.e. $\partial_{\hat{I}} = (0, \partial_i)$ -- in the terminology of DFT we have picked a solution to the section condition.   Notice also that the frame fields are (coordinate dependent) elements of $O(D,D)$.  

Of course we could swap the r\^{o}le of the two subgroups. If instead we parametrize $\mathdsl g(X)= \widetilde{m}(\tilde{x}) h(x)$ we can reduce to a theory on the coset $\widetilde{\cal M} =\mathdsl D/H \cong \widetilde{H}$. In that case we find the Poisson-Lie T-dual theory to eq.~\eqref{eq:PLact} given by an action
 \begin{equation}\label{eq:PLact2}
S_{\mathdsl D/H} =   \frac{1}{\pi {\mathrm{s}} }    \int d^2 \sigma\, \tilde{e}_{+a} \left( E_0  + \widetilde{\Pi} \right)^{-1}{}^{ab}   \tilde{e}_{-b}  =   \frac{1}{\pi {\mathrm{s}} }   \int d^2 \sigma\,  ( \widetilde{G} - \widetilde{B})^{\tilde i\tilde j  } \partial_+  \tilde{x}_{\tilde{i}} \partial_- \tilde{x}_{\tilde{j}} \, , 
\end{equation}
where   $\widetilde{m}^{-1} \partial_\pm \widetilde{m} = \tilde{e}_{ \pm a}  \widetilde T^a = \tilde{e}_{a}{}^{\tilde{i} }  \partial_{\pm }\tilde x_{\tilde{i}} \widetilde T^a   $ and $\widetilde{\Pi}$ defined via the adjoint action.    An important feature is that the two PL sigma-models are related by a canonical transformation \cite{Klimcik:1995ux,Klimcik:1995dy,Sfetsos:1996xj,Sfetsos:1997pi}  at the classical level which  can be derived from a generating functional
\begin{equation}
F= \oint d\sigma \theta(\tilde x, {x} ) \, ,
\end{equation}
in which  $\theta$ is the pull back of a one form to $S^1_\sigma$ whose form is known only implicitly.  However an elegant expression can be given for its derivative \cite{Sfetsos:1997pi}
\begin{equation}
\begin{aligned}
\omega= d\theta =2(  {\cal O}^{-1})_a{}^b e^a \wedge \tilde{e}_b + (  {\cal O}^{-1} \widetilde{\Pi} )_{ab} e^a \wedge e^b  - (  {\cal O}^{-T} \Pi)^{ab} \tilde{e}_a \wedge \tilde{e}_b   \ , \quad  
 {\cal O} = \textrm{id} - \widetilde{\Pi} \Pi \, . 
\end{aligned}
\end{equation}

\subsection{Integrable deformations} 
An application of ${\cal E}$-models is to provide a universal description of two superficial distinct classes of integrable deformations known as $\eta$- and $\lambda$-deformed theories\cite{Klimcik:2015gba}.  Let us review some salient features of these deformations which we shall return to in some detail later.

\subsubsection{$\eta$-deformation} 
In its simplest form the $\eta$-model is a deformation of the principal chiral model on a group manifold $G$ defined in terms of an operator ${\cal R}$,  an  endomorphism of $\frak{g}$ obeying the modified classical Yang-Baxter equation
\begin{equation}\label{eq:YBE}
[{\cal R} X, {\cal R} Y]- {\cal R} \left([ {\cal R} X, Y]  +[X,  {\cal R} Y]   \right) = - c^2 [X, Y]\, , \quad \forall X,Y \in \frak{g}  \, , 
\end{equation}
where $c^2 \in \{ -1,0,1\} $.    We require that $ {\cal R} $ be skew-symmetric with respect to the Cartan-Killing form   $\langle t_a ,t_b  \rangle= \kappa_{ab} = - \frac{1}{2 h^\vee} f_{ac}{}^d f_{b d}{}^c$ with $[t_a, t_b] =f_{ab}{}^c t_c$ with $\{ t_a \}$ the generators of $\frak{g}$. 

The $\eta$-deformation corresponds to taking the choice $c^2 = -1$ with ${\cal R}$ acting to swap positive and negative roots and as zero on the Cartan and   is defined by the action
\begin{equation}\label{eq:etaact}
 S_\eta =   \frac{1}{ \pi \mathrm{t} }  \int d^2\sigma \langle v_+,   \left(    1 - \eta {\cal R} \right)^{-1} v_- \rangle  \, .  
\end{equation}
This theory is of particular interest since it preserves the integrability \cite{Klimcik:2008eq} of the principal chiral model (at least classically).\footnote{In actuality, integrability of theory in eq.~\eqref{eq:etaact} is ensured for any value of $c^2\in\{-1,0,1\}$ \cite{Matsumoto:2015jja} and  the case of $c=0$ is of relevance in describing e.g. TsT deformations \cite{Hoare:2016wsk,Osten:2016dvf}.} 
  What may not be immediately obvious is that this is an example of a model admitting Poisson-Lie T-duality and thus an ${\cal E}$-model.   Indeed the action \eqref{eq:etaact} can be brought into PL form of eq.~\eqref{eq:PLact} with the identification
\begin{equation}\label{eqn:E0eta}
\left( E_0^{-1}  \right)^{ab} = \frac{\kappa^{ab} }{\eta} - R^{ab}  \, , \quad  \Pi^{ab} = R^{ab}  - D[g]^{ac} R_c{}^d D[g^{-1} ]_d{}^b   \, ,  \quad \mathrm{s}   =  \tilde{k}^{-1}= \mathrm{t} \eta \ ,
\end{equation}  
in which we have defined ${\cal R}(t_a) = R_a{}^b t_b $ and $\textrm{ad}_g t_a = g t_a g^{-1} = D[g]_{a}{}^b t_b$ for $g\in G$, and raised indices with $\kappa^{ab}$. To understand the $\mathcal E$-model corresponding to this sigma-model, one needs to identify the corresponding double $\mathfrak d$ and idempotent operator $\mathcal E$. Note that eq.~\eqref{eq:YBE} ensures that the bracket 
\begin{equation}
[X, Y]_{\cal R}=  [ {\cal R} X, Y]  +[X,  {\cal R} Y]  \, ,
\end{equation}
obeys the Jacobi identity and thus defines a second Lie-algebra we call $\frak{g}_{\cal R}$.   It is a standard result that $\frak{d} = \frak{g} +  \frak{g}_{ \cal R }$ can be identified with the complexification
$\frak{d} =  \frak{g}^{\mathbb{C}}$ which, when viewed as a real Lie algebra with elements $Z= X + i Y$, $X,Y \in  \frak{g} $,    can be equipped with an inner-product given by the imaginary part of the   Cartan-Killing form.    Under the Iwasawa decomposition $\frak{d} =\frak{g}^{\mathbb{C}} =  \frak{g} + (\frak{a}+\frak{n})$   both $\frak{g} $ and $\frak{\tilde{h} }=  \frak{a}+\frak{n}$ are maximal isotropic subalgebras.  Finally the operator $\mathcal E$ is given by \cite{Klimcik:2015gba}
\begin{align*}
\mathcal E: Z\mapsto  \frac{i}{2} \left( \eta-  \eta^{-1}\right)Z- \frac{i}{2} \left( \eta+  \eta^{-1}\right)Z^\dag\,. 
\end{align*}

\subsubsection{$\lambda$-deformations}  \label{sec:lambda}
Appearing  at first sight to be a rather different class of integrable models, $\lambda$-deformations can be thought of as a re-summed marginally relevant current-current perturbation of a WZW-model on a group manifold $G$. The $\lambda$-deformed WZW model is specified by the action \cite{Sfetsos:2013wia}
\begin{equation}\label{eq:actlambda} 
  S_\lambda =k S_{WZW  }[g]+ \frac{k}{\pi} \int d^2\, \sigma \langle\partial_+g g^{-1},  (\lambda^{-1} - \textrm{ad}_{g^{-1}} )^{-1} g^{-1} \partial_- g\rangle  \, ,
\end{equation} 
Here we use the WZW action for a group element $g\in G$ as in eq.~\eqref{eq:act} but with the inner-product simply given by the Cartan-Killing form, $\kappa= \langle \cdot , \cdot\rangle$.   
 In addition to the metric and $B$-field obtained from the above action, the construction of the $\lambda$-deformed theory  \cite{Sfetsos:2013wia} requires a Gaussian elimination of fields which when perfomed in a path integral gives rise to a dilaton 
\begin{equation}
\phi_\lambda = \phi_0 -\frac{1}{2} \log \det (1 - \lambda  \,  \textrm{ad}_{g^{-1}}) \ ,
\end{equation}
in which $\phi_0$ is constant. 
The $\lambda$-deformation can be recast  into an ${\cal E}$-model   for which $\frak{d} = \frak{g}   \oplus \frak{g} $  , whose elements are a pair  $\{ X,Y \}$,  equipped with an inner-product 
\begin{equation}
\left\llangle \{ X_1,Y_1 \}, \{ X_2, Y_2\}  \right\rrangle  =  \langle X_1, X_2\rangle - \langle Y_1, Y_2\rangle    \,  .
\end{equation} 
    With this inner-product it is clear that the diagonally embedded $G$ is a subgroup and a maximal isotropic. However the anti-diagonal, whilst being the complementary isotropic, is not a subgroup.  The specification of the  ${\cal E}$-model is completed by defining 
\begin{equation}
{\cal E}: \{X,Y\} \mapsto \frac{1+\lambda^2}{1- \lambda^2} \{ X, -Y\}  - \frac{2 \lambda}{1 -\lambda^2} \{Y, -X\}  \ ,
\end{equation} 
from which a flat space generalised metric can be obtained via eq.~\eqref{eq:genmet1}. 
The metric and B-field of the $\lambda$-model can be obtained by dressing this generalised metric with flat (algebra) indices with an appropriate frame field as in eq.~\eqref{eq:curvedspacegenmet}.  The construction of this frame field is slightly more involved than in the case of the PL model, principally because the anti-diagonal embedding of $G$ is not a subgroup and we are not dealing with a Drinfel'd double.  This feature is crucial to ensure that the WZ term in eq.~\eqref{eq:act} plays a role.  A second delicate matter is to relate the  coset representative $m(x)$, and quantities derived from it, to those obtained in terms of the group element $g(x)$ defining the $\lambda$-model. We shall return to both these points in the sequel.
 
\section{Target space description of  ${\cal E}$-models}\label{sec:TSEmodels}
We begin this section by reviewing double field theory on a group manifold, \DFTwzw{}, which will be our framework to implement ${\cal E}$-models. We will then show how the section condition can be solved by introducing a set of frame fields  that further describe the generalised geometry of ${\cal M}= \mathdsl{D}/\widetilde{H}$.  We will explain how modified supergravity arises out of this procedure.

\subsection{A brief review of  $\textrm{DFT}_{\textrm{wzw}}$}  

We now present a more consolidated  target space perspective of the discussion in the previous section.  For this we shall employ the framework of  \DFTwzw{}  \cite{Blumenhagen:2015zma}; a specification of the $O(D,D)$ symmetric double field theory that assumes an underlying group manifold, $\mathdsl D$, of dimension $2D$. The  corresponding algebra $\frak{d}$ is as in eq.~\eqref{eq:Dcomms}, and in particular is equipped with an ad-invariant inner-product, $\eta$, of split signature that will be used to raise and lower indices.

 We introduce a group  element, $\mathdsl g(X)$, depending on $X^I$, $I=1\dots 2D$, local  coordinates on $\mathdsl D$ and the left-invariant Maurer-Cartan forms $d\mathdsl  g \mathdsl g^{-1}   =\mathdsl  E^A{}_I \mathdsl T_A dX^I$ from which is constructed   $D_A =\mathdsl  E_A{}^I \partial_I$, (with $\mathdsl E_{A}{}^I$ the inverse transpose of $\mathdsl E^A{}_I$) a set of vector fields generating a right action obeying $[D_A , D_B] =\mathdsl  F_{AB}{}^C D_C$. 

\subsubsection*{The NS/NS sector}
The common NS sector  of  \DFTwzw{} consists of a generalised metric ${\cal H}^{AB}$, which {\em a priori} may depend on all of the $X^I$, and a generalised dilaton $d$.  The dynamics are encoded by a target space action \cite{Blumenhagen:2015zma}, 
\begin{align}
  S_{\mathrm{NS}} =& \int d^{2D} X e^{-2d} \Big(  \frac{1}{8} \mathcal{H}^{CD} \nabla_C \mathcal{H}_{AB}
    \nabla_D \mathcal{H}^{AB} -\frac{1}{2} \mathcal{H}^{AB} \nabla_{B} \mathcal{H}^{CD}
    \nabla_D \mathcal{H}_{AC} \nonumber \\
    & - 2 \nabla_A d \nabla_B \mathcal{H}^{AB} + 4 \mathcal{H}^{AB} \nabla_A d \nabla_B d + 
      \frac{1}{6} F_{ACD} F_B{}^{CD} \mathcal{H}^{AB} \Big)\, . \label{eqn:Sdftwzw}
\end{align} 
Here we have introduced a covariant derivative $\nabla$ that acts on a vector density $V^A$ with weight $w$ as, 
\begin{equation}\label{eqn:covderiv}
  \nabla_A V^B = D_A V^B + \frac13 F_{AC}{}^B V^C - w F_A V^B\, .
\end{equation}
The generalised metric has weight $w=0$ whilst the generalised dilaton $e^{-2d}$ has $w=1$ and $ \nabla_A d = -\frac{1}{2} e^{ 2d}\nabla_A e^{-2d}$. The density correction makes use of $ F_A = D_A \log \det\mathdsl  E$.

The local symmetries of the action comprise : 
\begin{enumerate}
\item {\em Generalised diffeomorphisms}  meditated by the generalised Lie derivative 
\begin{equation}
  {\cal L}_\xi V^A = \xi^B \nabla_B V^A -  V^B \nabla_B \xi^A + \eta^{AB}\eta_{CD}  V^C  \nabla_B \xi^D  + w    \nabla_B \xi^B V^A  \, , 
\end{equation}
\item {\em Conventional 2$D$-diffeomorphisms} meditated by the Lie derivative 
\begin{equation}\label{eq:2Ddiffeo}
  L_\xi V^A=    \xi^BD_B V^A    - w \xi^B F_B V^A  + w D_B \xi^B V^A \,. 
\end{equation} 
\end{enumerate} 
It should be emphasised here that under the conventional 2$D$-diffeomorphisms objects with curved space indices $I,J$ etc. transform tensorially whereas those with algebra indices $A,B$, transform as scalars.  In particular with respect to this transformation ${\cal H}_{AB}$ is a scalar and $\eta_{AB}$ is an invariant (i.e. constant) scalar. See table \ref{Tab:trans} for further details.

\begin{table} 
\begin{center}
   \begin{tabular}{r|lll}
          object & gen.-diffeomorphisms & $2D$-diffeomorphisms &
            global $O(D,D)$\\\hline
          $\mathcal{H}_{AB}$ & tensor  & scalar & tensor        \\
                    $\nabla_A d$ & not covariant & scalar & 1-form        \\
          $e^{-2d}$    & scalar density ($w$=$1$) & 
                         scalar density ($w$=$1$) & invariant   \\
          $\chi$  & spinor & scalar density ($w$=$\frac{1}{2}$)   & spinor \\
          $\eta_{AB}$  & invariant-tensor & invariant-scalar & invariant  \\
          $\mathdsl F_{AB}{}^C$ & invariant     & invariant & tensor     \\
          $E_A{}^I$    & invariant     & vector & 1-form        \\
          \hline
          $S_{\mathrm{NS}}$ & invariant& invariant & invariant  \\
          $S_{\mathrm{R/R}}$          & invariant     & invariant & invariant  \\
          \hline
          $D_A$        & not covariant & covariant & covariant  \\
          $\nabla_A$   & not covariant & covariant & covariant  \\
        \end{tabular}
       \caption{Transformation properties of objects under \DFTwzw{} symmetries.}
       \label{Tab:trans}
 \end{center}      
 \end{table}

Closure of the local symmetry algebra necessitates that fields and gauge parameters, and products thereof, can depend on coordinates in only a restricted way.  This restriction is called the section condition and reads 
\begin{equation}\label{eqn:SC}
  (D_A f_1 - w_1 F_A f_1) (D^A f_2 - w_2 F^A f_2) = 0 \, ,
\end{equation}
in which $f_1$ and $f_2$ indicate any field or combinations of fields with the corresponding weights $w_1$ and $w_2$, respectively. Notionally solving this condition should amount to giving a splitting of coordinates $X^I =( \tilde{x}_{\tilde i} , x^i)$ in to physical $\{x^i\}$, on which fields can depend, and non-physical $\{ \tilde x_{\tilde i} \}$ on which fields cannot depend. Once a solution to the section condition is adopted of course the full conventional $2D$-diffeomorphism symmetry is broken, and all that survives can in fact be absorbed into the generalised diffeomorphisms.

Having the action \eqref{eqn:Sdftwzw}, we can derive the corresponding equations of motion by varying it with respect to the generalised metric and the generalised dilaton. Doing so, we find \cite{Blumenhagen:2015zma}
\begin{equation}\label{eqn:varSNS}
  \delta S_{\mathrm{NS}} = \int d^{2D} X e^{-2 d}\mathcal{K}_{AB} \delta \mathcal{H}^{AB} 
    \quad \text{and} \quad
  \delta S_{\mathrm{NS}} = - 2 \int d^{2D} X e^{-2 d}  \mathcal{R}\delta d \, , 
\end{equation}
with
\begin{align}
  \mathcal{K}_{AB} &= \frac{1}{8} \nabla_A \mathcal{H}_{CD} \nabla_B \mathcal{H}^{CD} - \frac{1}{4} \big[ \nabla_C - 2 (\nabla_C d) \big] \mathcal{H}^{CD} \nabla_D \mathcal{H}_{AB} + 2 \nabla_{(A} \nabla_{B)} d \nonumber \\ &\,- \nabla_{(A} \mathcal{H}^{CD} \nabla_D \mathcal{H}_{B)C} + \big[ \nabla_D - 2 (\nabla_D d) \big] \big[ \mathcal{H}^{CD} \nabla_{(A} \mathcal{H}_{B)C} + {\mathcal{H}^C}_{(A} \nabla_C {\mathcal{H}^D}_{B)} \big]  \nonumber \\ &\,+ \frac{1}{6} F_{ACD} F_B{}^{CD} \label{eqn:Ktensor} \, ,
\end{align}
and
\begin{align}
  \mathcal{R} &=  4 \mathcal{H}^{AB} \nabla_A \nabla_B d - \nabla_A \nabla_B \mathcal{H}^{AB} - 4 \mathcal{H}^{AB} \nabla_A d\, \nabla_B d + 4 \nabla_A d \,\nabla_B \mathcal{H}^{AB} \nonumber \\ &\,+ \frac{1}{8} \mathcal{H}^{CD} \nabla_C \mathcal{H}_{AB} \nabla_D \mathcal{H}^{AB} - \frac{1}{2} \mathcal{H}^{AB} \nabla_B \mathcal{H}^{CD} \nabla_D \mathcal{H}_{AC} +\frac{1}{6} F_{ACD} F_B{}^{CD} \mathcal{H}^{AB} \,.\label{eqn:gencurvature}
\end{align}
In order to obtain the field equations, one has to take into account that $\delta\mathcal{H}^{AB}$ is not an arbitrary rank two tensor but restricted to symmetric $O(D,D)$ generators. Therefore, one introduces the generalised Ricci tensor \cite{Hohm:2010pp,Blumenhagen:2015zma}
\begin{equation}
  \mathcal{R}_{AB} = 2 P_{(A}{}^C \mathcal{K}_{CD} \overline{P}_{B)}{}^D
    \quad \text{with} \quad
      P_{AB} = \frac12 (\eta_{AB} + \mathcal{H}_{AB} )
    \quad \text{and} \quad
      \overline{P}_{AB} = \frac12 (\eta_{AB} - \mathcal{H}_{AB} )\,.
\end{equation}
It projects out the irrelevant components of $\mathcal{K}$ and allows to write the field equations for the NS/NS sector in the compact form 
\begin{equation}\label{eqn:nsnseom}
  \mathcal{R}_{AB} = 0 \quad \text{and} \quad \mathcal{R} = 0\,.
\end{equation}

\subsubsection*{The R/R sector}
Let us now  examine the R/R sector for which the target space action on $\mathdsl D$ reads \cite{Hohm:2011dv}
\begin{equation}\label{eqn:SdftwzwR/R}
  S_{\mathrm{R/R}} = \frac14 \int d^{2D} X\, (\slashed{\nabla} \chi)^\dagger \, S_{\mathcal{H}} \, \slashed{\nabla} \chi\,.
\end{equation}
Here $\chi$ is a Majorana-Weyl spinor of $Spin(D,D)$ and depending on its chirality encodes either type IIA or IIB theories. A natural way to parameterise this spinor is in terms of even or odd differential forms with degree up to $D$. Let us denote these forms as $C^{(p)}$ so that
\begin{equation}\label{eqn:formstospinor}
  \chi = \sum\limits_{p=0}^D \frac1{2^{p/2} \, p!} C^{(p)}_{a_1\dots a_p} \Gamma^{a_1} \dots \Gamma^{a_p} | 0 \rangle \, ,
\end{equation}
in which the $\Gamma$-matrices $\Gamma^A=\begin{pmatrix} \Gamma_a , & \Gamma^a \end{pmatrix}$  obey  $  \{\Gamma^A, \Gamma^B\} = 2 \eta^{AB}$ and $ | 0 \rangle$ is the Clifford vacuum annihilated by the $\Gamma_a$.   The action of an $O(D,D)$ element $O$ on a spinor, denoted as $S_O$, is implicitly defined by the Clifford relation
\begin{equation}
  \Gamma^A = S_O \Gamma^B S_{O}^{-1} O_B{}^A\,.
\end{equation}

The   covariant derivative for spinors entering the action is defined as 
\begin{equation}
  \slashed{\nabla} \chi = \Gamma^A \nabla_A \chi \quad \text{with} \quad
  \nabla_A \chi = D_A \chi - \frac{1}{12} \mathdsl F_{ABC} \Gamma^{BC} \chi - \frac12 F_A  \chi \, , 
\end{equation}
where we take into account that $\chi$ transforms as a density with weight $w=1/2$.   The Dirac operator   $\slashed{\nabla}$ is nilpotent providing that $\mathdsl F_{ABC}\mathdsl F^{ABC}= 0$, this requirement in fact follows from the section condition of DFT and we shall assume it to be the case.  

  Generalised diffeomorphisms act on the spinor as
\begin{equation}
  \mathcal{L}_\xi \chi = \xi^A \nabla_A \chi + \frac12 \nabla_A \xi_B \Gamma^{AB} \chi + \frac12 \nabla_A \xi^A \chi\, .
\end{equation}
and under 2D-diffeomorphisms it transforms exactly as in eq.~\eqref{eq:2Ddiffeo}  as a scalar with density $1/2$.
The field strengths are defined as $G =  \slashed{\nabla} \chi$.  In order that eq.~\eqref{eqn:SdftwzwR/R} describes the correct degrees of freedom a self-duality condition must be imposed \cite{Hohm:2011dv}
\begin{equation}\label{eqn:dualityconstr}
  G = - \mathcal{K} G \quad \text{with} \quad \mathcal{K} = 
    C^{-1} S_{\mathcal{H}} \, , 
\end{equation}
in which $C$ is the charge conjugation matrix.\footnote{Charge conjugation is  defined by its action
\begin{equation}
  C \, \Gamma^a \, C^{-1} = \Gamma_a = ( \Gamma^a )^\dagger
    \quad \text{and} \quad 
  C \, \Gamma_a \, C^{-1} = \Gamma^a = ( \Gamma_a )^\dagger
\end{equation}
on the $\Gamma$-matrices. This constraint requires that $(C^{-1} S_{\mathcal{H}})^2 = 1$ and therefore that $D(D-1)/2$ is odd. Thus, we can only impose it ($D\le10$) for \cite{Hohm:2011dv}
\begin{equation}
  D = \{10, 7, 6, 3, 2\}\,.
\end{equation}}

The variation of the action with respect to $\chi$ gives rise to the equations of motion
\begin{equation}\label{eqn:R/Reqm}
  \slashed{\nabla} ( \mathcal{K} G ) = 0 \, , 
\end{equation}
which is automatically satisfied providing the self-duality constraint \eqref{eqn:dualityconstr} and Bianchi identity are imposed. Furthermore, the NS/NS sector equations of motion \eqref{eqn:nsnseom} receive the additional contribution  from also varying the R/R part of the action\cite{Hohm:2011dv}
\begin{equation}\label{eqn:eomwithR/R}
  \mathcal{R}_{AB} - \frac{e^{2d}}{16} \mathcal{H}_{(A}{}^C G^\dagger \Gamma_{B)C} \mathcal{K} G   = 0 \,.
\end{equation}
 
\subsection{The generalised frame fields}\label{sec:genpp}

In order to present concrete solutions to the section condition let us restrict our attention to the case relevant to ${\cal E}$-models, i.e. that  $\mathfrak{d}$ admits a subalgebra $\tilde{\frak{ h}}\subset  \mathfrak{d}$. Let $\mathdsl T_A = (\widetilde T^a , T_a)$ where $\widetilde T^a$ are generators of $\tilde{\frak{ h}}$ and $T_a$ are the remaining generators whose span we denote  $\frak{k}$. The subalgebra is maximally isotropic with respect to  $\eta$.  The space $\frak{k}$ is automatically maximally isotropic but not necessarily a subalgebra. 
   Locally in a patch, one can always decompose a group element $ \mathdsl g \in \mathdsl D$  as   
\begin{equation}\label{eq:split}
  \mathdsl g (X^I) = \tilde{h} (\tilde x_{\tilde i}) m (x^i) \,, \quad   \quad  \tilde h\in \widetilde H \quad \text{and}  \quad m \in \text{exp}(\mathfrak{k})\,. 
\end{equation}  
This splitting should be extended globally, working patchwise if a global section $m$ is unavailable \cite{Klimcik:1996nq}.  Note that the coset-representative $m(x^i)$ is chosen to be just the exponential of coset generators; this represents a preferred choice of coordinates on $\mathdsl D/\widetilde{H}$ which will be employed in what follows.
 
Now we make one further important requirement, namely we demand that
  \begin{equation}\label{eq:struct}
   \eta^{ij} =   0
  \quad \text{of} \quad
   \eta^{IJ} = \mathdsl E_A{}^I \eta^{AB}\mathdsl E_B{}^J\,. 
\end{equation}
In was shown in  \cite{Hassler:2016srl} when  $\mathdsl D/\tilde{H}$ is identified with a group manifold (i.e. $\frak{d}$ is a Drinfel'd double) or is a symmetric coset then eq.~\eqref{eq:struct} follows directly from eq.~\eqref{eq:split}.  This will be the case in all examples we are interested in this paper,  however in anticipation that there may be more general solutions we   keep this as a separate requirement. 
 
Then the section condition in DFT is implemented by demanding physical fields just depend on the coordinates $x^i$ of the target space $\mathdsl D/\tilde{H}$ and not on $\tilde{x}_i$. A slight subtlety arises for densities $f$ of weight $w$ where only the combination
\begin{equation}
    f | \det \tilde e_a{}^{\tilde i} |^{-w} = \widehat{f}(x^i) \quad \text{with} \quad
 \widetilde T^a   \tilde{e}_a{}^{ \tilde i} d \tilde{x}_{ \tilde  i} = \tilde{h}^{-1} d \tilde{h}  . 
\end{equation}
depends on $x^i$, whilst $f$ depends on all coordinates. 
The physical fields are then the generalised metric and the corrected dilaton:
\begin{equation}\label{eqn:coorddepd}
  \widehat{\mathcal{H}}^{\hat I\hat J}( x^i )  \quad \text{and} \quad
  \widehat{d}( x^i ) =  d + \frac12 \log | \det \tilde e_a{}^{\tilde i} | \,. 
\end{equation}
The last equation takes into account that  $e^{-2d}$ is the covariant density with weight $w=1$.  Similarly in the R/R sector  the coordinate dependence of $\chi$ is restricted to
\begin{equation}
  S_{\widehat{E}} \chi = \widehat{\chi} \sqrt{| \det \tilde e^a{}_{\tilde i} | }
\end{equation}
 where $\widehat{\chi}$ depends on the physical coordinates $x^i$ only and  $S_{\widehat{E}}$ will be the spinorial counterpart of a certain frame-field we shall now define.  

We now need to express the actions eq.~\eqref{eqn:Sdftwzw} and \eqref{eqn:SdftwzwR/R} in terms of these restricted quantities $ \widehat{\mathcal{H}}^{\hat I\hat J} ,\widehat{d}$ and $\widehat{\chi}$ which can be thought of as living in the generalised tangent space of ${\mathcal{M}}=\mathdsl D/\tilde{H}$.  
To do so we shall show that when the factorisation eq.~\eqref{eq:split} is assumed we can define a  set of generalised frame fields $\widehat{E}_A{}^{\hat I}$  that obey 
\begin{enumerate}[i.]
  \item $\widehat{E}_A{}^{\hat I}$  is an $O(D,D)$ element \label{item:odd},
  \item  $\widehat{E}_A{}^{\hat I}$   only depends on the physical coordinates $x^i$,\label{item:coord}
  \item  $\widehat{E}_A{}^{\hat I}$ gives rise to the frame algebra, \label{item:framealg}
  \begin{equation}\label{eqn:framealg}
    \widehat{{\cal L}}_{\widehat{E}_A} \widehat{E}_B{}^{\hat I} = \mathdsl F_{AB}{}^C \widehat{E}_C{}^{\hat I}\, ,
  \end{equation}
   where $ \mathdsl F_{AB}{}^C$ are the structure constants of the double $\mathdsl D$ and $ \widehat{{\cal L}}$ is the  generalised Lie derivative of generalised geometry
  \begin{equation}
      \widehat{{\cal L}}_\xi V^{\hat I} = \xi^{\hat J} \partial_{\hat J} V^{\hat I} + ( \partial^{\hat I} \xi_{\hat J} - \partial_{\hat J} \xi^{\hat I} ) V^{\hat J} \,.
  \end{equation}
\end{enumerate}
At this stage we are working locally however where these frame fields can be globally extended they would define a {\em generalised Leibniz parallelisation} on  ${\mathcal{M}}=\mathdsl D/\tilde{H}$  \cite{Grana:2008yw,Lee:2014mla}. Also note that because of condition \ref{item:coord} we may use the term generalised Lie derivative and Courant bracket are interchangeable here. 

The hatted notation on indices is introduced to emphasise   quantities  that take values in the   generalised tangent space  of $\cal M$, i.e. when the section condition has been solved so for example
\begin{equation}\label{eqn:defhattedind}
  V^{\hat I} = \begin{pmatrix} v_i \\ v^i \end{pmatrix}\,, \quad
  \partial_{\hat I} = \begin{pmatrix} 0 & \partial_i \end{pmatrix}\,, \quad \text{and} \quad
 \eta_{\hat I\hat J} = \begin{pmatrix} 0 & \delta^i_j \\
    \delta^j_i & 0
  \end{pmatrix}\,. 
\end{equation}
 In the present context we have the following useful theorem:

  \begin{theorem}\label{thm:one}
  For each group $\mathdsl D$, with a non-degenerate, bilinear, ad-invariant split form $\eta$, and a maximally isotropic subgroup   $\widetilde{H}$, there exists   generalised frame fields on ${\cal M}= \mathdsl D/\widetilde{H}$ where eq.~\eqref{eq:split} and eq.~\eqref{eq:struct} hold  that obey   conditions \ref{item:odd}-\ref{item:framealg} above. These  are realized by
  \begin{equation}\label{eqn:genframefield}
    \widehat{E}_A{}^{\hat I} = M_A{}^B \widehat{V}_B{}^{\hat I} = M_A{}^B \begin{pmatrix}
      v^b{}_i & 0 \\
      v_b{}^j \rho_{ji }   & v_b{}^i
    \end{pmatrix}\indices{_B^{\hat{I}}}\,,
  \end{equation} 
with 
\begin{equation}
  M_A{}^B \mathdsl T_B = m\mathdsl T_A m^{-1} \,, \quad 
 T_a v^a{}_i d x^i + \widetilde T^a A_{ai} d x^i =\mathdsl T_A V^A{}_i d x^i = d m m^{-1} \,,  \quad
  v^a{}_i v_a{}^j = \delta_i^j \,,
\end{equation}
and $\rho_{ij}$ the components of a two-form 
\begin{equation}\label{eqn:rho2def}
 \rho^{(2)} = \frac{1}{2} \rho_{ij } dx^i \wedge dx^j = \omega^{(2)}  - \Omega^{(2)} \ , 
\end{equation} 
in which 
 \begin{equation}
  \omega^{(2) }   = \frac{1}{2}  v^a{}_i  A_{aj} dx^i \wedge dx^j  \ , 
\end{equation} 
and   $\Omega^{(2)} $ chosen such that
\begin{equation} \label{eq:domega}
d \Omega^{(2)}  = \Omega^{(3)} =    \frac{1}{12 } \llangle dm m^{-1}, [ dm m^{-1},dm m^{-1}]  \rrangle =   \frac{1}{12} \mathdsl F_{ABC} V^A\wedge V^B \wedge V^C \, . 
\end{equation} 
\end{theorem}
\begin{proof}
Condition \ref{item:odd} is trivially satisfied. The   adjoint action of any group element in $\mathdsl D$, and in particular $M_A{}^B$,   is   an $O(D,D)$ element as $\eta$ is adjoint invariant. The second part of \eqref{eqn:genframefield},   $\widehat{V}_B{}^{\hat I}$,  is also $O(D,D)$ valued, indeed it is the product of a $b$-field transformation and a $GL(D)$ action. 
  By construction the frame fields only depend on the coordinates $\{ x^i \}$ and the condition \ref{item:coord} is automatic.     Finally we have to check the frame algebra condition \ref{item:framealg}.   First we make use of the easy identity 
   \begin{equation}\label{eq:dM}
  \partial_{\hat I} M_A{}^B = V^C{}_{\hat{I}} M_A{}^D \mathdsl F_{CD}{}^B \, , 
  \end{equation}
  to show that 
   \begin{equation}
  \widehat{E}_{C \hat I }    \widehat{{\cal L}}_{\widehat{E}_A} \widehat{E}_B{}^{\hat I} =   M_{A}{}^D M_{B}{}^E M_{C}{}^F \left( T_{DEF} + S_{DEF}     \right)  \, ,
   \end{equation}
in which
\begin{equation}
T_{ABC} =  3\widehat{V}_{[A}{}^{\hat I}\partial_{\hat I}(  \widehat{V}_B{}^{\hat J })    \widehat{V}_{C] \hat J } \, , \quad S_{ABC}=  3 \Lambda_{[A}{}^G \mathdsl F_{BC]G} \, ,    \quad \Lambda_{A}{}^B =   \widehat{V}_A{}^{i} V^B{}_{i} = \begin{pmatrix} 0 & 0 \\
    v_a{}^i A_{b i} & \delta_a^b \end{pmatrix}  \, . 
\end{equation}
Since the structure constants are invariant under the adjoint action the proof is completed provided  
\begin{equation}
 T_{ABC} + S_{ABC}    =\mathdsl F_{ABC} \ . 
\end{equation} 
This can be verified component by component:  
\begin{equation} 
\begin{aligned} 
S^{abc} = 0    \ , & \quad T^{abc} =  0  \ ,  \\ 
S^{ab}{}_c =\widetilde F^{ab}{}_c  \ , & \quad T^{ab}{}_{c} =  0  \ ,   \\  
S^{a}{}_{bc}  = 2  F^{a}{}_{bc}   + 2\widetilde F^{d a}{}_{[b} \Lambda_{c] d}     \ , & \quad  T^{a}{}_{bc} =     -F^{a}{}_{bc}    - 2\widetilde F^{d a}{}_{[b} \Lambda_{c] d}  \ ,    \\  
S_{abc} = 3 F_{abc} + 3F^d{}_{[ab}   \Lambda_{c] d}   \ , &\quad T_{abc}  = 3 v_a{}^i   v_b{}^j  v_c{}^k   \partial_{[i} \rho_{jk]}  \, . 
\end{aligned} 
\end{equation} 
The proof is concluded by substituting the derivative of $\rho$   from  eq.~\eqref{eq:domega} calculated using  
\begin{equation} \label{eq:domega2}
d  \omega^{(2) }  = -\frac{1}{4}  \left(\mathdsl F_{ABc} V^A \wedge V^B \wedge v^c -\mathdsl  F_{AB}{}^c V^A \wedge V^B \wedge A_c    \right) \ ,  
\end{equation}  
which follows from the Maurer-Cartan identity for $V^A$. 
\end{proof}
 
\newtheorem{comment}{Comment}

   \begin{comment} \label{comm:3}
   Using these frame fields we construct derivative operators $\partial_{\hat I} = \widehat{E}^A{}_{\hat I } D_A$.  As detailed in Appendix \ref{app:framealg},  $\eta^{\hat  I \hat J} \partial_{\hat I} \partial_{\hat J} \equiv 0$ and therefore the section condition of DFT is solved.  
       \end{comment}

    \begin{comment} \label{comm:2} The twisting of a Courant bracket  by an ${\bf H} \in H_3(M , \mathbb{R} )$ (see e.g. \cite{2003math10359K})  provides an interpretation for the use of $\Omega^{(2)}$, which may not exist globally (even if the decomposition eq.~\eqref{eq:split} does).   Defining new generalised frame fields $\widehat{E}^\prime_A{}^{\hat I}$ as in  eq.~\eqref{eqn:genframefield} but with now $\rho_{ij} = \omega_{ij}^{(2)}$    we have that

    \begin{equation} 
       \widehat{{\cal L}}_{\widehat{E}^\prime_A} \widehat{E}^\prime_B{}^{\hat I} = F_{AB}{}^C \widehat{E}^\prime_C{}^{\hat I} +\begin{pmatrix}  \widehat{E}^\prime_A{}^j   \widehat{E}^\prime_B{}^k  \Omega^{(3)}_{jk i }   \\  0 \end{pmatrix} \ , 
       \end{equation} 
     such that ${\bf H} = \Omega^{(3)}$ appears as such a twisting.       If $\frak{k}$ is also a subalgebra (as in the case of a Drinfel'd double) this vanishes.   
        \end{comment} 
  
 \begin{comment}
  The assumption of eq.~\eqref{eq:struct} allows us to introduce, in addition to $\eta_{AB}$ and ${\cal H}_{AB}$, the structure 
  \begin{equation}
 \tilde \omega_{AB} = \begin{pmatrix}   0 & -1 \\ 1 & 2A_{ai} v^{i}{}_b \end{pmatrix}  \ . 
  \end{equation}   
 It will be shown in a forthcoming paper by one of the authors \cite{parakahler} that   $\tilde \omega$, dressed with an appropriate adjoint action to transport it around the group manifold $\mathdsl D$ , and  $\eta$ equips  $T\mathdsl D$ with an (almost) para-hermitian structure.    An interesting question, beyond the present scope, is to establish the circumstances in which ${\cal H}_{AB}$ will allow a full Born geometry of  \cite{Freidel:2018tkj}.  Recently these structures were examined for the case of Drinfel'd double \cite{Marotta:2018myj} and DFT \cite{Svoboda:2018rci}.
      \end{comment}

Let us close this section with some further properties of the frame field that will be employed in the sequel.   We consider the quantity
\begin{equation}\label{eq:Omega}
\widehat{\Omega}_{\hat{I} \hat{J} \hat{K}} = - \partial_{\hat{I}}   \widehat{E}_{\hat{J}A }  \widehat{E}_{\widehat{K}}{}^A \ . 
  \end{equation} 
An immediate consequence of the frame algebra is that 
\begin{equation}\label{eq:Fdef}
\widehat{\Omega}_{[\hat{I} \hat{J} \hat{K}] } =  \frac{1}{3}  \widehat F_{\hat{I} \hat{J} \hat{K} } \equiv  \frac{1}{3}   \widehat{E}^{A}{}_{\hat{I}}    \widehat{E}^{B}{}_{\hat{J}}    \widehat{E}^{C}{}_{\hat{K}}\mathdsl  F_{ABC}  \ . 
  \end{equation} 
 We shall also need the contraction 
\begin{equation}
\widehat{\Omega}_{  \hat{K}} =   \eta^{  \hat{I} \hat{J}}   \widehat{\Omega}_{\hat{I} \hat{J} \hat{K}} \ , 
  \end{equation} 
which can be simplified by making use of eq.~\eqref{eq:dM} to show 
\begin{equation}
\widehat{V}_A{}^{\hat{I}} \widehat{\Omega}_{  \hat{I}}  =  \Lambda_D{}^E\mathdsl F_{AE}{}^D - \partial_{\hat{I}} \widehat{V}_A{}^{\hat{I}} \, . 
  \end{equation} 
Then it follows using the   Maurer--Cartan identity for $V^A \mathdsl T_A = dm m^{-1}$ that 
  \begin{equation}\label{eq:OmegaI}
 \widehat{\Omega}_{  \hat{I}}  = \left( -v_a{}^i\tilde f^a,  \partial_i \log \det v^a{}_j  - \Omega^{(2)}_{ij}   v_a{}^j\tilde f^a    \right)_{  \hat{I}}  \ . 
  \end{equation} 
    in which
    \begin{equation}
    \quad \tilde f^a= \widetilde F^{ab}{}_b\ , \quad f_a = F_{ab}{}^b  \,.
    \end{equation}
    In the special case of a Drinfel'd double, the final term involving $\Omega^{(2)}$ vanishes. 
  Here we see that  $\widehat{\Omega}_{  \hat{I}}$ will play an important role in the case that $\tilde{\frak{h}}$ is {\em non-unimodular}, i.e when $ \tilde{F}^{ab}{}_b \neq 0$. The quantities $\widehat F_{\hat{I} \hat{J} \hat{K} }$ can be thought of as generalised fluxes and, whilst not essential for what follows, this view is explored in the appendix~\ref{appendix:fluxes}.
   
\subsection{Equivalence to (modified) type II supergravity}

In what follows, we apply the idea of \cite{Hassler:2017yza} and rewrite the action \eqref{eqn:Sdftwzw} using the generalised frame field. We parametrise the generalised metric as in eq.~\eqref{eq:curvedspacegenmet}   and write
\begin{equation}
  \widehat{\mathcal{H}}^{\hat I\hat J}   
    = \widehat{E}_A{}^{\hat I} \mathcal{H}^{AB} \widehat{E}_A{}^{\hat J}\,.
\end{equation}
Similarly, we write the generalised dilaton, recalling eq.~\eqref{eqn:coorddepd} as
\begin{equation}
  \widehat{d} = \phi - \frac14 \log | \det g_{ij} | = d + \frac12 \log | \det \tilde{e}^a{}_{\tilde k} | \,.
\end{equation}

Taking into account that,
\begin{equation}
  \partial_{\hat I} \widehat{\varphi}(x^i) = \widehat{E}^A{}_{\hat I} D_A \widehat{\varphi}(x^i) \, , 
\end{equation}
we can make use of the property that  $\widehat{E}_A{}^{\hat I}$ satisfies the frame algebra \eqref{eqn:framealg} to pull covariant derivatives to generalised tangent bundle.    An illustrative example is $\nabla_A$ acting on a weightless vector $V^B$ for which
\begin{equation}\label{eqn:covVtoTMT*M}
  \nabla_A V^B \rightarrow \partial_{\hat I} \widehat{V}^{\hat J} + (\widehat{\Omega}_{[\hat I\hat K\hat L]} - \widehat{\Omega}_{\hat I\hat K\hat L}) \eta^{\hat L\hat J} \widehat{V}^{\hat K}
\end{equation}
with $  \widehat{\Omega}_{\hat I\hat J\hat K}$ defined in eq.~\eqref{eq:Omega}.
The generalization to higher rank tensors follows immediately.

Next, we take a look at the covariant derivative for the generalised dilaton for which we must pay attention to the  weight factor, 
\begin{align}
  \nabla_A d &= D_A d + \frac12   D_A \log| \det\mathdsl  E^B{}_I | = D_A \widehat{d} + \frac12 D_A \log | \det v^a{}_i |\,.
\end{align}
 Making of use of eq.~\eqref{eq:OmegaI} we have 
\begin{align}\label{eqn:ourXInonDD}
   \nabla_A d & \rightarrow \partial_{\hat{I}} \widehat{d} + \mathbf X_{\hat{I}} + \frac12  \widehat{\Omega}_{\hat{I}}\, , \quad  \mathbf X_{\hat I} = \frac12\begin{pmatrix} \tilde f^a v_a{}^i\\   \tilde f^a   v_a{}^j  \Omega^{(2)}_{ij}     \end{pmatrix} \ . 
 \end{align}
When $\tilde{\frak{h}}$ is non-unimodular  we have by definition  $   \mathbf X_{\hat I} \neq 0$ \cite{Sakatani:2016fvh}. In the following, we will first consider the NS sector, treating unimodular and   non-unimodular cases in turn,  and then  we discuss  the R/R sector pertaining to both cases.

\subsubsection{Unimodular Case}
Pulling all quantities in the action~\eqref{eqn:Sdftwzw} to the generalised tangent space, and doing some algebra, we obtain for $\mathbf X_{\hat I} = 0$ the action of DFT with the section condition solved 
\begin{align}\label{eqn:SdftwzwRed}
  S_{\mathrm{NS}} &=  V_{\tilde{H}}  \int d^D x e^{-2 \widehat{d}} \Big(  \frac{1}{8} \widehat{\mathcal{H}}^{\hat K\hat L} \partial_{\hat K} \widehat{\mathcal{H}}_{\hat I\hat J} \partial_{\hat L} \widehat{\mathcal{H}}^{\hat I\hat J} - \\
    & 2 \partial_{\hat I} \widehat{d} \partial_{\hat J} \widehat{\mathcal{H}}^{\hat I\hat J} -\frac{1}{2} \widehat{\mathcal{H}}^{\hat I\hat J} \partial_{\hat J} \widehat{\mathcal{H}}^{\hat K \hat L} \partial_{\hat L} \widehat{\mathcal{H}}_{\hat I\hat K}  + 
    4 \widehat{\mathcal{H}}^{\hat I\hat J} \partial_{\hat I} \widehat{d} \partial_{\hat J} \widehat{d} \Big) \, .
    \nonumber 
\end{align} 
 All occurrences of $\widehat{\Omega}_{ \hat I\hat K\hat L }$ either directly cancel or occur in contractions that vanish due to working in a particular solution of the section condition. Let us emphasise that in eq.~\eqref{eqn:SdftwzwRed}  the section condition has been implemented, the fields only depend on the coordinates $x$, the integral  is only over these coordinates, and the integration over $\tilde{x}$ has been performed with a volume $V_{\tilde{H}}$ arising from the dilaton factor in the measure.
It is by now well established \cite{Hohm:2011dv,Hohm:2011zr} that the equations of motion derived from this theory can be equated to the common NS sector (super)gravity field equations for $g_{ij}, B_{ij}, \phi$ (see Appendix A for the supergravity field equations used).

\subsubsection{Non-Unimodular Case}
If $\tilde H$ is not unimodular, we instead obtain generalised type II \cite{Arutyunov:2015mqj}.  This  is a modification at the level of the equations of motion, described in detail in Appendix A, that depends crucially on a     Killing vector
$I$  obeying    
 \begin{equation} \label{eqn:geneom1}
 L_I g = 0  \, , \quad L_I H = 0  \, ,  
 \end{equation} 
where $L_I$ is the conventional Lie derivative along $I$, and a one form $Z$ 
 further constrained to obey 
\begin{equation} \label{eqn:geneom2}
d Z + \iota_I H = 0 \ , \quad \iota_I Z = 0  \, . 
\end{equation} 
The conditions eq.~\eqref{eqn:geneom2} allows the construction of a differential which acts on the formal sum of forms  
\begin{equation}\label{eq:exder}
  \mathbf{d}    = d ~ + H \wedge ~  - Z \wedge ~ - \iota_I~  \, , \quad   \mathbf{d}^2    =   -L_I   \ , 
\end{equation}
such that when the differential form is invariant under $I$ the differential is nilpotent.  This differential operator will be important when discussing the R/R sector. 
 
The   first of eq.~\eqref{eqn:geneom2} may be integrated to yield   
\begin{equation}\label{eq:defZ}
Z = d \phi + \iota_I B - V   \, ,
\end{equation}
  in which $H=dB$ locally and $L_I B = dV$. In the absence of the modifications due to the Killing vector $I$ the scalar field $\phi$ coincides with the conventional dilaton.\footnote{To make contact with the notation of   \cite{Sakatani:2016fvh} we define  $U = \iota_I B - V $ such that $Z = d\phi + U$. The split of $d\phi$ and $U$ is somewhat arbitrary since one could shift $\phi \to \phi + \alpha  $ and $U  \to U - d\alpha$ and so can be fixed by demanding 
  \begin{equation}\label{eqn:geneom3}
  L_I \phi = 0 \, , \quad \iota_I U = 0 \, .   
  \end{equation}}
  In the language of mDFT \cite{Sakatani:2016fvh,Sakamoto:2017wor},   the corresponding modifications to the DFT equations of motion are implemented by a shift in the derivative of the DFT dilaton
  \begin{equation}
  \partial_{\hat I} \widehat{d}  \rightarrow   \partial_{\hat I} \widehat{d} + \mathbf X_{\hat I}\,.\label{Eq-gSUGRAshiftdil}
\end{equation}
  The DFT shift vector in \eqref{Eq-gSUGRAshiftdil} is related to the modified supergravity vectors by,
\begin{align}
	\mathbf X_{\hat{I}}
		= \begin{pmatrix}
	 	I^i\\
  - V_i 
 \end{pmatrix}  \,.\label{Eq-shiftmDFTvector}
\end{align}
 The DFT vector $\mathbf X^{\hat I}$ is not an arbitrary, instead reflecting the requirements eq.~\eqref{eqn:geneom1}-\eqref{eqn:geneom3}  it is constrained to be a generalised Killing vector\footnote{To see this recall that a DFT gauge transformation generated by a vector $\mathbf V_{\hat{I}}  = ( v^i , \tilde{v}_i)$ acts as $\delta {\cal H}_{\hat{I}\hat{J}} =    \widehat{{\cal L}}_{\mathbf V}  {\cal H}_{\hat{I}\hat{J}}$ and in the solution to the section condition $\partial_{\hat{I}} = (0, \partial_i)$ simply generates diffeomorphism $ \delta g = L_v g $  and gauge transformations   $ \delta B = L_v B  + d \tilde{v}$.} obeying, 
\begin{equation}\label{eqn:constrX}
  \widehat{\eta}^{\hat I\hat J} \mathbf X_{\hat I} \mathbf  X_{\hat J} = 0\,, \quad
  \widehat{\mathcal{L}}_{ \mathbf X} \widehat{\mathcal{H}}^{\hat I\hat J} = 0 \,
    \quad \text{and} \quad
  \widehat{\mathcal{L}}_{ \mathbf X} \widehat{d} = 0\,.
\end{equation}
 
Since we know already in the unimodular case that the DFT equations of motion are recovered, it follows that in the non-unimodular case the mDFT equations are recovered with the identification of the DFT vector $\mathbf X^{\hat I}$ with that in eq.~\eqref{eqn:ourXInonDD} i.e. with 
\begin{equation}
I =\frac{1}{2} \tilde f^av_a{}^i \partial_i \, , \quad V=   \iota_I \Omega^{(2)}  \,  . 
\end{equation}  
Whilst $V$ here depends on a choice of $\Omega^{(2)}$ it was shown in \cite{Sakatani:2016fvh} that in fact there is a gauge freedom that allows one to take $V$ to be zero.    

It is immediate that the first of eq.~\eqref{eqn:constrX} holds but we   now investigate under what circumstances the remaining constraints of eq.~\eqref{eqn:constrX} are valid. 
Here we make use of the generalised frame fields and transport the results back to the flat indices with $\xi^A =\mathbf X^{\hat I}   \widehat{E}_{\hat I}{}^A $  and $ \mathcal{H}^{AB}   =    \widehat{E}_{\hat I}{}^A \widehat{\mathcal{H}}^{\hat I\hat J}  \widehat{E}_{\hat J}{}^B$.    A short calculation shows that, 
\begin{equation}
\widehat{E}_{\hat I}{}^A  \widehat{E}_{\hat J}{}^B \widehat{\mathcal{L}}_{ \mathbf X} \widehat{\mathcal{H}}^{\hat I\hat J} =   \xi^C D_C \mathcal{H}^{AB}  \, , 
\end{equation}  
and hence if $\mathcal{H}^{AB}$ is constant (as is in case of $\cal E$ models)  the  second of eq.~\eqref{eqn:constrX} holds.  
For the third  of eq.~\eqref{eqn:constrX} we have  
\begin{align*}
  \widehat{\mathcal{L}}_{ \mathbf X} \widehat{d} &= { \mathbf X}^{\hat I} \partial_{\hat I} \widehat{d} - 
    \frac12 \partial_{\hat I} { \mathbf X}^{\hat I} \nonumber   \\ 
    &=  { \mathbf X}^{\hat I} \left(  \partial_{\hat I} \widehat{d}  + \frac{1}{2}   \partial_{\hat I}  \log\det v_i{}^a \right) - \frac{1}{4} \tilde f^af_a + \frac{1}{4} \tilde f^e F_{e}{}^{pq} \Lambda_{pq} \,.
\end{align*}
Now taking the trace of the Jacobi identity for the subalgebra $\tilde{ \frak{h}}$ yields  $\tilde f^a \widetilde F^{bc}{}_a \equiv 0 $  such that the final term in the above vanishes.  A   consequence of the Jacobi identities for $\frak{d}$ is that $\mathdsl F_{ABC}\mathdsl F^{ABC}  = 4 \tilde f^af_a$. Since  we   require   $\mathdsl F_{ABC}\mathdsl F^{ABC} = 0$   to avoid violations of the section condition for the cases we are interested in $\tilde f^af_a =0$.  

  We can then make use of eq.~\eqref{eqn:coorddepd} to recast the result in terms of conventional volume preserving 2D-diffeomorphism acting on $d$:
 \begin{equation}
  \widehat{\mathcal{L}}_{ \mathbf X} \widehat{d} =  L_\xi d   \equiv  -\frac{1}{2} e^{2d}  L_\xi e^{-2 d}   \ . 
 \end{equation}  
 Hence when the \DFTwzw{}  dilaton, $d$, is invariant under   2D-diffeomorphisms then indeed we obey the criteria in eq.~\eqref{eqn:constrX} and  it is evident that we reproduce the field equations of modified SUGRA.

It is interesting to ask what happens at the level of the action since it is thought that generalised SUGRA does not admit an  action \cite{Baguet:2016prz}.  So what goes wrong when we try to derive an action analogous to \eqref{eqn:SdftwzwRed} by translating to the generalised tangent space? To solve this puzzle, remember that we need integration by parts to make sense of the action \eqref{eqn:SdftwzwRed} and i.e. derive the field equations. This operation requires that the identity
\begin{equation}
  \int d X^{2D} \partial_I ( | \det\mathdsl E |\mathdsl E_A{}^I \phi ) = \int d X^{2D} | \det \mathdsl E | D_A \phi\
\end{equation}
holds. A quick calculation shows that this relation requires $\mathdsl F_{AB}{}^B = 0$, which is always the case because the full double is always unimodular. However in \eqref{eqn:SdftwzwRed}, we also integrate out the non-physical directions $\{\tilde x_{\tilde i}\}$ to obtain an action just on the physical target space $\mathdsl D/\tilde H$. For the unimodular case this is perfectly fine because integration by parts works on $\mathdsl D/\tilde  H$ as well as on $\tilde  H$ independently.   But in the non-unimodular case, $\widetilde  F^{ab}{}_b = \tilde f^a\ne 0$ obstructs integration by parts on $\tilde H$. Therefore, we are not allowed to integrate out $\tilde H$ and write an action just on $\mathdsl D/\tilde  H$.  Instead we require a genuinely doubled action. That also explains the problems in conventional DFT/EFT to find an action. There,  the integration is only performed over the physical space after solving the section condition, while in  \DFTwzw{} it is always over the full space.

\subsubsection{R/R sector}
As for the NS/NS sector, we now want to show that this description is equivalent to the R/R sector of type IIA/B supergravity, or modified type II SUGRA if $\widetilde{H}$ is not unimodular. So, we pull all quantities to the generalised tangent space. We start with the covariant derivative
\begin{equation}
    | \det \tilde e^{a}{}_i |^{-1/2} S_{\widehat{E}} \slashed{\nabla} \chi = \left(
    \slashed{\partial} - ( \slashed{\partial} S_{\widehat{E}}) S_{\widehat{E}}^{-1} - 
    \frac1{12} \widehat  F_{\hat I\hat J\hat K} \widehat{\Gamma}^{\hat I\hat J\hat K} - 
    \frac12 \slashed{\partial} \log | \det v^a{}_i | \right) \widehat{\chi}
\end{equation}
which arises after substituting $\widehat{\chi} = |\det \tilde e^{a}{}_i  |^{-1/2} S_{\widehat{E}} \chi$ and identifying
\begin{equation}
  \widehat{\Gamma}^{\hat I} = S_{\widehat{E}} \Gamma^A S_{\widehat{E}}^{-1} \widehat{E}_A{}^{\hat I}\ , \quad     \slashed{\partial} =   \widehat{\Gamma}^{\hat I } \partial_{\hat{I}} \,.
\end{equation}
We can simplify further as
\begin{align}
 - ( \slashed{\partial} S_{\widehat{E}}) S_{\widehat{E}}^{-1} &=
    \frac14 \widehat{\Omega}_{\hat I\hat J\hat K} \widehat{\Gamma}^{\hat I} \widehat{\Gamma}^{\hat J\hat K} =
    \frac1{12}  \widehat  F_{\hat I\hat J\hat K} \widehat{\Gamma}^{\hat I\hat J\hat K} + 
    \frac12 \widehat{\Omega}^{\hat J}{}_{\hat J\hat I} \widehat{\Gamma}^{\hat I} \nonumber \\
  &= \frac1{12} \widehat F_{\hat I\hat J\hat K} \widehat{\Gamma}^{\hat I\hat J\hat K} + 
    \frac12 \slashed{\partial} \log | \det v^a{}_i | - {\bf X}_{\hat I} \widehat{\Gamma}^{\hat I}\, .
\end{align}
Thus we have 
\begin{equation}
\widehat{G}    = | \det\tilde e^{a}{}_{\tilde i} |^{-1/2} S_{\widehat{E}}  G =  | \det\tilde e^{a}{}_{\tilde i} |^{-1/2} S_{\widehat{E}} \slashed{\nabla} \chi = \left( \slashed{\partial} -
    {\bf X}_{\hat I} \widehat{\Gamma}^{\hat I} \right) \widehat{\chi} \,.
\end{equation}

We are now able to consider the self-duality constraint eq.~\eqref{eqn:dualityconstr} pulled to the generalised tangent space which gives
\begin{equation}
 \widehat{G}    = - C^{-1} S_{ \widehat{{\cal H}}  } \widehat{G}   \, . 
\end{equation}
To cast the results in the simplest form we follow  \cite{Hohm:2011dv} and make use of the decomposition of the spinor representative of the generalised metric\footnote{ This follows from writing  $ \widehat{{\cal H}} = \begin{pmatrix} 1& 0\\ B &1  \end{pmatrix} \begin{pmatrix} g^{-1} & 0\\ 0  &g   \end{pmatrix}\begin{pmatrix} 1& -B \\ 0 &1  \end{pmatrix} $. }
\begin{equation}
S_{\cal \widehat{H}} = S_{B}^{-1} S_{g^{-1}} S_B  \quad \text{with}  \quad   S_B = \exp ( - B_{ij}   \widehat\Gamma^{ij}   )  \ . 
\end{equation}
Defining 
\begin{equation}
\widehat{{\cal F} } = e^{\phi} S_B  \left( \slashed{\partial} -
    {\bf X}_{\hat I} \widehat{\Gamma}^{\hat I} \right) \widehat{\chi}  \, , 
\end{equation}
then,  from   \cite{Hohm:2011dv}, the self-duality condition reads,
\begin{equation}
\widehat{{\cal F} }  = -  S_g C^{-1}  \widehat{{\cal F} }  \, . 
\end{equation}
We also define a different set of potential $ \widehat{\alpha }=   e^{\phi} S_B   \widehat{\chi}$ such that 
\begin{equation}
\widehat{{\cal F} }  =   e^{\phi} S_B  \left( \slashed{\partial} -
    {\bf X}_{\hat I} \widehat{\Gamma}^{\hat I} \right)   e^{- \phi} S_B^{-1}  \widehat{\alpha }=   \mathbf{d}  \widehat{\alpha }
\end{equation}
  Here note the appearance of the  exterior derivative introduced in eq.~\eqref{eq:exder}. This is exactly as the R/R sector enters in mDFT   in  \cite{Sakatani:2016fvh,Arutyunov:2015mqj}.   Combining  the  Bianchi identity ${\bf{d}}\widehat{{\cal F} } =0$  and  eq.~\eqref{eq:exder} shows that the Lie derivative $L_I \widehat{{\cal F} } =0$   without imposing any further constraints on the R/R fields.

\section{The ${\cal E}$-model conditions} \label{sec:RRsector}
Here we define how the condition for Poisson-Lie symmetry or more generally the structure behind an ${\cal E}$-model can be simply stated in the context of \DFTwzw{}.  Namely we propose:
 \vskip 0.33cm
{\em{The conditions of an $\cal E$-model are that the fields  ${\cal H}^{AB}$, $d$ and  $G$  of \DFTwzw{} are   invariant under volume preserving 2D-diffeomorphisms.}}
 \vskip 0.33cm
\noindent In this section we shall follow through this proposal   to constrain the structure of the dilaton and R/R sector.  
\subsection{NS sector}
 
In this sector the condition simply implies that the ${\cal H}^{AB}$ is a constant, exactly matching the set up in section \ref{sec:Emodels}.    Applying this restriction, and similar on the dilaton that we turn to momentarily,  the equations of motion simplify significantly. The Ricci scalar reduces to the scalar potential of gauged supergravity
\begin{equation}
  \mathcal{R} = \frac1{12} \mathdsl F_{ACE} \mathdsl F_{BDF} \left( 3 \mathcal{H}^{AB} \eta^{CD} \eta^{EF} - \mathcal{H}^{AB} \mathcal{H}^{CD} \mathcal{H}^{EF} \right) \, , 
\end{equation}
without section condition violating contributions $1/6 \mathdsl F_{ABC} \mathdsl F^{ABC}$. For the generalised curvature   $\mathcal{R}_{AB}$, we find
\begin{equation}
  \mathcal{R}_{AB} = \frac18 (\mathcal{H}_{AC} \mathcal{H}_{BF} - \eta_{AC} \eta_{BF} ) (\mathcal{H}^{KD} \mathcal{H}^{HE} - \eta^{KD} \eta^{HE})\mathdsl F_{KH}{}^C\mathdsl F_{DE}{}^F \,.
\end{equation}
This results matches perfectly with the RG flow calculation for a double sigma model presented in \cite{Avramis:2009xi,Sfetsos:2009vt} in which 
\begin{equation}
\frac{\partial {\cal H}_{AB}}{\partial \log \mu   } = {\cal R}_{AB} \ . 
\end{equation}

\subsection{Dilaton}
For the dilaton we have to take into account that the covariant quantity $e^{-2 d}$ has weight $w=1$. Hence,   we demand 
\begin{equation}
  L_\xi e^{-2 d} \equiv \xi^A D_A e^{-2 d} -\xi^AF_Ae^{-2d} + D_A \xi^A e^{-2 d} = 0\,,
\end{equation}
where we recall  $F_A=D_A\log |\det \mathdsl E^B{}_I|$.  The last term vanishes, because the $2D$-diffeomorphisms which we are considering are area preserving. This leaves us with
\begin{equation}
  \xi^I \partial_I ( 2 d + \log | \det v^a{}_i | + \log | \det \tilde v_{a}{}^{\tilde i} | ) = 0\,.
\end{equation}
Plugging in the expression for the generalised dilaton
\begin{equation}
  d = \phi - \frac14 \log | \det g_{ij} | - \frac12 \log | \det \tilde v_{a}{}^{\tilde i} |\,,
\end{equation}
we obtain the condition
\begin{equation}\label{eqn:conddilaton}
  \phi - \frac14 \log | \det g_{ij} | + \frac12 \log | \det v^a{}_i |  =\phi_0 \, ,
\end{equation}
with $\phi_0$ a constant.  

For the case of unimodular PL models it can be seen in a few lines that this condition is fulfilled by the dilaton introduced using heavy duty mathematical treatment in \cite{Jurco:2017gii}. 
The details of this equivalence are provided as  appendix  material in section  \ref{app:dilaton}.  Similarly for $\lambda$-models this prescription provides the known dilaton, also detailed in \ref{app:dilaton}.

\subsection{R/R sector}
We demand that the field strength $G=\slashed{\nabla} \chi$ is invariant under arbitrary $2d$ diffeomorphisms i.e.
\begin{equation}
 L_\xi G \equiv \xi^A D_A G - \frac12 ( \xi^A F_A - D_A \xi^A ) G =0 \,.
\end{equation}
In general for a scalar density, $G$, of 2D-diffeomorphisms of weight $w$ (and here $w  = \frac{1}{2}$) we could define 
\begin{equation}
G = |\det \mathdsl E |^w G_0 \, 
\end{equation}
and the invariance condition is satisfied for $G_0$ being constant.  Here we have a further consequence since we can make use of the definition of the covariant derivative to show this requires that 
\begin{equation}\label{eq:Geq}
\nabla_A G = - \frac{1}{12} \mathdsl F_{ABC} \Gamma^{BC} G \, , 
\end{equation}
and as a consequence, assuming the Bianchi identity  $0 = \slashed{\nabla} G$,  upon contraction with a gamma matrix we have a necessary condition 
\begin{equation}\label{eq:Geq2}
 \mathdsl F_{ABC} \Gamma^{ABC} G =0\, .
\end{equation} 
Notice that the operator $ \slashed{\widehat{F}}$ is nilpotent by virtue of the standard properties of $\Gamma$-matrices and the Jacobi identity of $\mathdsl F_{AB}{}^{C}$.  Taking into account the dilaton and the R/R spinor weights we have that the equation of motion involves purely constant algebraic quantities
\begin{equation}\label{eqn:eomwithR/R}
  \mathcal{R}_{AB} - \frac{e^{2\phi_0}}{16} \mathcal{H}_{(A}{}^C G_0^\dagger \Gamma_{B)C} \mathcal{K} G_0   = 0 \,.
\end{equation}
   When transported to generalised tangent space via $  \widehat{G} = | \det \tilde v_{a}{}^{\tilde i} |^{-1/2} S_{\widehat{E}} G $  we simply have 
\begin{equation}
  \slashed{\widehat{F}} \widehat{G}=\frac1{12} \widehat F_{\hat I\hat J \hat K} \widehat{\Gamma}^{\hat I\hat J\hat K} \widehat{G} = 0\,.
\end{equation}
Notice further that $S_{\widehat{E}}$ contains three factors, the first is the spinor counterpart $S_M$ of the adjoint action $M_A{}^B$, the second is a B-field $S_\rho$ shift induced by the two form $\rho$ and the third is the spinor counterpart $S_{\widehat{V}}$  of the $GL(D)$ transformation induced by the vector fields $v_a{}^i$.  Now this last transformation $S_{\widehat{V}}$ carries with it a multiplicative factor of $|\det  v|^{-\frac{1}{2}}$.  This factor combines with the   $|\det   \tilde{v}   |^{-1/2}$ to cancel the same factors coming from the weighting  and pragmatically speaking in the end to pass to the target space it will be sufficient to calculate $S_\rho S_M G_0$.  Where the context is clear we shall not overcrowd and already burdensome notation with the subscript $G_0$ and understand the push to the generalised tangent space in the above sense.    

\subsection{Fourier-Mukai transformation}
An alternative approach to study the transformation of R/R fluxes is a Fourier-Mukai transformation. This idea  was already applied to Abelian \cite{Hori:1999me} and non-Abelian T-duality \cite{Gevorgyan:2013xka}. Here we show that our previous results allow us to write the R/R flux transformations also in terms of a Fourier-Mukai transformation for the full Poisson-Lie T-duality. Especially, we give an explicit construction for the gauge invariant flux $\omega$ of the topological defect mediating the transformation.

Let us first set up our notation. We have two (pseudo)-Riemannian target spaces $\mathcal{M}$ and $\widetilde{\mathcal{M}}$ which are connected by Poisson-Lie T-duality. 
 We are not restricted to the cases where there are two maximally isotropic subgroups in a single decomposition of $\frak{d}$, one could imagine taking an algebra $\frak{d}$ and by performing global $O(D,D)$ rotation making two different Manin quasi-triple decompositions.
Both target spaces are $D$-dimensional and we denote their coordinates as $x^i$ and $\tilde x^{\tilde i}$, respectively. Furthermore, their metrics, $g_{ij}$ and $\tilde g_{\tilde i\tilde j}$, are used to define a Hodge star on both of them\footnote{We use the explicit expression \begin{equation} (\star A)_{i_1 \dots i_p} = \frac1{\sqrt{| \det g|} (D-p)!} g_{i_1 j_1} \dots g_{i_p j_p} \epsilon^{k_{p+1} \dots k_D j_1 \dots j_p} A_{k_{p+1} \dots k_D} \end{equation} for the Hodge star with $\epsilon^{1 2 \dots D} = 1$.  In this section we chose to restore "upstairs" positions for the indices $\tilde{x}^{\tilde{i}}$ and "downstairs" for dual algebra generators $\tilde{T}_{a}$ -- this is so as not to interfere with the standard notation for differential forms.  }. We are interested in their R/R flux $F_{(p)}$ and $\widetilde{F}_{(p)}$. They are governed by the self-duality conditions
\begin{equation}
  F_{(p)} = (-1)^{\frac{(D-p)(D-p-1)}2} \star F_{(D-p)} \quad 
\end{equation}
and the same for $\widetilde{F}_{(p)}$. These fluxes can be related by the Fourier-Mukai transformation
\begin{equation}
  \widetilde{\mathscr{F}} = \frac1{V(\mathcal{M})} \int_{\mathcal{M}} \mathscr{F} \wedge e^{\omega}\,,
\end{equation}
Here, $V(\mathcal{M})$ denotes the volume of $\mathcal{M}$, which arises after integration over the volume form $v(\mathcal{M}) = \sqrt{|\det g|} d x^1 \wedge \dots \wedge d x^D$. The remaining, essential ingredient in the equation is the two-form $\omega$. In order to fix this form, we need to remember how Poisson-Lie T-duality works for the R/R fluxes in DFT. $\widehat{F}$ and $\widetilde{\widehat{F}}$ are represented by Majorana-Weyl spinor of O($D$,$D$). Using the generalised frames field $\widehat{E}_A{}^{\tilde{\hat I}}$ and $\widetilde{\widehat{E}}_A{}^{\hat I}$ of the two backgrounds we can write down the O($D$,$D$) transformation
\cite{Hassler:2017yza}
\begin{equation}
  \widehat{O}^{\tilde{\hat I}}{}_{\hat J} = \widetilde{\widehat{E}}_A{}^{\tilde{\hat I}} \widehat{E}^A{}_{\hat J}
\end{equation}
relating these two spinors. It acts as 
\begin{equation}\label{eqn:FMspinorversion}
  \widetilde{\widehat{F}} = \sqrt{|\det \tilde e^a{}_{\tilde i} e_a{}^j|} S_{\widetilde{B}} S_{\widehat{O}} S_{-B} \widehat{F}\,.
\end{equation}
A canonical way to parameterise the O($D$,$D$) element $\widehat{O}$ is
\begin{equation}
  \widehat{O}^{\tilde{\hat I}}{}_{\hat J} = \begin{pmatrix}
    r_{\tilde i}{}^j + b_{\tilde i\tilde k} r^{\tilde k}{}_l \beta^{lj} & 
    b_
    {\tilde i\tilde k} r^{\tilde k}{}_j \\
    r^{\tilde i}{}_k \beta^{kj} & r^{\tilde i}{}_j
  \end{pmatrix}\,.
\end{equation}
It allows us to directly identify $\omega$ with
\begin{align}
  \omega &= -\frac12 \Tr \log( r_a{}^b ) + B + \frac12 \beta_{ij} d x^i \wedge d x^j - \frac12 b_{\tilde i\tilde j} d \tilde x^{\tilde i} \wedge d \tilde x^{\tilde j} - \tilde B - r_{\tilde i j} d \tilde x^{\tilde i} \wedge d x^j \nonumber \\
  &= \omega_{(0,0)} + \omega_{(2,0)} + \omega_{(0,2)} + \omega_{(1,1)} \,. \label{eqn:omegaFM}
\end{align}
In the first line, we lowered the two indices of $\beta^{ij}$ with the metric $g_{ij}$ and the same for the second index of $r_{\tilde i}{}^j$. Additionally, we denote a $p$-form on $\mathcal{M}$ and a $q$-form on $\widetilde{\mathcal{M}}$ as $\omega_{(p,q)}$. As we will see later, the contribution $\omega_{(0,0)}$, which depends on $r_a{}^b =\tilde{ e}_a{}^{\tilde i} r_{\tilde i}{}^j e^b{}_j$, vanishes if we have a Manin pair and is only relevant for Manin quasi-pairs. Finally, $B$ and $\tilde B$ are the $B$-fields on the target space and its dual.

To show that the expression presented in eq.~\eqref{eqn:omegaFM} is indeed the correct form of $\omega$, we calculate 
\begin{equation}\label{eqn:FMlinear}
  \widetilde{\mathscr{F}} =
    \frac1{V(\mathcal{M})} \int_{\mathcal{M}} \left( F_{(D)} \wedge ( 1 + \omega_{(0,0)} + \omega_{(0,2)} ) + F_{(D-1)} \wedge \omega_{(1,1)} + F_{(D-2)} \wedge \omega_{(2,0)} + \dots \right)
\end{equation}
up to the linear order in $\omega$ and compare it with the DFT result. We have to take into account the two properties
\begin{equation}
  \star 1 = v(\mathcal{M})  \quad \text{and} \quad
  \star F_{(p)} \wedge \varphi_{(p)} = 
  (-1)^{(D-p) p} F_{i_1 \dots i_p} \varphi^{i_1 \dots i_p} v(\mathcal{M}) \,,
\end{equation}
of the Hodge star. They allow to simplify each term appearing in the expansion \eqref{eqn:FMlinear} to
\begin{equation}
  F_{(D-p)} \wedge \omega_{(p,q)} = s (-1)^{\frac{p(p+1)}2} \frac1{q!} F_{i_1 \dots i_p} \omega^{i_1 \dots i_p}{}_{\tilde j_1 \dots \tilde j_q} d \tilde x^{\tilde j_1} \wedge \dots \wedge d \tilde x^{\tilde j_q}
\end{equation}
where $s$ is the signature of the metric. Applying this relation, we find
\begin{align}
  \widetilde{\widehat{F}} &= s \left( F \left(1 - \frac12 \Tr \log (r_a{}^b) \right) -
    \frac12 F (b_{\tilde i\tilde j} +   B_{\tilde i\tilde j} ) d \tilde x^{\tilde i} \wedge d \tilde x^{\tilde j} +
  F_i r_{\tilde j}{}^i d \tilde{x}^{\tilde j} - F_{ij} ( B^{ij} + \beta^{ij} ) + \dots \right)\,.
\end{align}
Note that this relation crucially relies on the assumption that the R/R fluxes admit Poisson-Lie symmetry. Otherwise we would not be able to perform the integration and cancel the volume factor in front of the integral. One can check that performing the spinor transformation \eqref{eqn:FMspinorversion} gives exactly the same result. Thus the ansatz \eqref{eqn:omegaFM} is correct.

If we specialise to Poisson-Lie T-duality on a Drinfel'd double, using the generalised frame field \eqref{eqn:genframefield}, we find
\begin{equation}
  b_{\tilde i\tilde j} = 0\,, \quad
  \beta^{ij} = e_c{}^i  ( \Pi^{cd} -\tilde \Pi^{cd} ) e_d{}^j \quad \text{and} \quad
  r_{\tilde i}{}^j = \tilde e^a{}_{\tilde i} e_a{}^j\,.
\end{equation}
The last equation implies that $r_a{}^b = \delta_a^b$ and thus $\omega_{(0,0)}$=0. Writing furthermore the  metric on $\mathcal{M}$ as $g_{ij} = e^a{}_i g_{ab} e^b{}_j$, we obtain
\begin{equation}
  \omega = \tilde B - B - \frac12 g_{ac} ( \Pi^{cd} -\tilde \Pi^{cd} ) g_{db} e^a \wedge e^b - g_{ab} \tilde e^a \wedge e^b\,.
\end{equation} 
For the case of non-Abelian T-duality, $\Pi^{ab}=0$, $\tilde \Pi^{ab} =- f^{ab}{}_c x^c$ and $g_{ab}$ is constant. Then the equation for $\omega$ above simplifies to
\begin{equation}
  \omega = \tilde B - B - \frac12 g_{ac} g_{bd} f^{cd}{}_e x^e e^a \wedge e^b - g_{ab} d x^a \wedge e^b
\end{equation}
and matches the result in \cite{Gevorgyan:2013xka}.
 
\section{Application to integrable deformations}\label{sec:intdef}
 In the following we will give examples of how the formalism described in this paper can be applied to $\cal E$-models. In particular we will show how to recover the R/R-sector and dilaton completing the SUGRA embedding for the $\eta$- and the $\lambda$-models.
 
In the first subsection we study these theories at the level of the \DFTwzw defined on $\mathdsl D$ and then show that we recover the conventional target space descriptions on ${\cal M}=   \mathdsl D /\widetilde{H}$. Whilst the solutions presented here are not new to the literature they serve to demonstrate all the features we have described thus far.

\subsection{Deformations based on the (m)CYBE}
Each solution, ${\cal R}$, of the mCYBE on $
 \frak{g}$ gives rise to a canonical group manifold $\mathdsl D= \frak{g}\oplus\frak{g}_{\cal R}$ as described in the Appendix \ref{app:alg}.  The structure constants of $\mathdsl D$ are related to those of $\mathfrak{g}$ (denoted by $f_{ab}{}^c$) according to  
 \begin{equation}\label{eqn:DdfromR}
  F_{abc} = 0 \,,\quad 
  F_{ab}{}^c = f_{ab}{}^c\,, \quad
  \widetilde{F}^{ab}{}_c =  R^{ad} f_{cd}{}^b - R^{bd} f_{cd}{}^a = \tilde f^{ab}{}_c\,, \quad
  \quad \widetilde{F}^{abc} = 0\,.
\end{equation}
 For the YB-deformations described by the action \eqref{eq:etaact}, the generalised metric reads
\begin{equation}\label{eqn:genmetricetadef}
  \mathcal{H}^{AB} = \begin{pmatrix} \eta \kappa_{ab} & - \eta \kappa_{ac} R^{cb} \\
    \eta R^{ac} \kappa_{cb} & \frac{\displaystyle \kappa^{ab}}{%
    \displaystyle\eta}  - \eta R^{ac}\kappa_{cd} R^{db}
  \end{pmatrix}\, ,
\end{equation}
in which $\kappa$ is the Cartan-Killing form on $\frak{g}$.  Here $\eta$ can be considered a deformation parameter.
  One can simplify the form of $\mathcal{H}^{AB}$ considerably by performing the O($D$,$D$) transformation
\begin{equation}
  \mathcal{O}_{\breve A}{}^B = \begin{pmatrix} \sqrt{\eta} \delta^a{}_b & - \sqrt{\eta} R^{ab} \\
    0 & \frac1{\sqrt{\eta}} \delta_a{}^b 
  \end{pmatrix}{}\, .
\end{equation}
This leaves $\eta_{AB}$ invariant and gives rise to
\begin{equation}\label{eqn:widetildeH}
{\breve{\mathcal{H}}}^{\breve A\breve B} = \mathcal{O}^{\breve A}{}_{ C} \mathcal{H}^{ C D} \mathcal{O}^{\breve B}{}_{ D} =
    \begin{pmatrix} \kappa_{ab} & 0 \\ 0 & \kappa^{ab} \end{pmatrix}\, .
\end{equation}
The transformed components of the structure coefficients become, after using the mCYBE,
\begin{equation}\label{eqn:DDFtilde}
\breve{F}_{abc} = 0 \,, \quad
 \breve{F}_{ab}{}^c = \frac1{\sqrt{\eta}} f_{ab}{}^c \,, \quad
\breve{\widetilde{F}}{}^{ab}{}_c = 0  \,, \quad
\breve{\widetilde{F}}{}^{abc} = \eta^{3/2} c^2  \kappa^{ad} \kappa^{be} f_{de}{}^c\,.
\end{equation}
Notice here we need not specify the value of $c$, and the following considerations will hold for all cases.   The generalised curvature capturing the field equations for the metric and the B-field in this rotated frame   reads
\begin{equation}
  \breve{\mathcal{R}}^{\breve A\breve B} = \frac{h^\vee (1 - c^2 \eta^2)^2}{4 \eta} 
    \begin{pmatrix}
      \kappa_{ab} & 0 \\ 0 & - \kappa^{ab}
    \end{pmatrix}\,.
\end{equation}
After transforming back  to the original frame, we obtain
\begin{equation}
  \mathcal{ R}^{ A B} = \mathcal{O}_{ \breve C}{}^A \breve{\mathcal{R}}^{\breve C\breve D} \mathcal{O}_{\breve D}{}^B  
  =
    \frac{h^\vee (1 - c^2 \eta^2)^2}4 
    \begin{pmatrix}
      \kappa_{ab} & - \kappa_{ac} R^{cb} \\
      R^{ac} \kappa_{cb} & - \frac{\kappa^{ab}}{\eta^2} - R^{ac} \kappa_{cd} R^{db}
    \end{pmatrix} \, .
\end{equation}
This is consistent with the renormalisation of the sigma-model eq.~\eqref{eq:etaact} given by  \begin{equation}
  \frac{\partial \mathcal{H}^{AB}}{\partial \log\mu } = \mathcal{R}^{AB}
  \quad \text{with} \quad
  \frac{\partial\eta}{\partial\log\mu} = \frac{h^\vee (1 - c^2 \eta^2)^2}4\,.
\end{equation}
 
In the rotated  frame   the generalised Ricci scalar is quite easily calculated to be
\begin{equation}\label{eqn:RicciEta}
  \mathcal{R} =  \frac{1}{6} \left(\eta^3 c^4 - 6 \eta c^2 - 3 \eta^{-1} \right) h^\vee \dim\, \mathfrak{g}\,.
\end{equation}
There is no solution for $\mathcal{R}=0$, the dilaton equation of \DFTwzw{} which holds for arbitrary $\eta$ and $c$. However we may extend our considerations to include a direct sum of simple algebras
\begin{equation}\label{eqn:semisimpleg}
  \mathfrak{g} = \mathfrak{g}_1 \oplus \dots \oplus \mathfrak{g}_N\,.
\end{equation}
In this case, we can choose a different scaling for the inner product imposed on each simple factor $\mathfrak{g}_i$:
\begin{equation}
  \kappa_{ab}^{(i)} = - \frac1{2 h^\vee \alpha_i} f^{(i)}_{ac}{}^d f^{(i)}_{bd}{}^c\,.
\end{equation}
In this way we will be able to engineer a cancelation of contributions to the curvature coming from each group factor, as is typical between $AdS$ and internal factors of supergravity solutions. 
In principle we could have done this already for the simple case, but there such a scaling amounts to an trivial overall rescaling of the solution. For $N$ simple factors, we have $N-1$ additional degrees of freedom for which the dilaton field equation ${\cal R}=0$ implies  the constraint,
\begin{equation}\label{eqn:eomdilaton}
  \sum_{i=1}^N \alpha_i h_i^\vee \dim \mathfrak{g}_i = 0 \, . 
\end{equation}

This direct sum of algebras is however insufficient to   solve $\mathcal{R}_{AB}=0$. Hence, we conclude that in general there no setup which can solve the field equations without any contributions from the R/R sector. In order for the R/R sector to compensate the NS/NS contribution, we require     (again in the rotated frame where equations are simplified),
\begin{equation}\label{eqn:EOMintdef}
  \breve{\mathcal{H}}^{\breve A}{}_{\breve C} \breve{\mathcal{R}}^{\breve C\breve B} = 
  \frac{h^\vee (1 - c^2\eta^2)^2}{4 \eta} \begin{pmatrix}
    0 & - \delta_a^b \\ \delta^a_b & 0
  \end{pmatrix}
  = - \frac18 {{\breve G}}^T C \Gamma^{\breve A\breve B} {{\breve G}} \, . 
\end{equation}
 The field ${{\breve G}}$ is an eigenvector of $\breve{\mathcal{K}}$ with eigenvalue $-1$ as required by \eqref{eqn:dualityconstr} and of definite chirality. We discuss this condition in the following.

As the NS/NS sector, the R/R sector should also exhibit Poisson-Lie symmetry and in particular eq. \eqref{eq:Geq2} has to hold
\begin{equation}
 \breve{\mathdsl F}_{\breve A\breve B\breve C} \Gamma^{\breve A\breve B\breve C} \breve{{G}} = 0\,,
\end{equation}
at least if there are no sources like D-branes. Expanding this constraint into components, we obtain
\begin{equation}\label{eqn:bianchiintegrable}
\left( 3 f_{ab}{}^c \Gamma^{ab} \Gamma_c + c^2 \eta^2 \tilde f^{abc} \Gamma_{abc} \right) \breve{G} = 0\,.
\end{equation}
 
At this stage we should like to be explicit about solutions for  $\breve{G}$.  To do so we found it convenient to recast our manipulations in terms of an $O(D)$ Drinfel'd $\breve{{\bf G}}^{\underline{\alpha\beta}}$ which can be related to  $\breve{G}$ by a vectorisation map $\breve{G} = \VEC \breve{{\bf G}}$.  The presentation of this somewhat technical procedure is detailed below and can be skipped on a first reading jumping instead to the explicit solution in the case of an example $\frak{g}= \frak{sl}(2)\oplus \frak{su}(2)$. 

\subsubsection*{Bispinorisation}
 The strategy will be to find a representation of $\breve{G}$ such that the self-duality and  chirality constraints are automatically implemented and the only thing that remains to be taken care of is \eqref{eqn:bianchiintegrable}. We   introduce the $\gamma$-matrices for the $D$-dimensional target space obeying the Clifford algebra,
\begin{equation}
  \{\gamma_a, \gamma_b\} = 2 \kappa_{ab} \ . 
\end{equation}
Assuming that $D$ is even (which in our cases it shall be) they furthermore can be brought into the form,
\begin{equation}
  (\gamma_a)^{\underline{\alpha}}{}_{\underline{\beta}} = \begin{pmatrix}
    0 & (\gamma_a)_{\alpha\beta} \\
    (\gamma_a)^{\alpha\beta} & 0
  \end{pmatrix}\,, \quad \text{with} \quad
  (\gamma_{D+1})^{\underline{\alpha}}{}_{\underline{\beta}} = \begin{pmatrix}
    \delta_\alpha^\beta & 0 \\ 0 & - \delta^\alpha_\beta
  \end{pmatrix} \quad \text{and} \quad
  C_{\underline{\alpha\beta}} = \begin{pmatrix}
    0 & \delta^\alpha_\beta \\
    - \delta_\alpha^\beta & 0
  \end{pmatrix}
\end{equation}
denoting the chirality and charge conjugation matrices. We express the $2^D$ components of $\breve{G}$ as a bispinor $\breve{{\bf G}}^{\underline{\alpha\beta}}$, where $\underline{\alpha}$, $\underline{\beta}$, \dots are Dirac spinor indices which run from 0 to $D$. To get back and forth between these two representations, we use the vectorization
\def\GG{\breve{{\bf G}}}
\begin{equation}
\breve{G} = \VEC(\GG ) = \begin{pmatrix} \GG^{\underline{00}} ,&
    \dots , & \GG^{\underline{D0}}, & \GG^{\underline{D1}} ,& \dots, &
  \GG^{\underline{DD}} \end{pmatrix}^T\,.
\end{equation}
The O($D$,$D$) $\Gamma$-matrices can now be written as
\begin{equation}
  \Gamma_a  = \frac1{\sqrt{2}} ( \gamma_a \otimes 1 - i \gamma_{D+1} \otimes \gamma_a ) 
    \quad \text{and} \quad
  \Gamma^a  = \frac1{\sqrt{2}} ( \gamma^a \otimes 1 + i \gamma_{D+1} \otimes \gamma^a )   \ .
\end{equation}
At first glance this new representation looks somewhat unwieldy. However, it has the advantage that the operator $\breve{\mathcal{K}}$ and the chirality $\Gamma_{2D+1}$ have a very convenient form:
\begin{align}
   \breve{\mathcal{K}}     = - ( 1 \otimes \gamma_{D+1} )    \, , \quad 
  \Gamma_{2D+1}     = (\gamma_{D+1} \otimes \gamma_{D+1} )  \,.
\end{align}
Remember that $\breve{G}$ has to be an eigenvector of  $\breve{\mathcal{K}}$  with eigenvalue $-1$. Furthermore, it has to have a fixed chirality which also makes it an eigenvector of $\Gamma_{2D+1}$. The eigenvalue under this operator decides whether we are capturing a type IIA or a IIB theory. In the bispinor representation solving these two constraints requires just to pick a particular subblock of $\GG^{\underline{\alpha\beta}}$, namely
\begin{equation}
  \begin{tabular}{lcccc}
    block & $\GG^{ \alpha\beta }$ & $\GG^{ \alpha }{}_{ \beta }$ & $\GG_{ \alpha }{}^{ \beta }$ &
      $\GG_{ \alpha\beta }$ \\
    eigenvalue $\mathcal{K}$ & $+1$ & $-1$ & $+1$ & $-1$\phantom{\,.} \\
    eigenvalue $\Gamma_{2D+1}$ & $+1$ & $-1$ & $-1$ & $+1$\,.
  \end{tabular}
\end{equation}
The condition encoding the Poisson-Lie symmetry \eqref{eqn:bianchiintegrable} reads 
\begin{equation}
  \left[ (3+c^2 \eta^2) ( u \otimes 1 + i \gamma_{D+1} \otimes u ) + 3 (1-c^2 \eta^2) ( \gamma^a \otimes u_a
    + i u_a \gamma_{D+1} \otimes \gamma^a ) \right]   \breve{G}  = 0 \, ,
\end{equation}
with
\begin{equation}
  u= f_{abc} \gamma^{abc} \quad \text{and} \quad u_a = f_{abc} \gamma^{bc}\,.
\end{equation} 
Note that because $\breve{G}$ has to be an eigenvector of both $\breve{\mathcal{K}}$ and $\Gamma_{2D+1}$, two combinations of the terms in this constraint  have to vanish individually. This leaves us with the two equations 
\begin{align}
  \left[ (3+c^2 \eta^2) ( u \otimes 1 ) + 3 (1-c^2 \eta^2) (\gamma^a \otimes u_a) \right]  \breve{G} 
    &= 0  \ , \nonumber \\ 
  \left[ (3+c^2 \eta^2) ( 1 \otimes u ) + 3 (1-c^2 \eta^2) (u_a \otimes \gamma^a) \right]  \breve{G} 
    &= 0 \,.
\end{align}
In the following, we do not want to discuss all solutions of these equations, but only the ones that have a chance to give rise to backgrounds which solve the field equations. To this end, we restrict our attention to $\breve{G}$'s that are invariant under the action of $\mathfrak{g}$. More explicitly we impose  
\begin{equation}\label{eqn:invariantG}
  ( u_a \otimes 1 + 1 \otimes u_a )  \breve{G}  = 0\,.
\end{equation}
Using this identity, \eqref{eqn:bianchiintegrable} simplifies to
\begin{equation}\label{eqn:c^2ne0constr}
  c^2 \eta^2 ( u \otimes 1 )  \breve{G}  = 0 \quad \text{and} \quad
  c^2 \eta^2 ( 1 \otimes u )  \breve{G}  = 0\,.
\end{equation}
In particular, for the $\beta$-deformations for which $c^2=0$, the condition \eqref{eqn:invariantG} is sufficient and in all other cases, we have to additionally impose the two constraints above \eqref{eqn:c^2ne0constr}. 

In order to see what singles out these solutions, we have to take a closer look at the left hand side of R/R corrected field equation \eqref{eqn:EOMintdef}. To satisfy this equation the contributions from $\Gamma^{ab}$ and $\Gamma_{ab}$ to the left hand side vanish completely. Therefore, we just have to calculate the remaining:
\begin{equation}\label{eqn:R/Rcontrbispinor}
  \breve{G}^T C \Gamma^a{}_b \mathcal{K} \breve{G} = 
   \breve{G}^T ( \gamma_{D+1} \gamma^1 \gamma^a \otimes \gamma_{D+1} \gamma^1 \gamma_b )
   \breve{G}  = \pm \Tr \left(  (\GG\gamma^1 \gamma^a)^T \gamma^1 \gamma_b \GG \right) \,.
\end{equation}
This equation assumes a target space with Minkowski signature\footnote{For an Euclidean spacetime, we would just have to drop the $\gamma^1$'s.} with the time direction matching $\gamma^1$ and the $+$/$-$ depends on whether $\breve G$ is chiral/anti-chiral.  Here, we have used the charge conjugation matrix on $O(D,D)$ spinors given by
\begin{equation}
  C = i \gamma_{D+1} \gamma^1 \otimes \gamma_{D+1} \gamma^1 \, .
\end{equation}
 For a simple Lie group the Killing metric is up to a scaling factor the only invariant bilinear form. But this implies that because $ \breve{G} $ is invariant, the left hand side of \eqref{eqn:R/Rcontrbispinor} has to be proportional to $\delta^a_b$. So the only thing we have to fix is the normalization of $\breve{G}$. According to \eqref{eqn:EOMintdef} it becomes,
\begin{equation}\label{eqn:normalizationG}
\Tr \left(  (\GG\gamma^1 \gamma^a)^T \gamma^1 \gamma_a \GG \right) = \mp
    \frac{h^\vee ( 1 - c^2 \eta^2 )^2 \dim \mathfrak{g}}\eta \,. 
\end{equation}
If we have more than one simple factor, like in \eqref{eqn:semisimpleg}, there are additional constraints on $\breve{G}$:
\begin{equation}
  \Tr \left(  (\GG \gamma^1 \gamma^a)^T (P_i)^b{}_a \gamma^1 \gamma_b \GG  \right) = \mp
    \frac{\alpha_i h^\vee_i ( 1 - c^2 \eta^2 )^2 \dim \mathfrak{g}_i }\eta \,, 
\end{equation}
where $(P_i)^b{}_a$ denotes a projector on the $i$th simple factor.

\subsubsection*{Example}
Let us illustrate this procedure for deformations of AdS$_3 \times $S$^3$. In this particular case, the two relevant Lie algebras are
\begin{equation}
  \mathfrak{g}_1 = \mathfrak{sl}(2) \quad \text{and} \quad
  \mathfrak{g}_2 = \mathfrak{su}(2)
  \quad \text{with} \quad h^\vee_1 = h^\vee_2 = 2 \,, \quad
  \dim \mathfrak{g}_1 = \dim \mathfrak{g}_2 = 3 \,.
\end{equation}
In order to solve the field equation for the dilaton \eqref{eqn:eomdilaton}, we choose
\begin{equation}
  \alpha_1 = 1 \quad \text{and} \quad \alpha_2 = -1\,.
\end{equation}
This results in  $\kappa_{ab}$ of Minkowski signature, as required to describe AdS$_3 \times $S$^3$. A compatible R/R sector arises from \eqref{eqn:invariantG}. The corresponding R/R bispinor has the two solutions\footnote{We use the chiral $\gamma$-matrices
  \begin{align}
    (\gamma_1)_{\alpha\beta} &= \begin{pmatrix} 0 & \sigma_1 \\ - \sigma_1 & 0 \end{pmatrix}\,, &
    (\gamma_2)_{\alpha\beta} &= \begin{pmatrix} i \sigma_2 & 0 \\ 0 & -i \sigma_2 \end{pmatrix}\,, &
    (\gamma_3)_{\alpha\beta} &= \begin{pmatrix} 0 & -I \\ I & 0 \end{pmatrix}\,, \nonumber \\
    (\gamma_4)_{\alpha\beta} &= \begin{pmatrix} 0 & -i \sigma_2 \\ i \sigma_2 & 0 \end{pmatrix}\,, &
    (\gamma_5)_{\alpha\beta} &= \begin{pmatrix} 0 & i \sigma_3 \\ - i \sigma_3 & 0 \end{pmatrix}\,, &
    (\gamma_6)_{\alpha\beta} &= \begin{pmatrix} - \sigma_2 & 0 \\ 0 & \sigma_2 \end{pmatrix}\,.  \nonumber
\end{align}
They are conjugated by $\gamma_a^{\alpha\beta} = \epsilon^{\alpha\beta\gamma\delta} (\gamma_a)_{\gamma\delta}$ and give rise to the Killing metric $\delta_{ab} = \diag(-1, 1, 1, 1, 1, 1)$. $\sigma_i$ denotes the three Pauli matrices with $\sigma_i^2 = 1$ and $\epsilon^{\alpha\beta\gamma\delta}$ is totally anti-symmetric with $\epsilon^{1234}=1$.}
\begin{equation}
  \GG^\alpha{}_\beta \sim \diag( 1, 1, 1, 1 )
    \quad \text{and} \quad
   \GG_{\alpha\beta} \sim \diag( 1, -1, -1, 1)
\end{equation}
after restricting to the components of $\GG$ with $\mathcal{K}$ eigenvalue $-1$. Only the second one solves the additional constraint \eqref{eqn:c^2ne0constr}, which is required for $c^2\ne0$. Furthermore, the first solution can not be normalized such that \eqref{eqn:normalizationG} is satisfied for both the $\mathfrak{sl}(2)$ and $\mathfrak{su}(2)$ factors. Thus, we conclude that for arbitrary $c$, there is only one R/R field configuration
\begin{equation}
 \GG_{\alpha\beta} = \frac1{\sqrt{2 \eta}} (1-c^2 \eta^2) \diag(1, -1, -1, 1)
\end{equation}
that admits Poisson-Lie symmetry and in connection with the NS/NS sector solved all field equations. An alternative way to write this solution is
\begin{equation}\label{eqn:spinorads3xs3}
   \breve G  = \frac{1-c^2 \eta^2}{12 \sqrt{\eta}} f_{abc} \Gamma^{abc} |0\rangle\,,
\end{equation}
where $|0\rangle$ denotes the vector which is annihilated by all $\Gamma_a$.    
\subsection{$\eta$-deformation redux}\label{sec:etasol}
In this case the target space ${\cal M} = {\mathdsl D}/\widetilde{H}$ is equivalent to a group manifold $H$ whose algebra corresponds to the direct sum  of algebras introduced in eq.~\eqref{eqn:semisimpleg}.  Parametrising ${\cal M}$ by a group element $g\in H$ (with $e$ and $v$ corresponding left and right Maurer-Cartan forms respectively and $M_A{}^B \mathdsl T_B = g\mathdsl T_A g^{-1}$) we have from theorem  \ref{thm:one}  the generalised frame field
\begin{equation}\label{eq:GFFeta}
  \widehat{E}_A{}^{\hat I }= \begin{pmatrix} 
    e^a{}_i & \Pi^{ab} e_b{}^i\\
    0       &  e_a{}^i
  \end{pmatrix} \, , 
\end{equation}
 in which we recall $ \Pi^{ab} = M^{ac} M^b{}_c$.   The target space  metric and the $B$-field  are readily extracted from   $\widehat{\mathcal{H}}^{\hat I\hat J} = \widehat{E}_A{}^{\hat I} \mathcal{H}^{AB} \widehat{E}_B{}^{\hat J}$ and read\footnote{The overall factor of $\eta$ in front of the metric may look unfamiliar but the reader should recall that the normalisation $\tilde{k}$ of the ${\cal E}$-model, in which the DFT equations are  perturbative is related to the normalisation of the sigma-model by a corresponding factor $\frac{1}{\eta}$.  } 
  \begin{equation}
  \begin{aligned}
ds^2 &=  g_{ab} e^a e^b = \eta \kappa_{ab} \, e^a \otimes e^b - \frac{\eta^3}{1+\eta^2} \beta^{ab} \beta^{cd} \kappa_{bd} \kappa_{ae} \kappa_{cf} \, e^e   \otimes e^f   \ , \\
  B &= \frac{\eta^2}{2(1+\eta^2)} \left( \kappa_{ac} \beta^{cd} \kappa_{db} \, e^a \wedge e^b \right)\,.
  \end{aligned}
\end{equation}
 Here we have used the $\beta$-parametrisation of the generalised metric  
\begin{equation}\label{eqn:paramHIJ}
  \widehat{\mathcal{H}}^{\hat I\hat J}  =
  \begin{pmatrix} 
    e^a{}_i & 0 \\
    0 & e_a{}^i \end{pmatrix}
  \begin{pmatrix}
    \tilde g_{ab} & \tilde g_{ac} \beta^{cb} \\
    - \beta^{ac} \tilde g_{cb} & \tilde g^{ab} - \beta^{ac} \tilde g_{cd} \beta^{db}
  \end{pmatrix}\begin{pmatrix} 
    e^b{}_j & 0 \\
    0 & e_b{}^j \end{pmatrix} \, , 
\end{equation}
of the generalised metric for which 
\begin{equation}
  \tilde g_{ab} = \eta \delta_{ab} \quad \text{and} \quad
  \beta^{ab} = \Pi^{ab} - R^{ab} = - M^a{}_c M^b{}_d R^{cd}\,.
\end{equation}
In this parametrisation  the metric and the $B$-field in flat indices arise from inverting
\begin{equation}\label{eqn:extractetagandB}
  \left( \frac1\eta \kappa^{-1} - \beta \right)^{-1}_{ab} = g_{ab} - B_{ab}\,,
\end{equation}
whereas the curved version are obtained after contraction with $e^a{}_i$. 
A comparison with the action \eqref{eq:PLact} gives rise to
\begin{equation}
  (G_0 - B_0)^{-1\,ab} = \tilde g^{ab} + \beta^{ab} - \Pi^{ab} = \frac1\eta \kappa^{ab} - R^{ab}
\end{equation}
which is equivalent to \eqref{eqn:E0eta}.
The dilaton is determined by the Poisson-Lie condition  \eqref{eqn:conddilaton}  and has to have the form
\begin{equation}
  \phi = \phi_0 + \frac14 \log |\det g_{ij}| - \frac12 \log |\det e^a{}_i| = \phi_0 + \frac14 \log |\det g_{ab} |\,,
\end{equation}
where according to \eqref{eqn:conddilaton}, $\phi_0$ is a free constant.  

Since in general $\widetilde{H}$ will be non-unimodular we will have solutions of modified type II SUGRA encoded in the DFT vector  
 \begin{equation}
 {\bf X}^{\hat I} = \frac{1}{2}\begin{pmatrix} R^{bc} f_{bc}{}^a v_a{}^i \\ 0 \end{pmatrix}\,.
\end{equation}
from which the Killing vector $I$ of modified supergravity is 
\begin{equation}\label{eqn:Iieta}
  I^i = - \frac12 R^{bc} f_{bc}{}^a v_a{}^i = -\frac12 \beta^{bc} f_{bc}{}^a e_a{}^i\, .
\end{equation}
To complete the NS sector specification of the modified type II SUGRA one also needs  the one-form $Z$ defined in eq.~\eqref{eq:defZ}   which gives rise to
\begin{equation}
  Z = \frac14 g^{ab} d g_{ab} + \iota_I B \,.
\end{equation}
Finally, we have to obtain the R/R fluxes. To this end, we  begin with a solution  $\breve{G}$ of Poisson-Lie condition \eqref{eqn:bianchiintegrable} evaluated in the simplified rotated frame and then   calculate
\begin{equation}\label{eqn:betatrspinor}
  \sum_{p=1} \frac{1}{p! 2^{p/2}} \widehat{G}^{(p)}_{a_1\dots a_p} \Gamma^{a_1\dots a_p} = 
    \sqrt{\eta} S_\beta {\breve G} = 
  \sqrt{\eta} \exp\left( \frac14 \beta^{ab} \Gamma_{ab} \right)  \breve{G} 
\end{equation}
in flat indices and again contract with $e^a{}_i$ to get the curved versions. Converting this into a polyform one can construct the fluxes   $\widehat{ \cal F} = e^\phi  e^{-B} \widehat G$ which obey the flux equations and Bianchi identities $ {\bf d}  \widehat{ \cal F} =0$ with ${\bf d}$ the modified exterior derivative of eq.~\eqref{eq:exder}.

An  example is the AdS$_3\times $S$^3$ from the last section for which a realization of the Drinfel'd double provided in appendix~\ref{app:groupparam}.  However there is no need to resort here to an explicit parametrisation since the geometry can be entirely written in terms of  $e^a$ and $\beta^{ab}$ whose exterior derivatives are easily obtained as 
 \begin{equation}
  d e^a = - \frac12 f_{bc}{}^a e^b \wedge e^c \ , \quad  d \beta^{ab} = 2 f_{cd}{}^{[a} \beta^{cb]} e^d \ . 
 \end{equation} 
 The  metric, $B$-field and vector $I$  are already given in the simple forms above and in addition we have that for this example the dilaton is constant and 
 \begin{equation}
  \phi = \log \left( \frac{\eta^{3/2}}{1 + \eta^2} \right) + \phi_0 \ , \quad H  = dB = 0 \ , \quad Z = 0 \ . 
\end{equation}
Evaluating \eqref{eqn:betatrspinor}   for the solution given in eq.~\eqref{eqn:spinorads3xs3}, gives rise to
\begin{align}\label{eq:etaspin}
  \widehat{G}^{(1)} = -\frac{1+\eta^2}{\sqrt{2}}  \beta^{ab} f_{abc} e^c \ , \quad 
  \widehat{G}^{(3)} = \frac{1 + \eta^2}{3 \sqrt{2}} f_{abc} e^a \wedge e^b \wedge e^c\,.
\end{align}
At this stage the preceding discussion establishes that we have a solution of modified supergravity, or rather a six-dimensional truncated version thereof.   As a consistency check  and for completeness we provide details of the uplift to a  full ten-dimensional solution in appendix  \ref{sec:appendixetasol}. 

\subsection{$\lambda$-deformation redux}

We now describe the $\lambda$-model in this framework.  Here the underlying Double is formed from $\frak{d} = \frak{g} + \frak{g}$ with the maximal isotropic subgroup $\tilde{H}$ being the diagonal embedding $G_{diag} \subset \mathdsl D$ (see \ref{app:alg} for further details), whose generators are $\widetilde{T}^a$ in the canonical basis.   In this case however the complementary isotropic does not form a subgroup as can be seen from the structure constants of $\frak{d}$ given in terms of those of $\frak{g}$ by  
\begin{equation}
F_{ab}{}^c = 0 \, ,  \quad      F_{abc}   = \frac{1}{\sqrt{2} } f_{abc} =  \frac{1}{\sqrt{2} } f_{ab}{}^f \kappa_{  fc }\, ,   \quad \widetilde{F}^{abc}=0 \, ,\quad \widetilde{F}^{ab}{}_c = \frac{1}{\sqrt{2} } f^{ab}{}_c    = \frac{1}{\sqrt{2} } \kappa^{a d} \kappa^{b e}    f_{de}{}^f \kappa_{  f c} \, ,
\end{equation}
with others given by symmetry enforced by the ad-invariance of $\eta$.   This algebra admits a $\mathbb{Z}_2$ grading so that ${\cal M} = \mathdsl D  / \widetilde{H} =(G\times G)/G_{\diag}$ is a symmetric space and we can apply the construction of  theorem \ref{thm:one}
 to obtain generalised frame fields that describe the geometry.  To do so requires some care however in the parametrisation of coset representatives that we now explain. 
 
Before doing so let  us mention the alternative route to describe the $\lambda$-deformation  as the Poisson-Lie T-duality of the $\eta$-model and a subsequent analytic continuation. This can be made quite manifest at the level of  ${\cal E}$-models \cite{Klimcik:2015gba} and therefore in \DFTwzw{}. It is worthwhile briefly recasting this argument in the language we use here by identifying a  frame where the double $\frak d = \frak{g}^{\mathbb{C}}$ decomposes, up to an analytic continuation, into $\mathfrak d= \tilde{\mathfrak{h}} \oplus \mathfrak k= \mathfrak g_\mathrm{diag} \oplus \mathfrak g_\mathrm{antidiag}$. Starting from the frame \eqref{eqn:DDFtilde} we perform a rotation
\newcommand{\wideparen}[1]{  #1 }
\begin{equation}\label{eqn:ODDlambda}
  \mathcal{O}_{\wideparen{A'}}{}^{\breve B} = \begin{pmatrix} 0 & \kappa^{ab} \\ \kappa_{ab} & 0 \end{pmatrix}\,,
\end{equation} 
to obtain   structure coefficients $\wideparen{{\mathdsl F}'}_{\wideparen{A'}\wideparen{B'}\wideparen{C'}} = \mathcal{O}_{\wideparen{A'}}{}^{\breve{D}} \mathcal{O}_{\wideparen{B'}}{}^{\breve E} \mathcal{O}_{\wideparen{C'}}{}^{\breve F} \breve{\mathdsl F}_{\breve D\breve E\breve F}$ that have  components
\begin{equation}\label{eqn:framelambda}
  F'_{abc} = - \eta^{3/2}  f_{ab}{}^d \kappa_{dc} \,, \quad
  F'_{ab}{}^c = 0 \,, \quad
  F'{}^{ab}{}_c = \frac1{\sqrt{\eta}} \kappa^{ad} \kappa^{be} \kappa_{cf} f_{de}{}^f \,, \quad 
  F'{}^{abc} = 0 \, .
\end{equation}
One can write down ``generators''  for the commutation relations \eqref{eqn:framelambda},  \begin{equation}
    \wideparen{T}_a = \sqrt{\eta} \lbrace - i t_a , i t_a \rbrace \quad \text{and} \quad
   \wideparen{\widetilde{T}}^a = \frac{ \kappa^{ab} }{\sqrt{\eta}} \lbrace t_b , t_b \rbrace \,, \qquad t_a \in \mathfrak{g}\,,
\end{equation}
however we see that the generators here are not anti-hermitian and hence an analytic continuation
\begin{equation}
  T_a \rightarrow i T_a
\end{equation}
must be taken in order to match (up to scaling) the structure of $\frak{d}= \frak{g}+\frak{g}$.

We now resume the construction of the generalised frame fields. We must set the representative $m$ for coset  $\mathdsl D/\widetilde{H}$.  To be explicit we make the choice of parameterisation of the coset representative\footnote{Here we deviate from  \cite{Klimcik:2015gba} in which the coset representative is chosen as $m = \{g , e\}$, the reason will be that this choice is the one that matches the parametrisations used in theorem \ref{thm:one}.  } $m =\{ \bar{g} , \bar{g}^{-1} \}$  with $\bar{g}\in G$.      However to match directly to the $\lambda$-model of eq.~\eqref{eq:actlambda} which is parametrised by a group element $g$,  a further identification is needed namely that $\bar{g}^2 = g$.    We let  $\bar{e}$, $\bar{v}$, $\bar{D}$ be  the left/right-invariant forms and adjoint action on $\frak{g}$  constructed from $\bar{g}$  which can be related to those constructed from $g$ via,
\begin{equation}\label{bartonobar}
e^{a}{}_i = ( 1 + \bar{D}^{-T} )^a{}_b \bar{e}^b{}_i  \ , \quad    \bar{D}^2 = D \, . 
\end{equation}
The adjoint action of $m$ on $\frak{d}$,   i.e. $m \mathdsl T_A m^{-1} = M_A{}^B\mathdsl  T_B$, is given by
\begin{equation}\label{eqn:lambdaadj}
M_A{}^B = \frac{1}{2} \left(\begin{array}{cc}  ~\kappa^{-1} X^+ \kappa~    &     \kappa^{-1} X^-   \\   X^- \kappa      &  X^+    \end{array}\right) \ , \quad X^\pm = \bar{D} \pm \bar{D}^{-1} \, ,
\end{equation} 
and is easily seen to be an $O(D,D)$ element preserving $\eta_{AB}$. 
 The one-form $dm m^{-1}$ determines a veilbein and $\tilde{\frak{h}}$-valued connection according to  
\begin{equation} 
dm m^{-1}= V^A{}_i \mathdsl T_A dx^i =T_a  v^a{}_i dx^i  + \widetilde{T}^a A_{a i} dx^i\  , 
\end{equation} 
 with 
\begin{equation}
v^a{}_i    = \frac{1}{\sqrt{2}} \left(   \bar v + \bar e \right)^a{}_i  \ , \quad  A_{a i}  =  \frac{1}{\sqrt{2}}\kappa_{ab}    \left(  \bar v - \bar e \right)^b{}_i  \ . 
\end{equation}
From these we can construct a two-form  
\begin{equation}
 \quad \omega^{(2)} =  \frac{1}{2}   v^a{}_i    A_{a j} dx^i \wedge dx^j =   \frac{1}{4}   \left( \kappa \bar{D}^T - \bar{D} \kappa \right)_{ab} \bar{e}^a \wedge \bar{e}^b  \, , 
 \end{equation}  
such that  
\begin{equation}
\begin{aligned}
d\omega^{(2)} 
&=  -   \frac{1}{4}   f_{abc} \left(\bar{e}^a \wedge \bar{e}^b \wedge \bar{v}^c+\bar{v}^a \wedge \bar{v}^b \wedge \bar{e}^c  \right)  \, .
\end{aligned} 
\end{equation} 
In addition there is a globally defined three from 
\begin{equation}
\Omega^{(3) }  =   \frac{1}{12}    \llangle dm m^{-1},[dm m^{-1},dm m^{-1}] \rrangle =  \frac{1}{6}     f_{abc} \bar{e}^a   \wedge  \bar{e}^b \wedge \bar{e}^c  \ , 
\end{equation}
for which {\em locally} we can introduce a suitable potential $d \Omega^{(2)}  = \Omega^{(3)} $. 

Then in theorem \ref{thm:one} we have in combination
\begin{equation}
\rho =  \omega^{(2) }  - \Omega^{(2)}  \rightarrow d\rho = -\frac{1}{12} f_{abc} e^a \wedge e^b \wedge e^c
\end{equation}
giving the three form $H_{WZW}$ that comes from $S_{WZW}[g]$.

We now have all the ingredients to introduce a generalised frame field, itself also an $O(d,d)$ element.  We, as per comment \ref{comm:2}, will strip out the $\bf{H}$-flux contribution given by $\rho$ and consider the frame field  
 \begin{equation}
  \widehat{E}^\prime_A{}^{\hat I} = \frac{1}{2}
  \begin{pmatrix}
    \frac1{\sqrt{2}} \kappa^{-1}( 1  + D )\kappa & {\sqrt{2}} \kappa^{-1}( 1 - D ) \\
    \frac{1}{\sqrt{2}} (1 - D ) \kappa   & \sqrt{2} ( 1 + D )
  \end{pmatrix}\indices{_A^B}
  \begin{pmatrix}
    \;e\;\;\;\; & 0 \\
    \;0 \;\;\;\;&  e^{-T}
  \end{pmatrix}\indices{_B^{\hat I}}\,.
\end{equation}
in which we have massaged the expressions such that they only depend on the left/right-invariant forms and adjoint action on $\frak{g}$  constructed from  $g$ so as to match the parametrisation of  the $\lambda$-model of eq.~\eqref{eq:actlambda}. A useful decomposition of the frame field is given by
\begin{equation}
  \widehat{E}^\prime{}_A{}^{\hat I} = \begin{pmatrix}
    r^{-T} & 0 \\
    0 & r  
  \end{pmatrix} \begin{pmatrix}
    1 & 0 \\
   b & 1
  \end{pmatrix} \begin{pmatrix}
 1 & \beta  \\
    0 & 1
  \end{pmatrix} \begin{pmatrix}
    e   & 0 \\
    0 & e^{-T}  
  \end{pmatrix}
\end{equation}
with
\begin{equation}
  r^{-T} = \frac{1}{2\sqrt{2}} \left( 1 + D^{-T} \right) \ , 
 \quad  
  b  = \frac1{8} ( D^{-1} - D  )\kappa \quad  \text{and} \quad
  \beta  = 2 \kappa^{-1} \frac{1- D}{1+ D}   \,.
\end{equation}
From the definition of the operator ${\cal E}$ in the $\lambda$-model we have the generalised metric in flat space given by  
\begin{equation}\label{eqn:genmetriclambda}
 {\cal H}_{AB} =   \begin{pmatrix}  \epsilon^{- \frac{1}{2} } \kappa^{ab}    &   0     \\
     0   &       \epsilon^{  \frac{1}{2} } \kappa_{ab}    \end{pmatrix}    \ , \quad  \epsilon^{\frac{1}{2}}=  \frac{1- \lambda}{1+\lambda} \ , 
\end{equation}
from which as usual we construct the curved space generalised metric.  
From this the metric and B-field are readily extracted as  
\begin{eqnarray}
\begin{aligned}
  \mathrm{d}s^{2}_\lambda &=  \frac{1}{2}\left( (  O_{g^{-1}} + O_g- 1 )\kappa  \right)_{ab}e^a \otimes e^b \ ,  \label{eq:MetricLamba1}\\
  B_\lambda  &=  B_{\mathrm{WZW}} + \frac{1}{4}  \left( O_{g^{-1}} - O_g \right)_{ab} e^a  \wedge e^b  \ , 
\end{aligned}
\end{eqnarray} 
in which $H_{\mathrm{WZW}} =  \mathrm{d}B_{\mathrm{WZW}}$ and 
\begin{equation}
O_g = (1 -\lambda D )^{-1} \, .
\end{equation} 
As in our previous discussion, and detailed in appendix \ref{app:dilaton}, the PL conditions on the dilaton determine that 
\begin{equation}
\phi_\lambda = \phi_0 -\frac{1}{2} \log \det (1 - \lambda D^{-1} ) \ ,
\end{equation}
in which $\phi_0$ is constant.  This  matches the dilaton obtained due to a Gaussian elimination of fields in  the  construction  of \cite{Sfetsos:2013wia}. 

Since $\frak{g}$ is assumed to be unimodular we have that $\tilde{\frak{h}}$ is also unimodular, and thus we in the situation of  conventional (not modified) supergravity.  What remains is to determine the RR fluxes from the PL conditions. 
 
To be totally analogous with the discussion of the $\eta$ deformations we should actually further perform and $O(D,D)$ rotation 
\begin{equation}
  \mathcal{O}_{\breve A}{}^B = \begin{pmatrix} \epsilon^\frac{1}{4} \delta^a{}_b &  0 \\
    0 & \epsilon^{-\frac{1}{4}} \delta_a{}^b 
  \end{pmatrix}{}\, ,
\end{equation}
such that the structure constants become 
\begin{equation}
\breve{F}_{ab}{}^c = 0 \, ,  \quad     \breve{F}_{abc}   = \frac{\epsilon^{-\frac{3}{4}}}{\sqrt{2} } f_{abc}   \, ,   \quad \breve{\widetilde{F}}{}^{abc}=0 \, ,\quad \breve{\widetilde{F}}{}^{ab}{}_c = \frac{\epsilon^{ \frac{1}{4}}}{\sqrt{2} } f^{ab}{}_c     \, .
\end{equation}
In this frame the PL condition invokes 
\begin{equation}
  \left(\breve{\Gamma}^{abc} f_{abc} +3 \epsilon f^{ab}{}_c \breve{\Gamma}_{ab}\breve{\Gamma}^c \right) \breve G  = 0 \ .
\end{equation}
Notice here we have a direct similarity between the same constraint found in the $\eta$ branch.  Indeed, comparing this with eq.~\eqref{eqn:bianchiintegrable} we simply have to swap the index positions, up and down, (corresponding to the rotation eq.~\eqref{eqn:ODDlambda}), and identify $\lambda = (\eta + i)(\eta - i)^{-1}$.  Thus knowing already how to solve this equation we can construct  the curved space spinor  R/R field strengths via 
\begin{equation}\label{eqn:SbSbetaSrvecG}
  \widehat{G}^\prime=  \sum_{p=1} \frac{1}{p! 2^{p/2}} \widehat{G}^\prime{}_{a_1\dots a_p} \Gamma^{a_1\dots a_p} = 
     S_\beta S_b S_r  S_{{\cal O}}  \breve G \,.
\end{equation}
in flat indices and the final result arise after contraction with the left-invariant from $e^a$.  Note that here we are dressing with the spinor representative of $\widehat{E}^\prime$ in which the $B$-shift induced by $\rho$ has been stripped off.  To get the raw $\widehat{G}$ that obey $d\widehat{G}=0$ one must apply the spinor representative of a B-shift for $B_{WZW}$.
 
 \subsubsection*{Example}
Let us now come again to the $SL(2)\times SU(2)$ example. Like for the $\eta$-deformation in the last section, it is convenient to express all supergravity fields just in terms of the invariant tensors of the Lie algebra $\mathfrak{g}$ and the adjoint action $D_{ab}$. From this point of view $\beta^{ab}$ is the only quantity in \eqref{eqn:SbSbetaSrvecG} which also includes the inverse of $1 + D$. However, it always appears in either one of the two combinations
\begin{equation}
  r^a{}_c \beta^{cb} = \frac{1}{\sqrt{2}} \left( \kappa^{ab} - D^{ab}\right) \quad \text{or} \quad \beta_{ac} b^{cb} = -\frac12 \delta_a^b + \frac14 ( D_a{}^b + D^b{}_a )\,.
\end{equation}

At this stage we find it convenient to pick a particular representation for both $SU(2)$ and $SL(2)$ elements:
\begin{equation}
\begin{aligned}
  g_{SL(2)} &= \begin{pmatrix} \cosh\alpha_1 - \sinh\alpha_1 \cosh\alpha_2 & -\sinh \alpha_1 \sinh \alpha_2 e^{  \alpha_3} \\    \sinh \alpha_1 \sinh \alpha_2 e^{   - \alpha_3} &  \cosh\alpha_1 + \sinh \alpha_1 \cosh\alpha_2   \end{pmatrix} \, , \\
g_{SU(2)} &= \begin{pmatrix} \cos\beta_1 + i \sin \beta_1 \cos\beta_2 & \sin \beta_1 \sin \beta_2 e^{- i \beta_3} \\   - \sin \beta_1 \sin \beta_2 e^{  i \beta_3} &  \cos\beta_1 - i \sin \beta_1 \cos\beta_2   \end{pmatrix} \, . 
  \end{aligned}
\end{equation}
\def\c{{\textrm c}}
\def\ch{{\textrm{ch}}}
\def\s{{\textrm s}}
\def\sh{{\textrm{sh}}}
Here $\alpha_2$ will become the time-like direction.  To reduce paper we write $\c_i = \cos \beta_i$ and $\ch_i = \cosh \alpha_i$ etc. 
The line element is given by 
\begin{equation}
ds^2= \frac{1+\lambda}{1-\lambda} d\alpha_1^2 + \frac{(1-\lambda^2)\sh_1^2}{\Delta} \left(- d\alpha_2^2 + \sh_2^2 d\alpha_3^2 \right)+\frac{1+\lambda}{1-\lambda} d\beta_1^2 + \frac{(1-\lambda^2)\s_1^2}{\Theta} \left(  d\beta_2^2 + \s_2^2 d\beta_3^2 \right)\, ,  
\end{equation}
 and the B-field and dilaton 
 \begin{equation}
 \begin{aligned}
B = B_{WZW}  + B_0 \ , \quad B_0=  4 \lambda \Delta^{-1} \ch_1 \sh_1^3 \sh_2 d\alpha_2 \wedge d\alpha_3 + 4 \lambda \Theta^{-1} \c_1 \s_1^3 \s_2 d\beta_2 \wedge d\beta_3 \\ 
B_{WZW}= 2 \sh_1^2 \sh_2 \alpha_3 d\alpha_1 \wedge d\alpha_2 + 2 \s_1^2 \s_2 \beta_3 d\beta_1 \wedge  d\beta_2 \ , \quad \Phi =\phi_0 -  \frac{1}{2}\log \Delta \Theta  \ ,
  \end{aligned}
 \end{equation}
with 
\begin{equation}
\Delta = 1 + \lambda^2 - 2\lambda \ch_1 \ , \quad \Theta = 1 + \lambda^2 - 2\lambda \c_1 \ . 
\end{equation}
 We can apply directly eq.~\eqref{eqn:SbSbetaSrvecG} and after some work we   find
\begin{align}
  \widehat{G}'{}^{(1)} &= n \left(
\sh_1 \c_1 d \alpha_1 - \ch_1 \s_1 d \beta_1 \right) \nonumber\\
  \widehat{G}'{}^{(3)} &= n  \Bigl(
  \s_1 \sh_1^3 \sh_2 d \alpha_2 \wedge d \alpha_3 \wedge d \beta_1 - \c_1 \ch_1 \s_1^2 \sh_2 d\alpha_1 \wedge d\alpha_2 \wedge d \alpha_3 - \nonumber \\
  &\qquad\qquad \s_1^3 \s_2 \sh_1 d \alpha_1 \wedge d \beta_2 \wedge d \beta_3 - \c_1 \ch_1 \s_1^2 \s_2 d \beta_1 \wedge d \beta_2 \wedge d \beta_3 \Bigr) \nonumber \\
  \widehat{G}'{}^{(5)} &= n  \Bigl( \ch_1 \s_1^3 \s_2 \s_1^2 d \alpha_1 \wedge d \alpha_2 \wedge d \alpha_3 \wedge d \beta_2 \wedge d \beta_3 + \c_1 \s_1^2 \s_2^2 \sh_3^3 \sh_2 d \alpha_2 \wedge d \alpha_3 \wedge d \beta_1 \wedge d \beta_2 \wedge d \beta_3 \Bigr) \,.
\end{align}
which have to be twisted with the WZW contribution to the B-field  in order to obtain $\widehat{G} = e^{B_{\mathrm{WZW}}} \sum_{n=1}^3 \widehat{G}'{}^{(2n-1)}$ which satisfies the Bianchi identity $ d \widehat{G} = 0$.  All that remains to be fixed is the normalisation $n^2= \frac{\lambda^2}{8}e^{-2\phi_0}$.  One can lift this to a solution in 10d by adding an auxiliary $T^4$ exactly as is done in the appendix \ref{sec:appendixetasol}.

 \section{Conclusion}
 
In this work we have continued the development of Poisson-Lie T-duality, based on a Drinfel'd double, by describing   explicitly its embedding into DFT. We are further able to extend these ideas to include so-called  $\cal E$-models  (which incorporate e.g. integrable $\lambda$-deformed theories), the required Drinfel'd double is relaxed to that of a Manin quasi-triple.   In either case we are able to extend the conditions for Poisson-Lie symmetry or an ${\cal E}$-model to the dilaton, recovering rather simply a result that took some effort in the mathematical literature,  and to the R/R sector which to date had been treated in a somewhat ad-hoc fashion.  

The condition of having an $\cal E$-model can be understood in the context of $DFT_{\mathrm{WZW}}$ defined on  a $2D$-dimensional group manifold $\mathdsl D$ of an algebra $\frak{d}$  as demanding invariance under the $2D$ diffeomorphism symmetry.    Choosing a solution to the section condition amounts to finding a subalgebra $\tilde{\frak{h}} \in \frak{d}$,  reducing the dynamics to that defined on the coset ${\cal M} = \mathdsl D/\widetilde{H}$.  When $\tilde{\frak{h}}$ is unimodular, the equations of motion of regular DFT (within a solution to the section condition) are recovered, and when non-unimodular those corresponding to  a known modification of DFT (and SUGRA) are found.     In this way, for the backgrounds we are considering, the equations of motion for $DFT_{\mathrm{WZW}}$ become algebraic and match those derived some years ago from the context of the doubled worldsheet in the bosonic sector.   Similarly we are able to reduce the considerations of the R/R-sector to essentially an algebraic   problem.

Critical is that we are able to construct generalised frame fields that are $O(D,D)$ valued and moreover close under the generalised Lie derivative to generate structure constants of the Lie-algebra $\frak{d}$.
This is allows to translate the algebraic results for the $DFT_{\mathrm{WZW}}$  defined on $\mathdsl D$ to conventional target space fields on the physical spacetime ${\cal M}$.   We demonstrate this technology with examples corresponding to $\eta$ and $\lambda$ integrable deformations.   This is a quite satisfying result, viewed as conventional geometries these deformations look anything but simple, whereas in this language their underlying simplicity becomes transparent.  It seems plausible that more general integrable sigma models \cite{Georgiou:2016urf,Georgiou:2016urf2,Delduc:2014uaa,Demulder:2017zhz} may give rise to such a structure.  

One of the original motivations for this study was to see if by using $DFT_{\mathrm{WZW}}$ we could resolve the long standing questions concerning the global properties of non-Abelian duality transformations.  At first sight things seem promising since we have an underlying doubled group manifold $\mathdsl D$.  However, a closer look shows this is not quite the full story.  Firstly our constructions only make use of the algebra $\frak{d}$, additional input is required to specify the global structure of $\mathdsl D$, e.g. there may be discrete quotients to be taken.  A second challenge is that we assume a factorisation of the group elements of the form $\mathdsl g(X) = \tilde{h}(\tilde{x}) m(x)$ with particular parametrisation of $m(x)$ as the exponent of coset generators.  Neither the factorisation of $\mathdsl g$ nor the identity-connectedness of $m$ are guaranteed to hold globally.   Thus further work is needed to establish the patching required to extend our construction.  However, suppose that this procedure can be completed  and our frames given a global definition.  In that case we would have specified a generalised parallelisation for ${\cal M}$ and would provide an explicit demonstration of the fact  \cite{Lee:2014mla} that the reduction on such spaces constitutes a consistent truncation \cite{Duff:1994tn,Cvetic:2003jy}.      

This work prompts many interesting directions.  The most obvious is to describe the dressing coset procedure in this language (something that we intend to report on shortly \cite{cosets}), and eventually the extension to semi-symmetric spaces with the application to the full $AdS_5 \times S^5$ superstring in mind. In these more general cases we also intend to detail the question of supersymmetry, by making manifest the idea that whilst naively broken in conventional SUGRA it is recovered in DFT by allowing Killing spinors to have dependence  on  the `dual' coordinates \cite{Hassan:1999bv,Kelekci:2014ima}. 
In the present work we also showed that the PL T-duality rule on R/R fields can be recast in the format of a
Fourier-Mukai transformation, something which was known to be the case for Abelian T-duality \cite{Hori:1999me}.  It is well know that D-brane charges admit a K-theory classification \cite{Minasian:1997mm,Witten:1998cd}  and that this Fourier-Mukai transformation can be understood as implementing T-duality at the level of K-theory.  So one might  (optimistically perhaps given the state of knowledge of global properties)  hope to understand the Poisson-Lie transformation at this level.    Looking further ahead the prospect of using the algebraic description of these backgrounds to study higher order corrections \cite{Hohm:2014xsa,Hohm:2014eba,Deser:2014wva}  is enticing as is the interpretation of such generalised dualities in the context of the Exceptional Field Theory approach to M-theory.

\acknowledgments
We gratefully acknowledge  B.~Hoare, F.~Rudolph and C.~Strickland-Constable for helpful discussions.  
FH and SD would like to thank the organisers of Corfu Summer Institutes 2018 where some results of this work were presented, and the COST network for supporting the visit of FH to this workshop. 
The work of DCT is supported by a Royal Society University Research Fellowship URF 150185, and in part by the STFC grant ST/P00055X/1.   SD is supported by the FWO-Vlaanderen through aspirant fellowship and in part through  the project FWO G020714N.

\appendix

\section{Conventions and notation}\label{App:Conventions} 
There are many different groups, algebras, subgroups and subalgebras encountered in this paper -- we list the main definitions in table \ref{tab:convs1}.  Commensurate with this is an abundance of indices outlined in table \ref{tab:convs2}.

\begin{table}[h!]
\centering
\begin{tabular}{c|c|c|c|c!{\color{black}:}c}
                           & Double $\mathfrak d$                & Lagrangian & Lagriangian  & Algebra $\frak g$ &   \\
                                                       &              & subalg. $\tilde{\mathfrak{h}}$ & compl. $\mathfrak{k} $ &  &   \\\hline \Tstrut
Exponentiation          & $\mathdsl D$                & $\widetilde H$ & $\exp(\mathfrak k)$/$H$  & $G$ &   \\
                        &             &    &  for  $\frak k$ subsp./subalg.) &  &   \\
Algebra gen.          & $\mathdsl T_A$                               & $\widetilde T^a$                                        & $T_a$                                & $t_a$             &                               \\
Inner product           &  $\left\llangle \mathdsl T_A,\mathdsl T_B\right\rrangle = \eta_{AB}$     &                                      &                           & $\langle T_a, T_b \rangle = \kappa_{ab}$             &                               \\
Structure csts        & $\mathdsl F_{AB}{}^C$               &   $\widetilde F^{ab}{}_c$            			       &     $ F_{ab}{}^c$         &    $f_{ab}{}^c$               &   \\
Group element              &  $\mathdsl g(X^I)$           & $\tilde h(\tilde x_{\tilde i})$            & $m(x^i)$                                  & $g$               &    $\bar g^2=g$                           \\
Adjoint action             & $\mathdsl M_A{}^B$                   & $\widetilde M^a{}_b$                         & $M_a{}^b$                            &       $D_ a{}^b$         &    $\overline{D}_ a{}^b$                           \\   
L/R MC forms & $\mathdsl E^A{}_I$/$\mathdsl V^A{}_I$ & $\tilde E^{A\tilde  i}/ \tilde V^{A\tilde i}$    &$E^A{}_{ i}/ V^A{}_{i}$ &      &  \\

L/R MC comps &                                   & $\tilde e_{a}{}^{\tilde i}/\tilde v_{a}{}^{\tilde i}$ &$e^a{}_{ i}/ v^a{}_{i}/A_{ai}$ &  $\underline{e}^a{}_i$/$\underline{v}^a{}_i$    & $\bar e^a{}_i$/$\bar v^a{}_i$ \\
                                                
\end{tabular}
\caption{The variety of algebras, groups and group elements used.}
\label{tab:convs1}
\end{table}

\begin{table}[h!]
\centering
\begin{tabular}{c|c|c}
    & tensor & indices \\ \hline
Flat frame &  $T_A$  &$A,B,C,\cdots  = 1 \dots 2 D$ \\
Rotated flat frame & $\breve T_{\breve{A}}$  & $A,B,C,\cdots  = 1 \dots 2 D$ \\
Doubled curved space & $T_I$ &$I,J,K,\cdots = 1 \dots 2 D $ \\
Generalised tangent space  & $\widehat T_{\hat I}$  &$\hat I,\hat J,\hat K,\cdots= 1 \dots 2 D$ \\ 
\end{tabular}
\caption{The variety of indices used.}
\label{tab:convs2}
\end{table}

\subsubsection*{ Sigma-models and supergravity}
We consider 2d non-linear sigma models in Lorentzian signature given by 
\begin{equation}
S = \frac{1} {\pi s} \int d \sigma d \tau \partial_+ X^i ( G(X) - B(X) )_{ij} \partial_- X^j \ , 
\end{equation} 
in which $\partial_\pm = \frac{1}{2} (\partial_\tau \pm \partial_\sigma)$.   This sign choice for the NS two-form field means that for a constant $G$ and $B$ the Hamiltonian 
\begin{equation}
\textrm{Ham} = \dot X^i P_i - \textrm{L} = \frac{1}{4 \pi s}  {\cal Z}^M  {\cal H } _{MN } {\cal Z}^N \ , \quad  {\cal Z}^M = (2\pi s P_i , \partial_\sigma X^i) \, , 
\end{equation}
is written with the  generalised metric defined as 
\begin{equation}
{\cal H } _{MN } = \begin{pmatrix} G^{-1} & -G^{-1} B \\ B G^{-1} & G- BG^{-1} B \end{pmatrix} \, . 
\end{equation}

 The NS sector supergravity equations  are given by  (for type IIB)
\begin{equation}\label{eq:mSUGRA}
\begin{aligned}
 0=  &R_{mn} + 2\nabla_{mn }\Phi -\frac{1}{4}H_{mpq}H_{n}{}^{pq} \\
&\qquad - e^{2\Phi}\Big( \frac{1}{2}({F_1}^2)_{mn}+\frac{1}{4}({F_3}^2)_{mn} +\frac{1}{96}({F_5}^2)_{mn} - \frac{1}{4}g_{mn} \big( F_1^2 +\frac{1}{6}F_{3}^2 \big)\Big)  \ , \\
0  =  &d[e^{-2\Phi} \star H] + F_1\wedge \star F_3 + F_3 \wedge F_5  \ , \\
0  =   &R+ 4 \nabla^2 \Phi - 4 (\partial \Phi)^2 - \frac{1}{12} H^2 \, , 
\end{aligned}
\end{equation}
in which   $H=dB$. 
For the R/R fields  we have Hodge duals defined (in $d=10$ dimensions) according to $F_{(p)}= - (-1)^{p(p+1)/2} \star F_{(d-p)}$ for which the poly-form $  \mathscr F= \sum_p F_{(p)}$  obeys  
\begin{equation}
d_H \mathscr F = (d+ H\wedge) \mathscr F   = 0   \ . 
\end{equation}
The Hodge star operator is such that $\star^2 \omega_{(p)}=  s (-1)^{p (d-p)}\omega_{(p)}$, where $s$ is the signature. For IIB we have
\begin{equation}
 {\mathscr F} = F_{(1)} + F_{(3)} + F_{(5)} - \star F_{(3)} + \star F_{(1)}  \ , \quad  F_{(5)} = \star F_{(5)} \ . 
\end{equation}
 We occasionally also use $  \mathscr F  = e^{-B} {\mathscr G}  $ and ${ \cal F} = e^\Phi   {\mathscr F}$.

In modified supergravity \cite{Arutyunov:2015mqj}  we have instead 
\begin{equation}
\begin{aligned}
 0=  &R_{mn}   -\frac{1}{4}H_{mpq}H_{n}{}^{pq} + \nabla_m X_n + \nabla_n X_m \\
&\qquad -   \Big( \frac{1}{2}({{\cal F}_1}^2)_{mn}+\frac{1}{4}({{\cal F}_3}^2)_{mn} +\frac{1}{96}({{\cal F}_5}^2)_{mn} - \frac{1}{4}g_{mn} \big( {\cal F}_1^2 +\frac{1}{6}{\cal F}_{3}^2 \big)\Big)  \ , \\
0  =  &d   \star H  + {\cal F}_1\wedge \star {\cal F}_3 + {\cal F}_3 \wedge {\cal F}_5  -2 \star d X -2 X \wedge \star H \ , \\
0  =   &R +4 \nabla_n X^n -  4 X_n X^n - \frac{1}{12} H^2 \, ,   \\
0  = &  {\bf d} {\cal F} \equiv (d+ H\wedge - Z\wedge - \iota_I ) {\cal  F}      \ . 
\end{aligned}
\end{equation}
 Here the vector $X$ is given by 
\begin{equation}
X= Z+I \ , 
\end{equation}
with the constraints
\begin{equation}
dZ + \iota_I H = 0  \ , \quad \iota_I Z = 0  \ , \quad L_I g = L_I H = 0 \ . 
\end{equation}
For the case of $I=0$ we have that $X= d\phi$ and the conventional supergravity is recovered.   In general we identify the ``dilaton'' as the exact piece of $Z$;
\begin{equation}
Z= d\phi + \iota_I B - V \ , \quad L_I B = d V \ . 
\end{equation}
In these equations we use the interior contraction defined as $\iota_{\cal I} \omega = \star ( {\cal I} \wedge \star \omega)$ and recall $L_{\cal I} \omega   = d \iota_{{\cal I}} \omega+ \iota_{{\cal I}} d \omega $ .

\section{Algebraic structures}\label{app:alg}
\subsubsection*{Algebras and Groups}

We work with real Lie algebras $\frak{g}$, and corresponding group $G$, of dimension $\dim G = D$ with a basis of anti-Hermitian generators $\{ t_a\} $  equipped with an ad-invariant  symmetric pairing given by the Cartan-Killing form, $ \kappa = \langle\cdot, \cdot \rangle$, obeying  
\begin{equation}
[t_a, t_b] = f_{ab}{}^c t_c   \ , \quad \kappa_{ab} = \langle t_a , t_b\rangle =- \frac{1}{2h^\vee}f_{ad}{}^e f_{b e}{}^d  \ . 
\end{equation}
Left/right-invariant forms and adjoint actions for a group element $g(x) \in G$, depending on local coordinates $x^i$,  are defined according to
\begin{equation}
\begin{aligned}
d gg^{-1} =  v  &= v^a t_a = v^a{}_i dx^i t_a \ , \quad g^{-1} dg = e= e^a t_a = e^a{}_i dx^i t_a  \ , \\  
  \text{ad}_g t  &= g t_a g^{-1} = D[g]_{a}{}^b t_b \ , \quad v^a  =e^b  D[g]_{b}{}^a \ .
\end{aligned}
\end{equation}
This definition of the adjoint action obeys 
\begin{equation}
   D[g]^{-1} = D[g^{-1}]  \, ,\quad  D[g] \kappa D[g]^T= \kappa \, ,
\end{equation}
we will write $D\equiv D[g]$ when clear from the context.

 The Maurer-Cartan equations are
\begin{equation}
dv^a = + \frac{1}{2} f_{bc}{}^a v^b \wedge v^c \, , \quad de^a = -\frac{1}{2} f_{bc}{}^a e^b \wedge e^c \ ,  \quad dD[g]_a{}^b = v^c D[g]_a{}^d f_{cd}{}^b\,  . 
\end{equation}
 
\subsubsection*{R-matrices }

We consider ${\cal R}$  a skew-symmetric endomorphism of $\frak{g}$ defined as 
\begin{equation}
{\cal R}(t_a) = R_a{}^b t_b  \ , \quad  R_{ab} =  R_a{}^c\kappa_{cb}  =  - R_{ba} \ , \quad R^{ab} = \kappa^{ac}R_{c}{}^b = - R^{ba}  \ . 
\end{equation}
From ${\cal R}$ is constructed a second bracket over the   vector space   $\frak{g}$, 
\begin{equation}\label{eq:rbra} 
[x, y]_{\cal R} = [{\cal R}(x) , y] + [x, {\cal R}(y)] \ , \quad [t_a ,t_b ]_{\cal R} = \tilde{f}_{ab}{}^c t_c   \  ,  \quad \tilde{f}_{ab}{}^c= R_{a}{}^e f_{eb}{}^c + R_{b}{}^e f_{ae}{}^c \, .
\end{equation}
This will obey the Jacobi identity provided ${\cal R}$  solves the  modified classical Yang-Baxter equation 
\begin{equation}
[ {\cal R}(X),{\cal R}(Y) ] - {\cal R}( [X, Y]_{\cal R}) + c^2 [X,Y] = 0 ~~~~\forall X,Y, \in     \frak{g} \ . 
\end{equation} 
We will define the algebra constructed from this bracket as $\frak{g}_{\cal R}$.  We have two Lie-brackets giving algebras $\frak{g}$ and $\frak{g}_{\cal R}$ over the same vector space and this set up is also called a bi-algebra.  Technically the construction of eq.~\eqref{eq:rbra} means this is a coboundary bi-algebra.  It can be useful to define a element ${\bf r} \in  \bigwedge^2 \frak{g}$ as 
\begin{equation}
{\bf r} = R^{ab} t_a \otimes t_b \ . 
\end{equation}

\subsubsection*{Drinfel'd double} 

Here we consider real Lie algebra $\frak{d}$, and corresponding group $\mathdsl D$, of dimension $\dim\mathdsl D = 2d$ with a basis of anti-Hermitian generators $\{\mathdsl  T_A\} $  equipped with an ad-invariant  symmetric pairing, $ \eta =\llangle \cdot, \cdot \rrangle$, obeying  
\begin{equation}
[\mathdsl T_A, \mathdsl T_B] = \mathdsl F_{AB}{}^C \mathdsl T_C   \ , \quad \eta_{AB} = \llangle\mathdsl T_A,\mathdsl T_B \rrangle  \ . 
\end{equation}
A classical double is such a real Lie algebra that admits a decomposition $\frak{d} = \frak{g} \oplus \tilde{\frak{g}}$ as the sum of two Lie subalgebras each of dimension $d$ that are Lagrangian (maximally isotropic with respect to $\eta$).      In a basis $\mathdsl T_A = (\widetilde T^a , T_a)$ we have that 
\begin{equation}
\begin{aligned}
~ [T_a,T_b] = F_{ab}{}^c T_c \ , \quad [\widetilde T^a,\widetilde T^b ] = \widetilde{F}^{ab}{}_c\widetilde T^c  \, , \quad [T_a,\widetilde T^b ] = \widetilde{F}^{bc}{}_a T_c - F_{ac}{}^b\widetilde T^c \, , \\
\llangle T_a , T_b \rrangle =  \llangle\widetilde T^a ,\widetilde T^b \rrangle =0 \, ,\quad \llangle T_a , \widetilde T^b\rrangle  = \delta_a{}^b \, ,\quad \llangle\widetilde T^a, T_b \rrangle = \delta^a{}_b  \, . 
\end{aligned}
\end{equation}
The Jacobi identity of $\mathdsl F_{AB}{}^C$ places a compatibility condition on the two Lie subalgebras, namely that   $\delta(T_a) = \widetilde{F}^{bc}{}_a T_b \otimes T_c $   viewed as a map $\frak{g} \to  \bigwedge^2 \frak{g} $, should be a one-cocycle for $\frak{g}$ valued in   $\bigwedge^2 \frak{g}$ obeying 
\begin{equation}
0=d \delta (X, Y) \equiv  \text{ad}_X \delta(Y) -  \text{ad}_Y \delta(X) - \delta( [X,Y])  \, ,
\end{equation}
in which the adjoint action extends to the tensor product as $ \text{ad}_X Y = ( 1\otimes  \text{ad}_X +  \text{ad}_X \otimes 1) Y$ for $Y \in \frak{g}\otimes \frak{g}$.  Of particular interest will be the case when the one-cocycle is a one-coboundary 
\begin{equation}
  \delta (X )  = [X,{\bf r} ]  \ , \quad  {\bf r}  \in   \frak{g}\wedge \frak{g} . 
\end{equation}
Identifiying $T_a = t_a$ and with ${\bf r}  =R^{ab} t_a \otimes t_b$ we have that 
\begin{equation}
 \widetilde{F}^{ab}{}_c =   R^{a e}F_{c e }{}^b - R^{b e}F_{c e }{}^a  \ ,
\end{equation}
which is nothing other than the raising of indices on $\widetilde{F}_{ab}{}^c$ defined in  eq.~\eqref{eq:rbra} using $\kappa$.   In this case the double $\frak{d} = \frak{g} \oplus \frak{g}_{R}$ and the Jacobi identity of $\mathdsl F_{AB}{}^C$ follows from the mCYBE.  In this way we have an equivalence between such doubles and coboundary Lie-bialgebras.   The cocycle  $\delta$ can be integrated to give a cocycle on $G$ valued in $ \frak{g}\wedge \frak{g}$
\begin{equation}
\Pi[g] = \left( 1 -  \text{ad}_{g^{-1} } \otimes \text{ad}_{g^{-1} }   \right) {\bf r}   \, , 
\end{equation}
which   for $g = \exp (\epsilon x) $ has $\Pi[g] \sim \epsilon [X, {\bf r} ]    = \epsilon \delta(X) $ and obeys 
\begin{equation}\label{eq:Gcocycle}
\Pi[h g] = \Pi(g) +  \text{ad}_{g^{-1} } \otimes \text{ad}_{g^{-1} } \Pi[h]  \ , \quad \Pi[e] = 1 \, .  
 \end{equation}
 In components 
 \begin{equation}
   \Pi[g]^{ab} =R^{ab} - D[g^{-1}]_{c}{}^{a} R^{cd}  D[g^{-1}]_d{}^b \ . 
 \end{equation}
 Whilst the above expression applies in the case of the coboundary specialisation, one can construct the same group cocycle for any double as follows.   Let $g$ be a group element for $G =\exp \frak{g} \subset \mathdsl D$ and using  its adjoint action on $\frak{d}$, 
 \begin{equation}
g \mathdsl T_A g^{-1} = D[g]_{A}{}^B \mathdsl  T_B  \ , \quad  D[g]_{A}{}^B = \begin{pmatrix} D[g]^a{}_b &  D[g]^{ab}  \\ 0 & D[g]_{a}{}^b  \end{pmatrix}  \ ,
 \end{equation}
  one has
 \begin{equation}
 \Pi[g]^{ab} = D[g]^{a c}D[g]^{b}{}_c \ . 
 \end{equation}
The group composition of the adjoint action in $\mathdsl D$ shows that the cocycle properties eq.~\eqref{eq:Gcocycle} holds and its derivative returns the algebra cocycle $\delta$.  
The cocycle can be understood as being as element of $T_e(G) \otimes T_e(G)$, and by taking its right translation to  a point $g$ we have a bi-vector $\Pi_g \in T_g(G) \otimes T_g(G)$; this   endows a Poisson structure to $G$ making it  a Lie-Poisson group manifold. 
 
\subsubsection*{Manin pair, triple and quasi-triple} 
 We now describe a weakening of the above structure to define a Manin quasi-triple.  A pair  $(\mathfrak d, \tilde{\mathfrak h})$ consisting of an algebra, $\mathfrak d$ and a Lagrangian subalgebra $\tilde{\mathfrak h} \subset \mathfrak d$ is called a Manin pair. A  Manin quasi-triple $(\mathfrak d, \tilde{\mathfrak h} , \mathfrak k)$  is a Manin pair $(\mathfrak d, \tilde{\mathfrak h})$  together with a choice of complementary Lagrangian subspace $\mathfrak{k}$ such that $\mathfrak d = \tilde{\mathfrak h}  \oplus \mathfrak k$  .   Different choices of complementary subspaces are related by a twist $t\in \Lambda^2\tilde{\mathfrak h} $ \cite{2003math10359K}.     The salient difference to Drinfel'd double is that the complementary Lagrangian $\frak k$ need not be a subalgebra. 
  
We define a basis of anti-Hermitian generators $\{\mathdsl T_A\} $ for $\mathfrak d$ equipped with an ad-invariant  symmetric pairing, $ \eta =\langle \cdot, \cdot \rangle$, obeying  
\begin{equation}
[\mathdsl T_A,\mathdsl T_B] =\mathdsl F_{AB}{}^C\mathdsl T_C   \ , \quad \eta_{AB} = \langle\mathdsl T_A,\mathdsl T_B \rangle  \ . 
\end{equation}
 Letting  $\mathdsl T_A = (\widetilde T^a , T_a)$ be the decomposition   $\mathfrak d = \tilde{\mathfrak h}  \oplus \mathfrak k$ these  relations read 
 \begin{equation}
\begin{aligned}
~ [T_a,T_b] = F_{ab}{}^c T_c+\phi_{abc}\widetilde T^c \ ,& \quad [\widetilde T^a,\widetilde  T^b ] = \widetilde{F}^{ab}{}_c\widetilde  T^c  \, , \quad [T_a,\widetilde  T^b ] = \widetilde{F}^{bc}{}_a T_c - F_{ac}{}^b\widetilde  T^c \, , \\
\llangle T_a , T_b \rrangle &=  \llangle\widetilde  T^a ,\widetilde  T^b \rrangle =0 \, ,\quad \llangle T_a ,\widetilde  T^b\rrangle  = \delta_a{}^b \, . 
\end{aligned}
\end{equation}
  The object $\phi_{abc}$  is   antisymmetric in all its indices and invariant under the (co-adjoint) action of $\widetilde{H} = \exp  \tilde{\mathfrak h}$. 
 
\subsubsection*{ Algebraic structure of $\lambda$- and $\eta$ models}
 
 The integrable $\lambda$  and $\eta$ models can be placed into this algebraic framework \cite{Klimcik:2015gba,Vicedo:2015pna} in which $ \frak{d} = \frak{g} + \frak{g}$  (a  Manin quasi-triple)  and  $\frak{d} = \frak{g} + \frak{g}_{\cal R} = \frak{g}^{\mathbb{C}}$ (a Drinfel'd double)  respectively.

Consider a Lie algebra $\frak g$ endowed with an ad-invariant non-degenerate symmetric bilinear form $\kappa$. The construction of the $\lambda$-deformation  requires a double that is the direct sum  $\frak d =\frak g\oplus \frak g$  equipped with the   inner product
\begin{align*}
\llangle \{ X_1,Y_1 \}, \{ X_2, Y_2\}  \rrangle =  \langle  X_1, X_2 \rangle   - \langle Y_1, Y_2 \rangle  \,.
\end{align*}
The two subspaces completing the Manin quasi-triple are taken to be the diagonal subalgebra $\tilde{\mathfrak h}=  \frak g_{\mathrm{diag}}$, embedded in $\frak d$ by the map $X\mapsto \{ X,X \}/\sqrt{2}$ for $X\in \mathfrak g$, and the anti-diagonal subspace $\frak{k} = \frak g_{\mathrm{anti-diag}}$ embedded as $X\mapsto \{ X,-X \}/\sqrt{2}$ in $\frak d$.   Let $f_{ab}{}^c$ be generators of $\mathfrak g$, then in $\frak{d}$ we have that  
\begin{equation}
F_{ab}{}^c = 0 \ , \quad \widetilde{F}^{ab}{}_c = \frac{1}{\sqrt{2} } f^{ab}{}_c    = \frac{1}{\sqrt{2} } \kappa^{a d} \kappa^{b e}    f_{de}{}^f \kappa_{  f c} \ , \quad      \phi_{abc}   = \frac{1}{\sqrt{2} } f_{abc} =  \frac{1}{\sqrt{2} } f_{ab}{}^f \kappa_{  fc }  \, .
\end{equation}
 
 The $\eta$-deformation on the other hand, the double is determined by the operator $\mathcal R$ entering the definition of the deformation. This operator is the canonical R-matrix associated to a semi-simple Lie algebra $\frak g$ with Killing form $\kappa$. It acts by anti-symmetrically swapping positive and negative roots and annihilates the Cartan.  As described above, since ${\cal R}$ is a solution to the classical (modified) Yang-Baxter equation  it defines a second Lie-bracket $[\cdot,\cdot]_{\cal R}$ on $\frak g$. The double is the direct sum $\frak d= \frak g\oplus \frak g_{\cal R}$, which is isomorphic to the complexification $\mathfrak g^{\mathbb C}$ of $\frak g$. This double can be decomposed in to a Manin pair using the Iwasawa decomposition $\frak g^{\mathbb C}=\frak g \oplus (\frak a + \frak n)$, where   $\frak g$ and $\frak a + \frak n$ are both Lagrangian subalgebras of $\frak d = \frak g^{\mathbb C}$. The ad-invariant non-degenerate symmetric bilinear form on $\frak d = \frak g^{\mathbb C}$ is 
\begin{align*}
\llangle Z_1,Z_2\rrangle = -i\,\kappa(Z_1,Z_2)+ i\,\overline{\kappa(Z_1,Z_2)}\,,
\end{align*}
where $Z\in \mathfrak g^{\mathbb C}$ and $\overline{\cdot}$ denotes the complex conjugation.
 
  \section{Solution to the section condition}\label{app:framealg}
In this appendix we  show how the choice of generalised frame fields solves the section condition \eqref{eqn:defhattedind}. Even more, this equation together with the requirement to be an $O(d,d)$-element, will completely fix the form of the generalised frame field in \eqref{eqn:genframefield} completely in terms of the element of the right-invariant form on the double $\frak d=\tilde{\frak{h}}\oplus \frak k $,
\begin{equation}
\mathdsl  T_A \mathdsl V^A{}_I d X^I =\mathdsl T_A  \mathdsl M_B{}^A  \mathdsl E^B{}_I d X^I = \partial_I \mathdsl g \mathdsl g^{-1} d X^I\, . 
\end{equation} 
Let us now explicitly calculate $\mathdsl V^A{}_I$ by using the parameterization of the double element $\mathdsl g = \tilde h(\tilde x_{\tilde i}) \, m(x^i)$, $\tilde h \in \widetilde H$ and $m \in \exp(\frak k)$.
We obtain
\begin{equation}
 d\mathdsl g\mathdsl g^{-1}=\mathdsl T_A \mathdsl V^A{}_I d X^I =\tilde  h \partial_i m m^{-1} \tilde h^{-1} d x^i + \partial_{\tilde i}\tilde  h \tilde h^{-1} d \tilde x_{\tilde i}
    =\mathdsl T_A \mathdsl V^A{}_i d x^i + T^a \mathdsl V_{a}{}^{\tilde i} d \tilde x_{\tilde i}
\end{equation}
and $\mathdsl V^a{}_{\tilde i} = 0$ because $h$ is an element of the subgroup $\widetilde H$. As usual the inverse transpose of $\mathdsl V^A{}_I$ is denoted as $\mathdsl V_A{}^I$  
, for which we have that  $\mathdsl V^{ai} = 0$. 

Looking now at the generalised frame field, it is convenient to decompose it into two parts,
\begin{equation}
  \widehat{E}_A{}^{\hat I} = M_A{}^B \widehat{V}_B{}^{\hat I}\,.
\end{equation}
Using this decomposition we need to check that
\begin{equation}\label{eq:seccondreq}
  \widehat{E}^A{}_{\hat I} \mathdsl E_A{}^I \partial_I = \widehat{V}^A{}_{\hat I} \mathdsl V_A{}^I \partial_I = 
    \begin{pmatrix} 0 & \partial_i \end{pmatrix} \,.
\end{equation}
We need the above equation to only hold on the physical fields that we are considering, which will depend only on the coordinates $x^i$ so that we may use on the left hand side that   $\partial_I = \begin{pmatrix} 0& \partial_i  \end{pmatrix}$.    To see the first equality of eq.~\eqref{eq:seccondreq} the parametrisation $ \mathdsl g(x, \tilde{x}) = \tilde h(\tilde x_{\tilde i}) m(x^i) $ is paramount since it ensures that the differences between the adjoint action  $\mathdsl{M}$ of $\mathdsl g $ and the adjoint action $M$ of $m$ don't contribute.  Then  eq.~\eqref{eq:seccondreq}   reduces to
\begin{equation}
  \begin{pmatrix} \widehat{V}^{Ai}\mathdsl V_A{}^j \partial_j & 
    \widehat{V}^A{}_i \mathdsl V_A{}^j \partial_j \end{pmatrix} = \begin{pmatrix} 0 & \partial_i \end{pmatrix}
\end{equation}
or equivalently
\begin{equation}
  \widehat{V}^{Ai} \mathdsl V_A{}^j = 0 \quad \text{and} \quad
  \widehat{V}^A{}_i \mathdsl V_A{}^j = \delta_i^j\,.
\end{equation}
The second of these is satisfied providing
\begin{equation}
 \mathdsl T_A \widehat{V}^A{}_i = T_a \mathdsl V^a{}_i +  \widetilde T^a\mathdsl V_a{}^j \rho_{ji}\,,
\end{equation}
with arbitrary and to be fixed matrix $\rho_{ij}$. For the first component $\widehat{V}^{Ai}$ we find
\begin{equation}
\mathdsl  T_A \widehat{V}^{Ai} \partial_i = \widetilde T^a  \mathdsl V_a{}^i \partial_i \,,
\end{equation}
because $\mathdsl V^{ai} = 0$. Furthermore we need to require the generalised frame field to be an $O(D,D)$ element,
that is $\widehat{V}^A{}_{\hat I}$ has to have the property
\begin{equation}
  \widehat{V}^A{}_{\hat I} \eta_{AB} \widehat{V}^B{}_{\hat J} = \eta_{\hat I\hat J}\,.
\end{equation}
This implies several constraints (those on the right being implied by those on the left):
\begin{align}
  \widehat{V}^A{}_i \eta_{AB} \widehat{V}^B{}j &= 0  &
    \langle \partial_i\mathdsl g\mathdsl g^{-1} , \partial_j\mathdsl g\mathdsl g^{-1} \rangle + 2 \rho_{(ij)} &= 0 \\
  \widehat{V}^{Ai} \eta_{AB} \widehat{V}^{Bj} &= 0 &
    \langle\widetilde T^a,\widetilde T^b \rangle  v_a{}^i  v_b{}^j &= 0 \\
  \widehat{V}^{Ai} \eta_{AB} \widehat{V}^B{}_j &= \delta^i_j &
    \langle \widetilde T^a, T_b \rangle  v_a{}^i  v^b{}_j &= \delta^i_j\,.
\end{align}
The first term in the first equation on the right vanishes by assumption  and the second one implies that $\rho_{ij}$ has to be antisymmetric. All other identities follow automatically.

Summarising the discussion above, $\widehat{V}_A{}^{\hat I}$ reads
\begin{equation}
  \widehat{V}_A{}^{\hat I} = \begin{pmatrix}
    v^a{}_i & 0\\ 
     v_a{}^j \rho_{ji}  & v_a{}^i 
  \end{pmatrix}\,,
\end{equation} 
as claimed in \eqref{eqn:genframefield}.  The precise form of $\rho_{ij}$ will be fixed to ensure the frame fields obey a frame algebra under the generalised Lie derivative.  
\section{Fluxes  in the generalised parallelizable frame}\label{appendix:fluxes}
 
It is instructive to compute the components of $F_{\hat I\hat J\hat K}$, the structure constants dressed by the generalised frame fields as in eq.~\eqref{eq:Fdef}, explicitly. In conventional notation these are denoted ${\bf{H}}, {\bf{Q}}, {\bf{F}},  {\bf{R}}$ \cite{Shelton:2005cf}.  First let us consider the case of a Drinfel'd double for which we have in general  
\begin{align*}
	{\bf{H}}_{ijk}&=0\,,\\
	 {\bf{F}}_{ij}{}^k&=e^b{}_{[i}e^c{}_{j]}e_a{}^kF_{bc}{}^a\,\\
	{\bf{Q}}^{ij}{}_k&=	e_a{}^{[i}e_b{}^{j]}e^c{}_k(\widetilde{F}^{ab}{}_c+F_{dc}{}^b\Pi^{ad})=-e_a{}^{[i}e_b{}^{j]}\partial_k \Pi^{ab}\,,\\ 
	{\bf{R}}^{ijk}&=\frac12e_a{}^{[i}e_b{}^je_c{}^{k]}(F_{de}{}^{a}\Pi^{bd}\Pi^{ce}+\tilde{F}^{ab}{}_d\Pi^{dc})\,,\\
	&= -\frac{1}{4}e_a{}^{[i}e_b{}^je_c{}^{k]}(\widetilde{F}^{a[b}{}_d\Pi^{c]d}-2F_{de}{}^a\Pi^{db}\Pi^{ec})=0\,.
\end{align*}   
 The  identity required to show the vanishing of the R-flux is slightly involved and was provided in \cite{Sfetsos:1997pi}. 

In the case that the Drinfel'd Double corresponds to a coboundary Lie bialgebra, i.e. both $\widetilde{F}^{ab}{}_c$ and $\Pi^{ab}$ are expressible in terms of an $R$-matrix (see Appendix A) we can go a little further to express the fluxes as       \begin{align*}
	{\bf{H}}_{ijk}&=0\,,\\
	 {\bf{F}}_{ij}{}^k&=v^b{}_{[i}v^c{}_{j]}v_a{}^kf_{bc}{}^a\,,\\
		{\bf{Q}}^{ij}{}_k&=2 	v_a{}^{[i}v_b{}^{j]}v^c{}_k F_{dc}{}^a R^{bd} =  	v_a{}^{[i}v_b{}^{j]}v^c{}_k \widetilde{F}^{ab}{}_c\,,\\
	{\bf{R}}^{ijk}&=  3 \left(  v_a{}^{[i}v_b{}^jv_c{}^{k]}  -  e_a{}^{[i}e_b{}^je_c{}^{k]}  \right) F_{d e}{}^a R^{b d} R^{c e} \\
	&= -3c^2  \left(  v_a{}^{[i}v_b{}^jv_c{}^{k]}  -  e_a{}^{[i}e_b{}^je_c{}^{k]}  \right) \widetilde{F}^{abc} =0\,.
\end{align*} 
In the two last fluxes, we used   the definition of the modified Yang-Baxter equation.

Let us now turn to the more  general case of a  Manin quasi-triple.  Here   we must make a slight refinement, the $H$-flux also has a contribution that arises as a twisting of the Courant bracket as  discussed in comment \ref{comm:2}.  In what follows we shall strip off this twisting   using the frame fields
 \begin{equation} 
    \widehat{E}^\prime_A{}^{\hat I}  = M_A{}^B \begin{pmatrix}
      v^b{}_i & 0 \\
      0 &    & v_b{}^i
    \end{pmatrix}\indices{_B^{\hat{I}}}\,,
  \end{equation} 
 that obey 
  \begin{equation} 
       \widehat{{\cal L}}_{\widehat{E}^\prime_A} \widehat{E}^\prime_B{}^{\hat I} = F_{AB}{}^C \widehat{E}^\prime_C{}^{\hat I} +\begin{pmatrix}  \widehat{E}^\prime_A{}^j   \widehat{E}^{}\prime_B{}^k  ( \Omega^{(3)} - d\omega^{(2)})_{jk i }   \\  0 \end{pmatrix} \ .  
       \end{equation} 
  We then consider $F^\prime_{\hat{I}\hat{J}\hat{K}} =  F_{ABC} \widehat{E}^\prime{}^C{}_{\hat{I}} \widehat{E}^\prime{}^B{}_{\hat{J}} \widehat{E}^\prime{}^C{}_{\hat{K}}$ and simply add back the contribution to the  ${\bf H}$ given by $\bar{\bf H} =   \Omega^{(3)} - d\omega^{(2)}$.

 We recall that from the coset representative $m(x)$ for $\mathdsl D/\widetilde H$ we have 
 \[
 dm m^{-1} = v^a{}_i dx^i T_a + A_{ai} dx^i \widetilde T^a \ . 
 \]
 Then evaluating the fluxes one finds 
     \begin{align*}
	{\bf{H}}_{ijk}&= \left( \phi_{abc} M^a{}_d M^b{}_e M^c{}_f   + \frac{1}{2} F_{ab}{}^c M^a{}_d M^b{}_e M_{cf}-   \frac{1}{2} \widetilde{F}^{ab}{}_c  M_{ad}  M_{be} M^c{}_f    \right) v^d{}_{[i}v^e{}_{j}v^f{}_{k]}   + 	{\bf{\bar{H}}}_{ijk}  \,,\\
	 {\bf{F}}_{ij}{}^k&=2   v_a{}^k v^e{}_{[i }v^f{}_{j]}  \Big(   F_{bc}{}^d(M_{de}M^b{}_f M^{ca}- M_d {}^aM^b{}_eM^c{} _f)- \widetilde F^{bc}{}_d (M_{be}M_c{}^aM^d{}_f- M_{be}M_{cf}M^{da}) \\
&\qquad \qquad \qquad +   M^c{}_eM^d{}_fM^{ba} \phi_{bcd} \Big)\,, \\
		{\bf{Q}}^{ij}{}_k&=2   v_a{}^{[i}v_b{}^{j]}v^f{}_k   \Big(  -   F_{cd}{}^e( M_{ef}M^{ca}M^{db} - M_e{}^b M^c{}_fM^{da}  ) + \widetilde F^{cd}{}_e(M_c{}^aM_d{}^bM^e{}_f+ M_{cf}M_d{}^bM^{ea})\\
&\qquad \qquad \qquad +   M^d{}_fM^{ea}M^{cb}  \phi_{dec}\Big)\,,\\   
	{\bf{R}}^{ijk}&=  \left( \phi_{abc}M^{da} M^{eb} M^{cf} + \frac{1}{2} F_{ab}{}^c M^{da} M^{eb} M_c{}^{f}   + \frac{1}{2} \widetilde{F}^{ab}{}_c  M_{a}{}^d  M_{b}{}^e M^{cf}  \right) v_d{}^{[i}v_e{}^jv_f{}^{k]}    \,.
\end{align*} 
 Here 
\begin{equation}
 	{\bf{\bar{H}}}_{ijk}  = - 3 A_{a[i} A_{b j} v^c{}_{k]} \widetilde{F}^{ab}{}_c +\phi_{abc}  v^a{}_{[i }   v^b{}_{j }   v^c{}_{k] } 
\end{equation}
is the contribution to the H-flux coming from the twisting of the Courant bracket. 
 
Specialising to  $\frak d =\frak g\oplus \frak g$ relevant to the $\lambda$-model we can now go further by using the explicit form of the adjoint action $M_{A}{}^B$,  given in eq.~\eqref{eqn:lambdaadj}.  Doing so we find numerous cancellations to leave 
   \begin{align*}
	{\bf{H}}_{ijk}&=  -  \frac{3}{\sqrt{2}} A_{a[i} A_{b j} v^c{}_{k]} f^{ab}{}_c +\frac{2}{\sqrt{2}}f_{abc}  v^a{}_{[i }   v^b{}_{j }   v^c{}_{k] }   \,,\\
	 {\bf{F}}_{ij}{}^k&= 0 \,, \quad  		{\bf{Q}}^{ij}{}_k = \frac{1}{\sqrt{2}}   v_a{}^{[i}v_b{}^{j]}v^c{}_k    f^{ab}{}_c \,,\quad 
	{\bf{R}}^{ijk} =  0    \,,
\end{align*} 
 in which $\kappa$ is used to raise algebra indices out of position.  It might seem contrary to have $ {\bf Q}$ rather than ${\bf F}$ flux but it  reflects the construction of the geometry as a coset $\mathdsl D/G_{\textrm{diag}} $ and that $G_{\textrm{anti-diag}}$ is not a subgroup of $\mathdsl D$.

  \section{Dilaton in PL and $\lambda$ models  } 
\label{app:dilaton}
Here we show that the constraint on the doubled dilaton $d$ matches the (conventional) dilatons for both PL and $\lambda$-models.

We begin by extracting the metric   for the PL model of eq.~\eqref{eq:PLact} 
\begin{equation}
\begin{aligned}
G &= e^T \left( 1 + E_0^- \Pi\right)^{-1} G_0  \left( 1  -  \Pi  E_0^+\right)^{-1} e  \\  
&=e^T \left( 1 + \tilde{g}_0^{-1}  ( \Pi- \widetilde{B}_0 ) \right)^{-1} \tilde{g}_0^{-1}  \left( 1  -   ( \Pi- \widetilde{B}_0 )\tilde{g}_0^{-1} \right)^{-1} e    \, ,  
\end{aligned} 
\end{equation} 
in which $E_0^\pm = G_0 \pm B_0$ and $ E_0^- = \left( \tilde{g}_0 -\tilde{b}_0  \right)^{-1}$ and $e$ are the components of the left-invariant forms.  It is simple to  take the determinant
\begin{equation}
\log \det G=  2 \log \det e  - \log \det \tilde{g}_0 -2  \log \det \left( 1 + \tilde{g}_0^{-1}  ( \Pi- \widetilde{B}_0 ) \right) 
\end{equation} 
Since $\det e = \det v$ we conclude from  \eqref{eqn:conddilaton} 
\begin{equation}
\phi = \phi_0 - \frac{1}{4}  \log \det \tilde{g}_0 -\frac{1}{2}  \log \det \left( 1 + \tilde{g}_0^{-1}  ( \Pi- \widetilde{B}_0 ) \right) \, , 
\end{equation} 
and choosing $ \phi_0 = \frac{1}{4}  \log \det \tilde{g}_0$, which is of course a constant, gives the result provided in the more mathematically inclined treatment of \cite{Jurco:2017gii}.

For the case of the $\lambda$-model, the expression correctly reduces to the known expression, see e.g. \cite{Sfetsos:2014cea}. This can be readily seen by conveniently expressing the curved metric $g_{ij}$ in terms of right-invariant form on $G$ and the flat metric
\begin{align}
g_{ab}&= \left((1-\lambda D^{-1})^{-1}\kappa + (1-\lambda D)^{-1}\kappa-\kappa \right)_{ab}\notag \\
&= (1-\lambda^2)\left((1-\lambda D^{-1})^{-1}\kappa (1-\lambda D)^{-1}\kappa\right)_{ab}\,.\label{Eq:gab}
\end{align}

Starting from the expression \eqref{eqn:conddilaton}, we indeed obtained the correct expression for the dilaton of the $\lambda$-model:
\begin{align*}
\phi &=\phi_0  + \frac14 \log | \det g_{ij} | - \frac12 \log | \det v^a{}_i | =\phi_0'  + \frac14 \log | \det g_{ab} | =\phi_0''  - \frac12 \log | \det (1-\lambda D^{-1}) |\,,
\end{align*}
where we have used that the adjoint action has unit determinant. The last step is obtained by plugging \eqref{Eq:gab} and using that $D^{-T}=\kappa^{-1} D \kappa$. All constant contribution were successively absorbed into the constant dilaton term.

\section{Details of $\eta$-supergravity solution}\label{sec:appendixetasol} 
In this appendix we detail the full  modified supergravity   solution outlined in section \ref{sec:etasol}.  For the generators $t_a$ of   $\frak{g} =\frak{su}(1,1) \oplus    \frak{su}(2)$ we let  $t_{i}$ be those of   $\frak{su}(1,1)$   and  $t_{\bar i}$ be those of $\frak{su}(2)$    in a basis where the   non-vanishing structure constants are 
given by
\begin{equation}
f_{12}{}^3 = f_{13}{}^2 = f_{32}{}^1=   -1   \, ,  \quad  f_{\bar{1}\bar{2} }{}^{\bar{3} }  =  f_{\bar{2}\bar{3} }{}^{\bar{1} } =   f_{\bar{3}\bar{1} }{}^{\bar{2} }  =  -1 \ . 
\end{equation} 
To raise and lower indices we use the ad-invariant inner-product given by 
\begin{equation}
\kappa_{ij } = \frac{\alpha }{2} f_{i k }{}^l f_{j l}{}^k  = \textrm{diag}(\alpha , \alpha, -  \alpha) \ , \quad \kappa_{\bar i \bar j } = - \frac{\alpha }{2} f_{\bar i  \bar k }{}^{\bar l} f_{\bar j  \bar l}{}^{\bar k}  =\textrm{diag}(\alpha , \alpha, \alpha)\, , 
\end{equation} 
in which we note that the overall normalisation of the  $\frak{su}(2)$ part is of opposite sign to that of the  $\frak{su}(1,1)$.  The solution of the $c^2 = -1$ mCYBE is given by an R-matrix with non-vanishing components 
\begin{equation}
R_{12} =- R_{21} = - \alpha  \ , \quad R_{\bar{1}\bar{2}} = -   R_{\bar{2} \bar{1}} = \alpha \ . 
\end{equation}

We supplement  the six dimensional space corresponding to the deformed $AdS_3 \times S^3$ with a four-torus  (with coordinates $x^\mu$, $\mu =1 \dots 4$) such that the NS data is   \begin{equation}
\begin{aligned}
ds^2 &=  v^a g_{ab } v^b +   dx^\mu dx^\mu  =  v^a \left(  \eta \kappa_{ab}   + \frac{\eta^3}{1+\eta^2}  R_{a}{}^d  R_{db}  \right) v^b +  dx^\mu dx^\mu  \ ,  \\ 
B &=  -\frac{\eta^2}{2(1+\eta^2)} \left( R_{ab }  \, v^a \wedge v^b \right)  = \frac{\eta^2\alpha }{(1+\eta^2)} \left(    \,  v^1 \wedge v^2  - v^{\bar{1}} \wedge v^{\bar{2}}  \right) \ , \\
\Phi &=  \log \left( \frac{\eta^{3/2}}{1 + \eta^2} \right) + \phi_0  \quad H = dB = 0 \ . 
\end{aligned}
\end{equation}
Note that we have chosen to work with the right-invariant Maurer-Cartan forms rather than the left; this removes all coordinate dependance from the metric and fluxes. 
The curvatures that follow from this metric have non-vanishing components 
\begin{equation}
\textrm{Ric}_{i j } = \frac{(1+\eta^2)}{\alpha} \kappa_{ij} - \frac{(1+\eta^2) (3+\eta^2)}{2 \eta \alpha} g_{ij}    \ , \quad  \textrm{Ric}_{\bar{i} \bar{j} } = -\frac{(1+\eta^2)}{\alpha} \kappa_{\bar i \bar j} + \frac{(1+\eta^2) (3+\eta^2)}{2 \eta \alpha} g_{\bar i \bar j}       \ ,
\end{equation}
and are such that the curvature scalar is zero,
\begin{equation}
  R = -\frac12 f_{ac}{}^d f_{bd}{}^c g^{ab} - \frac14 f_{ac}{}^e f_{bd}{}^f g^{ab} g^{cd} g_{ef} = 0\,.
\end{equation}
This is fundamentally due to the choice of opposing normalisations  for the $\frak{su}(2)$ and $\frak{su}(1,1)$ in the inner-product. 

The modified supergravity is defined by the one-form 
\begin{equation}
I  = \frac{1}{2} R^{ab} f_{ab}{}^c g_{cd } v^d  = - \eta \left( v^{\bar{3}} +  v^{3}  \right) \ , 
\end{equation} 
and related \
\begin{equation}
Z =  d\Phi + \iota_I B =  0 \ , \quad X = I  + Z = I  \, . 
\end{equation}

Eq.~\eqref{eqn:spinorads3xs3} encodes the unique solution for the R/R fluxes however this is in six-dimensions.  Here we need to uplift it to 10-dimensions.  In six-dimensions we define 
\begin{align}
  \widehat{G}_{6d}^{(1)} = -\frac{1+\eta^2}{\sqrt{2}}  R^{ab} f_{abc} v^c \nonumber \ , \quad  \widehat{G}_{6d }^{(3)} = \frac{1 + \eta^2}{3 \sqrt{2}} f_{abc} v^a \wedge v^b \wedge v^c\, ,
\end{align} 
and the six-dimensional poly-form 
\begin{align}
   \widehat{{\cal F}}_{6d }    = e^\Phi  e^{-B} \left(  \widehat{G}_{6d}^{(1)} +  \widehat{G}_{6d}^{(3)}\right) ,
\end{align}
which by construction has vanishing Lie derivative along ${\cal I}$ i.e. $L_{\cal I} \widehat{{\cal F}}_{6d }   =0   $. 
The components of this obey 
\begin{align}
  \iota_{\mathcal{I}} \widehat{\mathcal{F}}_{6d }^{(1)} &= 0\,, &
  d \widehat{\mathcal{F}}_{6d }^{(1)} &= \iota_{\mathcal{I}} \widehat{\mathcal{F}}_{6d }^{(3)} \,, &
  d \widehat{\mathcal{F}}_{6d }^{(3)} &= \iota_{\mathcal{I}} \widehat{\mathcal{F}}_{6d }^{(5)} \,,
  & &
  \star_6 \widehat{\mathcal{F}}{}^{(1)}_{6d } &= - \widehat{\mathcal{F}}{}^{(5)}_{6d }\,, &
 ~~ \star_6 \widehat{\mathcal{F}}{}^{(3)}_{6d } &= \widehat{\mathcal{F}}{}^{(3)}_{6d }\,. 
\end{align}
From this we can build a ten-dimensional R/R poly-form
\begin{align}
   {\cal F}_{10d }    =  \mu  \widehat{{\cal F}}_{6d }  \wedge \left( 1  + \omega  - \textrm{vol}_4 \right)  \ , \quad (d+ H\wedge - Z\wedge - \iota_I )   {\cal F}_{10d }  = 0 
\end{align}
in which $\mu$ is a normalisation to be fixed, and  $\omega = \vec{\textrm{n}}\cdot \vec{\omega} $ is expanded in the basis of  self-dual three forms on $T^4$ and $ \textrm{vol}_4  = dx^1 \wedge dx^2 \wedge dx^3 \wedge d  x^4$.  With this in mind it is quite easy now to verify the modified supergravity equations of eq.~\eqref{eq:mSUGRA} are satisfied providing the normalisation of the R/R sector are set such that  
\begin{equation}
\mu = e^{-\phi_0 } \left( 2 \eta  ( 1 +\vec{\textrm{n} }^2 ) \right)^{- \frac{1}{2}} \ .
\end{equation} 
  Explicitly we have 
  \begin{equation}
\begin{aligned}
   {\cal F}_{10d }^{(1)} &=   \rho \left( v^3 + v^{\bar{3}}  \right) \ , \quad 
    {\cal F}_{10d }^{(3)} =   \rho \left( v^3 + v^{\bar{3}}  \right) \wedge \omega + \frac{\alpha \rho }{(1+\eta^2)  }\left( v^{123} -   v^{\bar 1 \bar 2 \bar 3 }  -  \eta^2 v^{12 \bar{3} }   +   \eta^2v^{\bar 1 \bar 2   3 }   \right)   \\ 
     {\cal F}_{10d }^{(5)} &= -   \rho  \left( v^3 + v^{\bar{3}}  \right)\wedge \textrm{vol}_4  +     \frac{\alpha \rho }{(1+\eta^2)  } \left( v^{123} -   v^{\bar 1 \bar 2 \bar 3 }  -  \eta^2 v^{12 \bar{3} }   +   \eta^2v^{\bar 1 \bar 2   3 }   \right) \wedge \omega  + \frac{\alpha^2 \eta^2 \rho }{(1+\eta^2)^2   } \left(v^{123 \bar 1 \bar 2} + v^{12 \bar 1 \bar 2 \bar 3}    \right) \ ,   
\end{aligned}
\end{equation}
in which we let $\rho   = \eta  ( 1 +\vec{\textrm{n} }^2 )^{-\frac{1}{2}} $.

\section{Drinfel'd Doubles and Group Parameterisations}\label{app:groupparam}

\subsection{$\frak{su}(2) \oplus \frak{e}_3$} 
We work with the following basis of generators 
\begin{equation}\label{eq:su2e3gens}
\begin{aligned}
T_1 &= \frac{i}{2} \mathbb{I}\otimes \sigma_1 &
T_2 &= \frac{i}{2} \mathbb{I}\otimes \sigma_2 &
T_3 &= \frac{i}{2} \mathbb{I}\otimes \sigma_3 \\
\widetilde T^1 &= -\frac{1}{2} \sigma_3 \otimes \sigma_1 - \frac{i}{2} \mathbb{I}\otimes \sigma_2  &
\widetilde T^2 &= -\frac{1}{2} \sigma_3 \otimes \sigma_2 +  \frac{i}{2} \mathbb{I}\otimes \sigma_1 &
\widetilde T^3 &= -\frac{1}{2} \sigma_3 \otimes \sigma_3     \, . 
\end{aligned}
\end{equation}
Defining projectors
\begin{equation}
  P^\pm = \frac{1}{2} ( \mathbb{I}  \pm  \sigma_3) \otimes \mathbb{I}  \, , 
\end{equation}
allows us to realise the inner-product as 
\begin{equation}\label{eqn:pairingsl2c}
  \llangle\mathdsl T_A  ,\mathdsl T_B \rrangle = i \mathrm{Tr}\left( P^+ \mathdsl T_A P^+\mathdsl T_B - P^-\mathdsl T_A P^-\mathdsl T_B \right) = \eta_{AB}\, . 
\end{equation} 

The ${\cal{R}}$-matrix that gives this a bialgebra structure is   
\begin{equation}
{\cal R} = \left(  \begin{array}{ccc}
 0 & 1 & 0 \\
 -1 & 0 & 0 \\
 0 & 0 & 0 \\
\end{array}
\right) \,,
 \end{equation} 
and we parametrise an $SU(2)$ element as  
\begin{equation}
  G = \begin{pmatrix} ~g~ & ~0~ \\ ~0~ & ~g~ \end{pmatrix} \ , \quad 
  g = \begin{pmatrix} e^{-i x} \sqrt{1-r^2} & -e^{-i \phi} r\\
    e^{i\phi} r & e^{i x} \sqrt{1-r^2} \end{pmatrix}\,, 
\end{equation}
such that the left-invariant forms  defined by $G^{-1}dG = e^a T_a$ read
\begin{equation} 
\frac{1}{2}\left( e^1 \pm i e^2 \right)  =  \frac{e^{\mp i (x- \phi)}}{\sqrt{1-r^2}} \left( r(r^2-1)  d(x+\phi) \pm i dr  \right) \ , \quad \frac{e^3}{2} = (r^2-1) dx + r^2 d\phi  \,,
\end{equation} 
The metric on $S^3$ is obtained as 
\begin{equation}
  \frac{-1}{4} \mathrm{Tr} \left(G^{-1} dG G^{-1} dG \right) =   (1- r^2  ) dx^2 + \frac{dr^2}{1- r^2} + r^2 d\phi^2 \,,
\end{equation}
which is rendered more familiar with $r =\sin\theta$. Finally, we need the combination of adjoint actions that enter into the PL sigma models:
 \begin{equation}
 \Pi =  
\left(
\begin{array}{ccc}
 0 & 2 r^2 & -2 r \sqrt{1-r^2} \sin (x-\phi ) \\
 -2 r^2 & 0 & -2 r \sqrt{1-r^2} \cos (x-\phi ) \\
 2 r \sqrt{1-r^2} \sin (x-\phi ) & 2 r \sqrt{1-r^2} \cos (x-\phi ) & 0 \\
\end{array}
\right) \ . 
 \end{equation}

\subsection{$\frak{su}(1,1) \oplus \frak{e}_3 $} 
We work with the following basis of generators 
\begin{equation}
\begin{aligned} 
  T_1 &= \frac{1}{2} \sigma_3 \otimes \sigma_2 &
  T_2 &= \frac{1}{2} \sigma_3 \otimes \sigma_1  &
  T_3 &= \frac{i}{2} \mathbb{I}\otimes \sigma_3 \\
\widetilde  T^1 &= -\frac{1}{2} \sigma_3 \otimes \sigma_1 - \frac{i}{2} \mathbb{I}\otimes \sigma_2 &
\widetilde  T^2 &= \frac{1}{2} \sigma_3 \otimes \sigma_2  -  \frac{i}{2} \mathbb{I}\otimes \sigma_1 &
\widetilde  T^3 &= -\frac{1}{2} \sigma_3 \otimes \sigma_3 \, . 
\end{aligned} 
\end{equation}
and realise the inner-product again as \eqref{eqn:pairingsl2c}.

Recall that in its defining representation $SU(1,1)$ consists of complex matrices of unit determinant that satisfy
\begin{equation}
g^\dag \omega g = \omega \ , \quad \omega = \begin{pmatrix} 1 &0 \\ 0 &-1   \end{pmatrix} \, ,  
\end{equation}
and such a group element can be parameterized as
\begin{equation}
g= \left(
\begin{array}{cc}
 e^{-i t} \sqrt{\rho ^2+1} & e^{-i \psi } \rho  \\
 e^{i \psi } \rho  & e^{i t} \sqrt{\rho ^2+1} \\
\end{array}
\right) \, .
\end{equation}
 In the $4\times 4$ representation used for the Drinfel'd double double we have 
\begin{equation}
G= \begin{pmatrix} ~g~ & ~0~ \\ ~0~ & ~ \textrm{ad}_\omega g  ~ \end{pmatrix}   \, . 
\end{equation} 

The left-invariant one forms are given as 
\begin{equation} 
\frac{1}{2}\left( e^1 \pm i e^2 \right)  =  \frac{e^{\pm i (t- \psi)}}{\sqrt{1+\rho^2}} \left( \rho(1 + \rho^2)  d(t+\psi) \pm i d \rho  \right) \ , \quad \frac{e^3}{2} =-  (1+ \rho^2) dt - \rho^2 d\psi  \, . 
\end{equation} 
 The combination of adjoint actions that enter into the PL sigma models is
 \begin{equation}
 \Pi =  \left(
\begin{array}{ccc}
 0 & -2 \rho ^2 & 2 \rho  \sqrt{\rho ^2+1} \sin (t-\psi ) \\
 2 \rho ^2 & 0 & -2 \rho  \sqrt{\rho ^2+1} \cos (t-\psi ) \\
 -2 \rho  \sqrt{\rho ^2+1} \sin (t-\psi ) & 2 \rho  \sqrt{\rho ^2+1} \cos (t-\psi ) & 0
   \\
\end{array}
\right)  \ . 
 \end{equation}

The metric on $AdS_3$ is obtained as 
\begin{equation}
\begin{aligned}
\frac{1}{4} \mathrm{Tr} \left(G^{-1} dG G^{-1} dG \right) &=&   - (1+\rho^2) dt^2 + \frac{d\rho^2}{1+\rho^2} + \rho^2 d\psi^2 \\ 
&=&  - \cosh^2 \sigma dt^2 + d\sigma^2 + \sinh^2 \sigma d\psi^2 \, ,
\end{aligned} 
\end{equation}
with $\rho = \sinh \sigma$. This   follows from the embedding 
\begin{equation}
-1 = - X_0^2 +X_1^2 +X_2^2 - X_3^2 
\end{equation}
with 
\begin{equation}
X_0 + i X_3 = e^{i t} \cosh \sigma \ , \quad X_1 + i X_2 = e^{-i \psi} \sinh \sigma \ .
\end{equation}
Although we shall not directly need it we note for completeness the isomorphism to $SL(2, \mathbb{R})$ is made by defining 
\begin{equation}
 g_{SL(2)} = \begin{pmatrix} X_0 + X_1 & X_2 + X_3 \\ X_2 - X_3 & X_0 -X_1   \end{pmatrix} \ . 
\end{equation}

 \bibliography{literatur}
\bibliographystyle{JHEP}
\end{document}